\def\submissiontype{0}			%
\newif\ifShowComments			%
\ShowCommentsfalse
\newif\ifEprint					%
\newif\ifSubmission				%
\newif\ifCameraReady			%
\newif\ifSupplementaryMaterial	%
\SupplementaryMaterialfalse

\newif\ifTightOnSpace			%
\TightOnSpacefalse

\ifnum\submissiontype=0     											%
	\Eprinttrue
	\Submissionfalse
	\CameraReadyfalse
\else																	%
	\Eprintfalse								
	
	\ifnum\submissiontype=1     										%
		\Submissiontrue
		\CameraReadyfalse
		
		\SupplementaryMaterialtrue										%
	
	\else     															%
		\CameraReadytrue
		\Submissionfalse
	\fi 
\fi

\RequirePackage{amsmath}  %

\ifEprint
\documentclass[a4paper,runningheads]{llncs}
	\usepackage[all]{nowidow}
\else
	\documentclass{llncs} %
\fi

\usepackage[utf8]{inputenc}
\usepackage[T1]{fontenc}
\usepackage{lmodern}

\usepackage{float} 

\usepackage{tikz}
\usetikzlibrary{positioning}
\usetikzlibrary{calc}
\usetikzlibrary{shapes.geometric}
\usetikzlibrary{matrix,arrows}

\usepackage{nicefrac}
\usepackage{nicodemus}
\usepackage{stmaryrd}

\usepackage{amsfonts,amsmath,amscd,amssymb,array}
\usepackage{algorithm2e,algpseudocode}
\usepackage{graphicx}
\usepackage{xspace}
\usepackage{ifthen}
\usepackage{braket}
\usepackage{dsfont}
\usepackage{breakcites}

\usepackage{thmtools, thm-restate}

\usepackage{csquotes}

\usepackage{xcolor}
\definecolor{linkcolor}{rgb}{0.65,0,0}
\definecolor{citecolor}{rgb}{0,0.65,0}
\definecolor{urlcolor}{rgb}{0,0,0.65}
\usepackage[colorlinks=true, linkcolor=linkcolor, urlcolor=urlcolor, 
citecolor=citecolor]{hyperref}
\usepackage[capitalise]{cleveref}

\ifSupplementaryMaterial
	\ifTightOnSpace
		\crefname{appendix}{Suppl. Mat.}{Supplementary Material} 
	\else
		\crefname{appendix}{Supplementary Material}{Supplementary Material} 
	\fi
\fi

\ifTightOnSpace
\crefname{theorem}{Thm.}{Thms.}
\crefname{lemma}{Lem.}{Lems.}
\crefname{corollary}{Cor.}{Cors.}
\crefname{section}{Sect.}{Sect.}
\fi

\usepackage[framemethod=default]{mdframed}

\newcommand{\eps}{\varepsilon}

\usepackage{mathtools}

\usepackage{ctable} %
\newcommand{\A}{\ensuremath{\Ad{A}}\xspace}

\newcommand{\id}{\ensuremath{\textrm{id}}\xspace}

\newcommand{\todo}[1]{\ifShowComments{\textcolor{red}{TODO: #1}}\fi}

\def\subheading#1{\medskip\noindent{\boldmath\textbf{#1}}~\ignorespaces}

\definecolor{dgreen}{rgb}{.1,.5,.1}
\newcommand{\cm}[1]{\ifShowComments{\textcolor{dgreen}{[CM: #1]}}\fi}

\newcommand{\tinka}[1]{\ifShowComments{\textcolor{cyan}{[Tinka: #1]}}\fi}

\newcommand{\proj}[1]{\ket{#1}\!\!\bra{#1}}
\newcommand{\Tr}{\mathrm{Tr}}

\DeclareMathOperator*{\E}{\mathbb E}

\newcommand\bits{\ensuremath{\{0,1\}}\xspace}					%
\newcommand{\uni}{\ensuremath{\leftarrow_\$}\xspace}			%
\newcommand{\bool}[1]{\ensuremath{\llbracket #1\rrbracket}\xspace}	\DeclareMathOperator\supp{supp}									%
\DeclareMathOperator*{\Exp}{\E}
\newcommand{\Adv}{\ensuremath{\mathrm{Adv}}\xspace}				%
\newcommand{\Ad}[1]{\ensuremath{\mathcal{#1}}\xspace}			%
\newcommand{\Time}{\textnormal{Time}}							%
\newcommand{\QMem}{\textnormal{QMem}}							%
\newcommand{\state}{\textnormal{st}}
\newcommand{\inputVar}{\ensuremath{inp}\xspace}

\newcommand{\mathsc}[1]{\text{\textsc{#1}}}
\newcommand{\heading}[1]{{\vspace{1ex}\noindent\sc{#1}}}

\mathchardef\ordinarycolon\mathcode`\:				%
\mathcode`\:=\string"8000							%
\begingroup \catcode`\:=\active						%
\gdef:{\mathrel{\mathop\ordinarycolon}}				%
\endgroup

\newcommand{\FIND}{\ensuremath{\mathsc{FIND} }\xspace}
\newcommand{\FINDshort}{\ensuremath{\mathsc{F} }\xspace}
\newcommand{\EVENTEXT}{\ensuremath{\mathsc{EXT} }\xspace}
\newcommand{\EVENTEXTshort}{\ensuremath{\mathsc{EXT} }\xspace}
\newcommand{\EVENT}{\ensuremath{\mathsc{E} }\xspace}

\newcommand{\RO}[1]{\ensuremath{\mathsf{{#1}}}\xspace}				%

\newcommand{\pcif}{\textbf{if }}
\newcommand{\pcelse}{\textbf{else }}
\newcommand{\pcand}{\textbf{and }}
\newcommand{\pcor}{\textbf{or }}
\newcommand{\pcreturn}{\textbf{return }}
\newcommand{\gcom}[1]{\hfill $\sslash$#1}				%

\newcommand{\before}[1]{\the\numexpr\value{#1}-1\relax}	%
\newcommand{\gameDist}[2]{|\Pr[G_{\before{#1}}^{\Ad{#2}}=1] - \Pr[G_{\the\numexpr\value{#1}\relax}^{\Ad{#2}}=1]|}		%
\newcommand{\game}[1]{Game G\ensuremath{_{#1}}\xspace}
\newcommand{\gameseq}[2]{Games\ \ensuremath{G_{#1}-G_{#2}}\xspace}
\newcommand{\bfgame}[1]{\ensuremath{\mathbf{Game\ G_{#1}}}\xspace}
\newcommand{\bfgameseq}[2]{\ensuremath{\mathbf{Games\ G_{#1}-G_{#2}}}\xspace}

\newcommand{\RightarrowROM}{\ensuremath{\stackrel{\mathrm{\scalebox{.6}{ROM}}}{\Rightarrow}}\xspace}%

\newcommand{\Measure}{\mathsc{Measure}}												%
\newcommand{\RightarrowaugQROM}[1]{\ensuremath{\stackrel{%
		\mathrm{\scalebox{.6}{\augQROM{#1}}}}{\Rightarrow}}\xspace}					%

\newcommand{\GDPB}{\ensuremath{\mathsf{GDPB}}\xspace}

\newcommand{\augQROM}[1]{\ensuremath{\mathrm{eQROM}_{#1}}\xspace}
\newcommand{\augQRO}[1]{\ensuremath{\mathrm{eQRO}_{#1}}\xspace}
\newcommand{\augmentedORextended}{extended\xspace}
\def\SupOr{\ensuremath{\mathsf{eCO}}\xspace}
\def\SE{\ensuremath{\SupOr.\mathsf{E}}\xspace}
\def\SRO{\ensuremath{\SupOr.\mathsf{RO}}\xspace}
\def\OrUnit{\ensuremath{O_{XYD}}\xspace}

\newcommand{\orSemiClassical}[1]{\ensuremath{O^{\textsf{SC}}_{#1}}\xspace}					%

\def\OWTH{\ensuremath{\mathsf{OWTH}}\xspace}												%
\def\Extractor{\ensuremath{\mathsf{ExtractSet}}\xspace}										%
\def\qExtract{\ensuremath{ {q_\mathrm{E} } }\xspace}										%
\def\GenInput{\ensuremath{\mathsf{GenInp}}\xspace}

\newcommand{\puncture}[2]{\ensuremath{#1{\setminus}#2}\xspace}								%
\newcommand{\reproSet}{\ensuremath{\mathcal{S}}\xspace}										%
\newcommand{\extractionSet}{\ensuremath{\mathcal{S'}}\xspace}								%

\newcommand{\commitmentSet}{\ensuremath{\mathcal{T}}\xspace}								%

\def\SupOrNotProg{\ensuremath{\SupOr^0}\xspace}												%
\def\SRONotProg{\ensuremath{\SupOr.\!\!{}^0\mathsf{RO}}\xspace}
\def\SENotProg{\ensuremath{\SupOr.\!\!{}^0\mathsf{E}}\xspace}

\def\SupOrProg{\ensuremath{\SupOr^1}\xspace}												%

\def\SEProg{\ensuremath{\SupOr.\!\!{}^1\mathsf{E}}\xspace}

\def\logRegForFind{\ensuremath{ {L_{F} } }\xspace}											%
\newcommand{\unitaryForFind}[1]{\ensuremath{ {U_S^{#1}}  }\xspace}									%

\def\logRegForExtract{\ensuremath{ {L_{E} } }\xspace}										%
\newcommand{\unitaryForExtract}[1]{\ensuremath{ {U_{f}^{#1} } }\xspace}						%
\newcommand{\unitaryForExtractX}[2]{\ensuremath{ {U_{f, #2}^{#1} } }\xspace}				%

\def\MeasureLogs{\ensuremath{\mathcal{E}_{\logRegForFind, \logRegForExtract}}\xspace}		%

\def\FinalStateAD{\ensuremath{\rho_0'}\xspace}
\def\FinalStateNotPuncWithExtract{\ensuremath{\rho_0''}\xspace} 							%
\def\FinalStateNotPuncWithoutExtract{\ensuremath{\rho_0}\xspace}							%
\def\FinalStatePuncWithExtract{\ensuremath{\rho_1''}\xspace}								%
\def\FinalStatePuncWithoutExtract{\ensuremath{\ensuremath{\rho_1}}\xspace}					%

\def\instanceOWTH{\ensuremath{ins}\xspace}
\def\GenInstance{\ensuremath{\mathsf{GenInst}}\xspace}

\def\FinalStatePureNotPunc{\ensuremath{\phi_0}\xspace}
\def\FinalStatePureNotPuncIndexed{\ensuremath{\FinalStatePureNotPunc^{\instanceOWTH}}\xspace}
\def\FinalStatePureKeepTrack{\ensuremath{\phi_1}\xspace}
\def\FinalStatePureKeepTrackIndexed{\ensuremath{\FinalStatePureKeepTrack^{\instanceOWTH}}\xspace}

\def\InitialState{\ensuremath{\ket{\Phi^{(0)}}\xspace}}
\def\InbetweenState{\ensuremath{\Phi}\xspace}
\newcommand{\InbetweenStateNotPunc}[1]{\ensuremath{\ket{\InbetweenState_0^{#1}}}\xspace}
\newcommand{\InbetweenStateKeepTrack}[1]{\ensuremath{\ket{\InbetweenState_1^{#1}}}\xspace}

\newcommand{\bracket}[2]{\left\langle#1|#2\right\rangle}

\newcommand{\CNOT}{\mathsf{CNOT}}

\newcommand{\EventDecapsSimDiffers}{\ensuremath{ \mathsf{DIFF}}\xspace }
\newcommand{\EventGuessedCT}{\ensuremath{ \mathsf{GUESS}}\xspace }
\newcommand{\ListQueriesG}{\ensuremath{ \mathcal{L}_{\RO{G}}}\xspace }
\newcommand{\ListFail}{\ensuremath{ \mathcal{L}_{\mathsc{FAIL}}}\xspace }
\newcommand{\numberDecQueries}{\ensuremath{ q_{\mathsf{D}}}\xspace }
\newcommand{\numberDecapsQueries}{\ensuremath{ q_{\mathsf{D}}}\xspace }
\newcommand{\numberDecOrDecapsQueries}{\ensuremath{ q_{D}}\xspace }

\newcommand{\PKE}{\ensuremath{\mathsf{PKE}}\xspace}
\newcommand{\KG}{\ensuremath{\mathsf{KG}}\xspace}
\newcommand{\pk}{\ensuremath{\mathit{pk}}\xspace}
\newcommand{\sk}{\ensuremath{\mathit{sk}}\xspace}
\newcommand{\Encrypt}{\ensuremath{\mathsf{Enc}}\xspace}
\newcommand{\Decrypt}{\ensuremath{\mathsf{Dec}}\xspace}

\newcommand{\MSpace}{\ensuremath{\mathcal{M}}\xspace}
\newcommand{\RSpace}{\ensuremath{\mathcal{R}}\xspace}
\newcommand{\CSpace}{\ensuremath{\mathcal{C}}\xspace}
\newcommand{\KeySpace}{\ensuremath{\mathcal{K}}\xspace}

\newcommand{\Tone}{\ensuremath{\mathsf{T}}\xspace}
\newcommand{\notationDerand}[1]{{#1}^{\RO{G}}}
\newcommand{\PKEDerand}{\ensuremath{ \notationDerand{\PKE}}\xspace}
\newcommand{\EncryptDerand}{\ensuremath{\notationDerand{\Encrypt}}\xspace}
\newcommand{\DecryptDerand}{\ensuremath{\notationDerand{\Decrypt}}\xspace}

\newcommand{\KEM}{\ensuremath{\mathsf{KEM}}\xspace}
\newcommand{\KemGen}{\ensuremath{\KG}\xspace}
\newcommand{\Encaps}{\ensuremath{\mathsf{Encaps}}\xspace}
\newcommand{\Decaps}{\ensuremath{\mathsf{Decaps}}\xspace}

\newcommand{\Ttwo}{\ensuremath{\mathsf{U}}\xspace}
\newcommand{\FO}{\ensuremath{\mathsf{FO}}\xspace}
\newcommand{\explicitReject}{{\bot}}										%
\newcommand{\hashOnlyMessage}{{\mathit{m}}}									%
\newcommand{\TtwoExplicit}{\ensuremath{\Ttwo^\explicitReject}\xspace}		%
\newcommand{\FOExplicit}{\ensuremath{\FO^\explicitReject}\xspace}			%
\newcommand{\KemExplicit}{\ensuremath{\mathsf{KEM}^\explicitReject}\xspace}	%
\newcommand{\TtwoExplicitMess}{\ensuremath{\TtwoExplicit_\hashOnlyMessage}\xspace}	%
\newcommand{\KemExplicitMess}{\ensuremath{\KemExplicit_\hashOnlyMessage}\xspace}	%

\newcommand{\implicitReject}{{\not\bot}} 									%
\newcommand{\FOImplicit}{\ensuremath{\FO^\implicitReject}\xspace}			%
\newcommand{\KemImplicitMess}{\ensuremath{\mathsf{KEM}^\implicitReject_\hashOnlyMessage}\xspace}	%
\newcommand{\atk}{\mathsf{ATK}}
\newcommand{\IND}{\ensuremath{\mathsf{IND}}}
\newcommand{\OW}{\ensuremath{\mathsf{OW}}\xspace}
\newcommand{\CPA}{\ensuremath{\mathsf{CPA}}\xspace}
\newcommand{\CCA}{\ensuremath{\mathsf{CCA}}\xspace}
\newcommand{\OWCPA}{\ensuremath{\OW\text{-}\CPA}\xspace}
\newcommand{\OWVA}{\ensuremath{\OW\text{-}\mathsf{VA}}\xspace}
\newcommand{\INDCPA}{\ensuremath{\IND\text{-}\CPA}\xspace}
\newcommand{\INDCCA}{\ensuremath{\IND\text{-}\CCA}\xspace}
\newcommand{\INDCCAKEM}{\ensuremath{\INDCCA\text{-}\KEM}\xspace}
\newcommand{\INDCPAKEM}{\ensuremath{\INDCPA\text{-}\KEM}\xspace}
\newcommand{\INDATKKEM}{\ensuremath{\IND\text{-}\atk\text{-}\KEM}\xspace}
\newcommand{\INDATK}{\INDATKKEM}
\newcommand{\oracleDecaps}{\ensuremath{\mathsc{oDecaps}}\xspace}
\newcommand{\oracleDecapsSim}{\ensuremath{\oracleDecaps'}\xspace}
\newcommand{\oracleDecapsSimFail}{\ensuremath{\oracleDecaps''}\xspace}

\newcommand{\deltaWorstCase}{{\ensuremath{\delta_{\mathrm{wc}}}}\xspace}
\newcommand{\deltaIndKey}{\ensuremath{{\delta_{\mathrm{ik}}}}\xspace}
\newcommand{\deltaRandKey}{\deltaIndKey}
\newcommand{\FFP}{\ensuremath{\mathsf{FFP}}\xspace}
\newcommand{\FFPCPA}{\ensuremath{\FFP\text{-}\mathsf{CPA}}\xspace}
\newcommand{\FFPNoKey}{\ensuremath{\mathsf{FFP}\text{-}\mathsf{NK}}\xspace}
\newcommand{\FFPCCA}{\ensuremath{\FFP\text{-}\mathsf{CCA}}\xspace}
\newcommand{\FFPATK}{\ensuremath{\FFP\text{-}\atk}\xspace}

\newcommand{\COR}{\ensuremath{\mathsf{COR}}\xspace}

\newcommand{\NonGeneric}{\ensuremath{\mathsf{NG}}\xspace}
\newcommand{\FFPNG}{\ensuremath{\FFP\text{-}\NonGeneric}\xspace}
\newcommand{\FngFPCPA}{\FFPNG}

\newcommand{\oracleDecrypt}{\ensuremath{\mathsc{oDecrypt}}\xspace}
\newcommand{\oracleDecryptSim}{\ensuremath{\oracleDecrypt'}\xspace}
\newcommand{\FCO}{\ensuremath{\mathsf{FCO}}\xspace}

\newcommand{\FrodoPKE}{\ensuremath{\mathsf{FrodoPKE}}\xspace}
\newcommand{\FrodoKG}{\ensuremath{\FrodoPKE.\KG}\xspace}
\newcommand{\FrodoMSpace}{\ensuremath{\FrodoPKE.\MSpace}\xspace}
\newcommand{\FrodoCSpace}{\ensuremath{\FrodoPKE.\CSpace}\xspace}
\newcommand{\FrodoEnc}{\ensuremath{{\FrodoPKE.\Encrypt}}\xspace}
\newcommand{\FrodoEncode}{\ensuremath{{\mathsf{Frodo}.\mathsf{Encode}}}\xspace}
\newcommand{\HQCPKE}{\ensuremath{\mathsf{HQC.PKE}}\xspace}
\newcommand{\HQCKG}{\ensuremath{\HQCPKE.\KG}\xspace}
\newcommand{\HQCMSpace}{\ensuremath{\HQCPKE.\MSpace}\xspace}
\newcommand{\HQCCSpace}{\ensuremath{\HQCPKE.\CSpace}\xspace}
\newcommand{\HQCEnc}{\ensuremath{{\HQCPKE.\Encrypt}}\xspace}
\newcommand{\seedA}{\ensuremath{\mathit{seed}_A}\xspace}
\newcommand{\instance}{\ensuremath{\mathit{i}}\xspace}

\newcommand{\hashBoth}{{\mathit{m,c}}}													%
\newcommand{\FOBothGeneral}{\ensuremath{\FO_\hashBoth}\xspace}							%
\newcommand{\FOExplicitBoth}{\ensuremath{\FO^\explicitReject_\hashBoth}\xspace}			%
\newcommand{\FOImplicitBoth}{\ensuremath{\FO^\implicitReject_\hashBoth}\xspace}			%

\newcommand{\FOMessGeneral}{\ensuremath{\FO_\hashOnlyMessage}\xspace}					%
\newcommand{\FOExplicitMess}{\ensuremath{\FO^\explicitReject_\hashOnlyMessage}\xspace}	%
\newcommand{\FOImplicitMess}{\ensuremath{\FO^\implicitReject_\hashOnlyMessage}\xspace}	%

\newcommand{\FOquantumExplicit}{\ensuremath{\mathsf{QFO}^\explicitReject_\hashOnlyMessage}\xspace}		%
\newcommand{\FOquantumImplicit}{\ensuremath{\mathsf{QFO}^\implicitReject_\hashOnlyMessage}\xspace}		%

\newenvironment{nicodemus}[1][\thenicolinenr]{%
	\begin{enumerate}[
		topsep=0ex,
		label=\nicolinenrformat\PaddingUp*,
		ref=\nicorefprefix\PaddingUp*,
		align=right,
		leftmargin=0em,
		itemindent=!,
		labelindent=0em,
		labelwidth=\nicolinenrwidth,
		labelsep=\nicolinenrsep,
		listparindent=\parindent,
		noitemsep,
		]%
		\setcounter{enumi}{#1}%
		\addtocounter{enumi}{-1}%
	}{%
	\end{enumerate}%
	\addtocounter{enumi}{1}%
	\setcounter{nicolinenr}{\theenumi}%
}

\titlerunning{FO and decryption failures}
\title{Failing gracefully: Decryption failures and the Fujisaki-Okamoto transform%
}

\ifSubmission
	\author{\vspace{-0.5in}}
	\institute{}
\else
	\author{
	  Kathrin Hövelmanns\inst{1}
	\and
	  Andreas Hülsing\inst{1}
	\and
	  Christian Majenz\inst{2}
	}
	\institute{
		Eindhoven University of Technology,
		The Netherlands\\
		\and
	Department of Applied Mathematics and Computer Science, Technical University of Denmark \\
		\email{authors-fo-failure@huelsing.net} %
	}
\authorrunning{K. Hövelmanns, A. Hülsing, C. Majenz}
\fi

\def\acknowledgmenttext{
	 A.H. was funded by an NWO VIDI grant (Project No. VI.Vidi.193.066).
	 C.M. was funded by a NWO VENI grant (Project No. VI.Veni.192.159).
}

\begin{document}

\maketitle
	
\begin{abstract}
In known security reductions for the Fujisaki-Okamoto transformation, decryption failures are handled via a reduction solving the rather unnatural task of finding failing plaintexts \emph{given the private key}, resulting in a Grover search bound. Moreover, they require an implicit rejection mechanism for invalid ciphertexts to achieve a reasonable security bound in the QROM. We present a reduction that has neither of these deficiencies: %
We introduce two %
security games related to finding decryption failures, one capturing the \emph{computationally hard} task of \emph{using the public key} to find a decryption failure, and one capturing the \emph{statistically hard} task of searching the random oracle for \emph{key-independent} failures like, e.g., large randomness.
As a result, our security bounds in the QROM are tighter than previous ones with respect to the generic random oracle search attacks: The attacker can only partially compute the search predicate, namely for said key-independent failures. %
In addition, our entire reduction works for the explicit-reject variant of the %
transformation and improves significantly over all of its known reductions. Besides being the more natural variant of the %
transformation, security of the explicit reject mechanism is also relevant for side channel attack resilience of the implicit-rejection variant.
Along the way, we prove several technical results characterizing preimage extraction and certain search tasks  in the QROM that might be of independent interest.
\\[6pt]
  \textbf{Keywords:} Public-key encryption, post-quantum security, QROM, Fujisaki-Okamoto transformation, decryption failures, NIST
\end{abstract}

\ifSubmission \else
	\begingroup
	\makeatletter
	\def\@thefnmark{} \@footnotetext{\relax \acknowledgmenttext
	Date: \today}
	\endgroup
\fi

	\setcounter{tocdepth}{2}
	\tableofcontents
\section{Introduction}\label{sec:Intro}

The Fujisaki-Okamoto (FO) transform~\cite{C:FujOka99,JC:FujOka13} is a well known transformation that combines a weakly secure public-key encryption scheme and a weakly secure secret-key encryption scheme into an \INDCCA secure public-key encryption scheme in the random oracle model. Dent~\cite[Table 5]{IMA:Dent03} gave an adoption for the setting of key-encapsulation.
This adoption for key encapsulation mechanisms (KEM) is now the de-facto standard to build secure KEMs.
In particular, it was used in virtually all KEM submissions to the NIST PQC standardisation  process~\cite{NIST:Competition}.
In the context of post-quantum security, however, two novel issues surfaced:
First, many of the PKE schemes being transformed into KEM are not perfectly correct, i.e., they sometimes fail to decrypt a ciphertext to its plaintext.
Second, security proofs have to be done in the quantum-accessible random oracle model (QROM) to be applicable to quantum attackers. 

Both problems were tackled in~\cite{TCC:HofHovKil17} and a long sequence of follow-up works (among others \cite{EC:SaiXagYam18,C:JZCWM18,TCC:BHHHP19,PKC:HKSU20, EC:KSSSS20}).
While these works made great progress towards achieving tighter reductions in the QROM, the treatment of decryption failures did not improve significantly.
In this work, we make significant progress on the treatment of decryption failures.
Along the way, we obtain several additional results relevant on their own.

An additional quirk of existing QROM reductions for the FO transform is that they require an \emph{implicit rejection} variant, where pseudorandom session keys are returned instead of reporting decapsulation errors%
, to avoid extreme reduction losses. (The only known concrete bound~\cite{DFMS21} for Dent's variant is much weaker then those known for the implicit rejection variant.)

\subheading{The Fujisaki-Okamoto transformation.} %
We recall the FO transformation for KEM as introduced in \cite[Table 5]{IMA:Dent03} and revisited by \cite{TCC:HofHovKil17}, there called \FOExplicitMess.
\FOExplicitMess constructs a KEM from a public-key encryption scheme \PKE, and the overall transformation \FOExplicitMess can be described by first modifying \PKE to obtain a deterministic scheme \PKEDerand, and then applying a PKE-to-KEM transformation (called \TtwoExplicitMess in \cite{TCC:HofHovKil17}) to \PKEDerand:

{\heading{Modified scheme \PKEDerand.}
	Starting from \PKE and a hash function \RO{G}, deterministic encryption scheme \PKEDerand is built by letting \EncryptDerand encrypt messages $m$ according to the encryption algorithm \Encrypt of \PKE, but using the hash value $\RO{G}(m)$ as the random coins for \Encrypt:
        \ifTightOnSpace
\else
	\[
		\EncryptDerand(\pk,m):=\Encrypt(\pk,m; \RO{G}(m)) \enspace ,
	\] 
        \fi
	\DecryptDerand uses the decryption algorithm \Decrypt of \PKE to decrypt a ciphertext $c$ %
	to obtain $m'$, and rejects by returning \ifTightOnSpace \else a failure symbol \fi $\bot$ if $c$ fails to decrypt or $m'$ fails to encrypt back to $c$.
	\ifTightOnSpace\else(For the formal definition, see \cref{fig:Def-Derandomized-PKE} on page~\pageref{fig:Def-Derandomized-PKE}).\fi
}

{\heading{PKE-to-KEM transformation \TtwoExplicitMess.}
	Starting from a deterministic encryption scheme \PKE' and a hash function \RO{H},
	key encapsulation algorithm $\KemExplicitMess := \TtwoExplicitMess[\PKE',\RO{H}]$ is built by letting
	\ifTightOnSpace
\else
	\[
		\Encaps(\pk):= (c :=  \Encrypt'(\pk,m), K:= \RO{H}(m)), 
	\]
        \fi
	where $m$ is picked at random from the message space.
	Decapsulation will return $K:= \RO{H}(m)$ unless $c$ fails to decrypt, in which case it returns failure symbol $\bot$.
	\ifTightOnSpace\else(For the formal definition, see \cref{fig:Def-FOExplicit} on page~\pageref{fig:Def-FOExplicit}).\fi
}

{\heading{Combined PKE-to-KEM transformation \FOExplicitMess.}
	The 'full FO' transformation \FOExplicitMess is defined by taking \PKE and hash functions \RO{G} and \RO{H},  and defining $\FOExplicitMess[\PKE,\RO{G},\RO{H}]  :=  \TtwoExplicitMess[\PKEDerand, \RO{H}]$.
	While there exists a plethora of variants that differ from \FOExplicitMess, it was proven~\cite{TCC:BHHHP19} that security of these variants is either equivalent to or implied by security of \FOExplicitMess.
\ifTightOnSpace \else
	To offer a more complete picture, we  recap these variants and their relations in
	\ifCameraReady the full version.
	\else \cref{sec:appendix:FO} (page~\pageref{sec:appendix:FO}).
	\fi
	The take-away message is that any security result for \FOExplicitMess also covers its variants.
\fi
}

\subheading{The role of correctness errors in security proofs for FO.} %
{
	Correctness errors play a role during the proof that an FO-transformed KEM is \INDCCA secure:
	To tackle the \CCA part, it is necessary to simulate the decapsulation oracle \oracleDecaps without the secret key,
	meaning the plaintext has to be obtained via strategies different from decrypting.
	While different strategies for this exist in both ROM and QROM, they all have in common that the obtained plaintext is rather a plaintext that encrypts to the queried ciphertext (a ``ciphertext preimage'') than the decryption.
	Consequently, the simulation fails to recognise \emph{failing ciphertexts}, i.e., ciphertexts for which decryption results in a plaintext different from the ciphertext preimage (or even in $\bot$), and will in this case behave differently from \oracleDecaps.
	Hence, the simulations are distinguishable from \oracleDecaps if the attacker can craft such failing ciphertexts.
	
    The approach chosen by~\cite{TCC:HofHovKil17} was to show that the distinguishing advantage between the two cases can be bounded by the advantage in a game \COR. Game \COR (defined in~\cite{TCC:HofHovKil17}) provides an adversary with a key pair (including the secret key) and asks to return a {\em failing message}, i.e., a message that encrypts to a failing ciphertext, for the derandomized scheme \PKEDerand. \cite{TCC:HofHovKil17} further bounded the maximal advantage in game \COR for \PKEDerand in terms of a statistical worst-case quantity $\deltaWorstCase$ of \PKE,
    which is the expected maximum probability for plaintexts to cause a decryption failure, with the expectation being taken over the %
    key pair.
    This results in a typical %
    search bound as the adversary can use the secret key to check if a %
    ciphertext fails. In the QROM, the resulting bound is therefore $8q^2\deltaWorstCase$, %
    $q$ being the number of queries to $\RO G$.\footnote{
		Some publications (e.g., \cite{C:JZCWM18}) use the bound $2 q \cdot \sqrt{\deltaWorstCase}$,
		it is however straightforward to verify that the bound above can be achieved by using \cite[Lemma 2.9]{PKC:HKSU20} as a drop-in replacement. Note that this is indeed a quadratic improvement unless $4 q \cdot \sqrt{\deltaWorstCase} >  1$,
		in which case the $\INDCCA$ bound is meaningless, anyways.
	}

	Intuitively, this notion suffers from two related unnatural features:
	\ifTightOnSpace
\else
	\begin{itemize}
		\item  First, it looks rather unnatural to provide any adversary with the secret key, as long as the scheme achieves at least some basic notion of security.\footnote{Schemes that allow for a key recovery attack serve as pathological examples why this argument does not hold in generality.}
		In particular, this observation applies to adversaries tasked with finding failing plaintexts,
		and in fact, this is not a mere issue of aesthetics: If the secret key is given to the adversary, an analysis of this bound cannot make use of computational assumptions without becoming heuristic.\footnote{An example we happen to be aware of is the analysis of the correctness error bound of Kyber~\cite{kyber}.} 
		
		\item Second, it seems unnatural that the bound contains a Grover-like search term with regard to \deltaWorstCase: 
		As \INDCCA adversaries do not have access to the secret key, they can only check whether ciphertexts fail via their \emph{classical} \CCA oracle, which should render a Grover search impossible. Furthermore, in both ROM and QROM, it should be the (usually much smaller) number of  \CCA queries that limits the adversary's ability to search, and not the number of random oracle queries. %
		Hence this bound seems overly conservative as long as the scheme achieves at least some basic notion of security.
	\end{itemize}
\fi

	While follow-up works have used different games in place of \COR to deal with decryption errors, all result in the same quantum search bound in terms of \deltaWorstCase.
}

{\ifodd0		
     In the QROM, there exists a common simulation strategy for FO variants whose decapsulation rejects ciphertexts that fail to pass the de-/reencryption check rather by returning a pseudorandom value derived from $c$ than by returning $\bot$.
     (This is often called ``implicit reject''; returning $\bot$ is then called ``explicit reject''.)	
     The simulation, however, is distinguishable from \oracleDecaps if the attacker can craft encryptions that exhibit decryption failure.

	Note that \Encaps computes its ciphertexts by executing the deterministic encryption algorithm \EncryptDerand.
	By plugging \EncryptDerand into the key derivation oracle, i.e., by letting $\RO{H} \coloneqq \RO{H'}(\EncryptDerand(-))$,
	$\oracleDecaps(c)$ can be simulated by simply returning $\RO{H'}(c)$.
	This is where correctness errors come into play: If an attacker is able to come up with some $c$ that is a proper encryption of some message $m$, but decrypts to some message $m' \neq m$, it will recognise this simulation as $\oracleDecaps(c)$ will respond with $\RO{H}(m)$ instead of $\RO{H}(m')$.
	It is furthermore clear that this simulation approach would not work for FO variants whose decapsulation algorithm returns $\bot$ for invalid ciphertexts.

	As a ``sledge-hammer solution'' to the distinguishability issue, it is then shown that an attacker succeeding in distinguishing between \oracleDecaps and its simulation can be used to 
	win a correctness game, introduced in~\cite{TCC:HofHovKil17}. In this game, the adversary is given access to the full key pair, including the secret key, and 
	
	succeed in a particular quantum distinguishing problem, called \GDPB.
	In the \GDPB game, the distinguisher \Ad{D} has access to a function $f: \MSpace \rightarrow \left[0, 1\right]$ that is either the constant zero function or a function that returns 1 on message $m$ with the probability that $m$ fails to decrypt,
	and succeeds if it correctly guesses to which function it had access to.
	Invoking \GDPB, however, adds a Grover-like term of $8 q^2 \cdot \deltaWorstCase$ to the security bound, where \deltaWorstCase is the worst-case correctness term of \PKE as introduced in \cite{TCC:HofHovKil17} and $q$ is the number of random oracle and decapsulation queries.%
	\footnote{
		Some publications \cite{C:JZCWM18, PKC:JiaZhaMa19,EPRINT:JiaZhaMa19b} use the bound $2 q \cdot \sqrt{\deltaWorstCase}$,
		it is however straightforward to verify that the bound above can be achieved by using \cite[Lemma 2.9]{PKC:HKSU20} as a drop-in replacement. Note that this is indeed a quadratic improvement unless $4 q \cdot \sqrt{\deltaWorstCase} >  1$,
		in which case the $\INDCCA$ bound is meaningless, anyways.
	}
	Furthermore, we want to stress that \deltaWorstCase is a statistical term that can be viewed as the advantage of an attacker trying to find a failing ciphertext, \emph{given the private key}. As the private key is not handed to \INDCCA attackers, we would view it as more natural to also handle decryption failures without handing the private key to adversaries attacking the correctness of the scheme. \todo{mention that we also have a $\deltaWorstCase$ -related bound in the ROM.}
	
	Compared to the respective ROM bound, we lose an additional factor of $q$, simply because \INDCCA adversaries now have quantum access to their random oracles.
	But is this scaling up really necessary? In the \GDPB game, \Ad{D} can search the function for failing messages in superposition. \INDCCA adversaries, however, are neither provided with the secret key nor do they have quantum access to \oracleDecaps. How are they supposed to search for failing messages in superposition?
	There is reason to suspect that with a more sophisticated simulation strategy, we might be able to achieve several goals in one:
	\begin{itemize}
		\item Not having to invoke \GDPB at all, thereby hopefully improving the bound related to decryption failures in the QROM,
		\item having a simulation that is applicable to explicitly rejecting FO variants in the QROM, 
		\item working with a correctness requirement that is more natural in the context of security games, both for ROM and QROM.
	\end{itemize}
	
	Since such a result for \FOExplicitMess would imply respective results for all other variants (see \cref{sec:appendix:FO}),
	it would immediately also yield a result for all other FO variants with similar bounds and the same correctness requirement.
\fi }

\subheading{Main contribution.} 
Our main contribution is a new security reduction for the FO transformation that improves over existing ones in two ways.\vspace{.1cm}

\heading{Decryption failures.} We introduce a family of new security games, the \emph{\underline{F}ind \underline{F}ailing \underline{P}laintext} (\FFP) games. These provide a much more natural framework for dealing with decryption errors in the FO transformation, and it is the novel structure of our reduction that allows their usage.
Two important members of the \FFP family are as follows:
The first one, 
\emph{\underline{F}ind \underline{F}ailing \underline{P}laintext that is \underline{N}on-\underline{G}eneric} (\FFPNG),
gives a public key to the adversary and %
asks it to find
a message that triggers a decryption failure more likely 
with respect to this key pair than with respect to an independent key pair.
The second one, \emph{\underline{F}ind \underline{F}ailing \underline{P}laintext with \underline{N}o \underline{K}ey} (\FFPNoKey), tasks an adversary with producing a message that triggers a decryption failure with respect to an independently sampled key pair, without providing any key to the adversary. 
As summarised in \cref{fig:Intro}, we provide a reduction from \FFPNG and passive security of \PKE together with \FFPNoKey for \PKEDerand to \INDCCA security of the FO-transformed of \PKE.
This new reduction structure avoids both unnatural features mentioned above: 
	\ifTightOnSpace
\else
\begin{itemize}
	\item None of the two failure-related games \FFPNG and \FFPNoKey provide the adversary with the secret key. In particular, we show how to bound an adversary's advantage in game \FFPNoKey in terms of \deltaIndKey, the worst-case decryption error rate when the message is picked \emph{independently of the key}, and additional statistical parameters of the probability distributions of decryption failures for fixed message. We give two concrete example bounds, one involving the variance based on Chebyshev's inequality and one based on a Gaussian-shaped tail bound. We expect that these ``independent-key'' statistical parameters can be estimated more conveniently and \emph{without heuristics}, by exploiting the computational assumptions of the PKE scheme at hand.
	\item Game \FFPNoKey still allows for a Grover search advantage, but only when searching for 
	messages that are more likely to cause a failure \emph{on average over the key}. This game corresponds, e.g., to the first attempt at finding a failure in attacks like \ifTightOnSpace\cite{DVV18,PQCRYPTO:BinSch20, EC:DAnRosVir20}\else \cite{EPRINT:DAnVerVer18a,PQCRYPTO:BinSch20, EC:DAnRosVir20}\fi. 
	In the context of the entire security reduction for the FO transformation, the advantage in this game is multiplied with the number of decapsulation queries a \CCA attacker makes, correctly reflecting the fact that the ability of \emph{identifying} a decryption failure should depend on the \CCA oracle and is thus limited.
\end{itemize}
 \fi
Game \FFPNG defines a property of the underlying PKE scheme, it thus allows to analyze the hardness of finding meaningful decryption failures independently from the hardness of searching a random oracle for them. \FFPNG seems thus more amenable to both security reductions and cryptanalysis.\vspace{.1cm}

\heading{FO with explicit rejection.} Our reduction employs a technique for generalized \emph{preimage extraction} in the QROM that was recently introduced in \cite{DFMS21}. As shown by \cite{DFMS21}, this technique is well-suited for proving  \FOExplicitMess secure. We furthermore generalize the one-way to hiding (\OWTH) lemma \cite{C:AmbHamUnr19} such that it is compatible with the technique from \cite{DFMS21}.
\OWTH was used to derive the state-of-the-art bounds for implicitly rejecting variants, and combining the two techniques,
we obtain a security bound for \FOExplicitMess that is competitive with said state-of-the-art bounds. 

\heading{QROM tools.} 
To facilitate the above-described reduction, we provide two technical tools that might be of independent interest: Firstly, we generalize the \OWTH framework from \cite{C:AmbHamUnr19} such that it can be combined with the extractable quantum random oracle simulation from \cite{DFMS21}, rendering the two techniques compatible with being used together in the same security reduction. We make crucial use of this possibility to avoid the additional reduction losses that \cite{DFMS21} need to accept to be able to use the plain one-way to hiding framework in juxtaposition with the extractable simulator.

Secondly, we prove query lower bounds for tasks where an algorithm has access to a QRO (or even an extractable simulator thereof) and has to output an input value $x$ which, together with the corresponding oracle output $\RO{RO}(x)$, achieves a large value under some figure-of-merit function. We use this technical result to provide the aforementioned bounds for the adversarial advantage in the \FFPNoKey game, but they might prove of independent interest.

\ifTightOnSpace\else 
	\subheading{Organisation of this work.}
	\cref{sec:prelims} recalls standard definitions for PKE schemes/KEMs,
	and the formal definition of \FOExplicitMess.
	\cref{sec:ROM} gives our random oracle model reduction, substantiating the upper half of \cref{fig:Intro} in the ROM.
	\cref{sec:QROM} is the QROM equivalent of \cref{sec:ROM}.
	Since \cref{sec:QROM} uses the extractable quantum random oracle simulation from \cite{DFMS21}, we squeeze in a recap of this extension in \cref{sec:prels:compressed} to establish notation and for the reader's convenience.
	\cref{sec:PKEDerand:FFPCPA} analyzes \FFPCPA security of \PKEDerand further, thereby substantiating the lower half of \cref{fig:Intro}.
	\cref{sec:final-result} ties together \cref{sec:ROM}/\ref{sec:QROM} with \cref{sec:PKEDerand:FFPCPA} by providing corollaries that use concrete bounds for the \INDCCA security of $\FOExplicitMess[\PKE,\RO{G},\RO{H}]$.
	The bounds include a term in $\gamma$, the spreadness of \PKE. In \cref{sec:spreadness}, we calculate this term for two easy-to-analyze candidates, \HQCPKE and \FrodoPKE.
	
\fi 

\subheading{TL;DR for scheme designers.}
\cref{sec:final-result} provides concrete bounds for the \INDCCA security of $\FOExplicitMess[\PKE,\RO{G},\RO{H}]$.
Besides having to analyze the conjectured passive security of \PKE, applying the bounds to a concrete scheme \PKE requires to analyze the following computational and statistical properties:
\ifTightOnSpace
\else 
\begin{itemize}
	\item $\gamma$, the spreadness of \PKE.%
	\item An upper bound for \FngFPCPA against \PKE.
	\item Either an upper bound for \FFPNoKey for \PKEDerand, in our extended oracle model that allows preimage extractions, or alternatively, two statistical values: \deltaIndKey, the worst-case decryption error rate when the message is picked \emph{independently of the key}, and $\sigma_\deltaIndKey$, the maximal variance of \deltaIndKey.
\end{itemize}
\fi

\cm{Do we want the $\FngFPCPA_{\PKE}$+failure tail bound$\implies\FFPCPA_{\PKEDerand}$ in the ROM as well?}\tinka{Maybe after the deadline ;) }

\begin{figure}[t]
	\begin{center}\scalebox{0.9}{
\begin{tikzpicture} \small
			\node (FFPCPA){\begin{minipage}{2cm}\centering \PKEDerand \\ \FFPCPA \end{minipage}};
			\node at ($(FFPCPA)+(4.5,0)$) (FFPCCA) {\begin{minipage}{2cm}\centering \PKEDerand \\ \FFPCCA\end{minipage}};
			
			\node (PKEind) [below = 0.5cm of FFPCPA] {\begin{minipage}{2cm}\centering \PKE \\ \INDCPA \end{minipage}};
			\node (PKEow) [below = 0.5cm of PKEind] {\begin{minipage}{2cm}\centering \PKE \\ \OWCPA \end{minipage}};
			
			\node at ($(PKEind) +(4.5,0)$) (KEMindCPA) {\begin{minipage}{2cm}\centering \KemExplicitMess \\ \INDCPA \end{minipage}};
			\node at ($(FFPCCA) !.5! (KEMindCPA)+(1.2,0)$) (connectFFPandIND) {};
			
			\node  (belowKEMindCPA) [below = 0.5cm of KEMindCPA] {\begin{minipage}{2cm}\centering  \phantom{X} \\ \phantom{X} \end{minipage}};
			
			\node at ($(connectFFPandIND)+(3.3,0)$) (KEMindCCA)
			{\begin{minipage}{2cm}\centering \KemExplicitMess \\ \INDCCA \end{minipage}};
			
			\node at ($(belowKEMindCPA) +(4.5,0)$) (KEMImplicitINDCCA) {\begin{minipage}{2cm}\centering \KemImplicitMess \\ \INDCCA \end{minipage}};

			\draw[>=latex, dashed, ->] (FFPCPA.east) -- (FFPCCA.west) node [draw=none, midway, above]{Thms. \ref{thm:SimDec:ROM}/\ref{thm:SimDec:QROM}};
			
			\draw[>=latex, dashed,->] (PKEind.east) -- ($(KEMindCPA.west)$)
				node [draw=none, midway, above] {\FOExplicitMess,
												Thms. \ref{thm:PKEpassiveToINDCPAKEM:ROM}/\ref{thm:INDPKEToINDCPAKEM}};
			
			\draw[>=latex, dashed] (PKEow.east) -- ($(belowKEMindCPA.west)$)
				node [draw=none, midway, above] {\FOExplicitMess,
												Thms. \ref{thm:PKEpassiveToINDCPAKEM:ROM}/\ref{thm:OWPKEToINDCPAKEM}};			
											
			\draw[>=latex, dashed, ->] ($(belowKEMindCPA.west)$) -| ($(KEMindCPA.south)-(0,0.1)$);			
			
			\draw (FFPCCA.east) -| (connectFFPandIND.east) ;
			\draw (KEMindCPA.east) -| (connectFFPandIND.east) ;
			
			\draw[>=latex, ->] (connectFFPandIND.east) -- (KEMindCCA.west)
			node [draw=none, midway, above] {Thm. \ref{thm:INDandFFPtoCCA:ROM}/\ref{thm:INDandFFPtoCCA:QROM}};
			
			\draw[>=latex, ->] (KEMindCCA.south) -- (KEMImplicitINDCCA.north)
			node [draw=none, midway, left] {Rmk. \ref{remark:FOexplicitToFOImplicit}};
\end{tikzpicture} 	}\end{center}
	\vspace{-7pt} {\unskip\ \hrule\ } \vspace{-3pt}
	\begin{center}\scalebox{0.9}{
		\begin{tikzpicture} \small
			\node (deltaSigma){\begin{minipage}{2cm}\centering \deltaIndKey, $\sigma_{\deltaIndKey}$ small \end{minipage}};
	
			\node at ($(deltaSigma)+(4.5,0)$) (FFPNoKey) {\begin{minipage}{2cm}\centering \PKEDerand \\ \FFPNoKey \end{minipage}};
			
			\node (FngFP) [below = 0.5cm of FFPNoKey] {\begin{minipage}{2cm}\centering \PKE \\ \FngFPCPA \end{minipage}};

			\node at ($(FFPNoKey) !.5! (FngFP)+(1.2,0)$) (connectNoKeyandFngFP) {};
			
			\node at ($(connectFFPandIND)+(2.8,0)$) (FFPCPA) {\begin{minipage}{2cm}\centering \PKEDerand\\\FFPCCA\end{minipage}};
			
			\draw[>=latex, dashed, ->] (deltaSigma.east) -- (FFPNoKey.west)node [draw=none, midway, above]{Thm. \ref{thm:PKEDerand:FFPNoKey}};
			
			\draw (FFPNoKey.east) -| (connectNoKeyandFngFP.east) ;
			\draw (FngFP.east) -| (connectNoKeyandFngFP.east) ;
			
			\draw[>=latex, ->] (connectNoKeyandFngFP.east) -- (FFPCPA.west) node [draw=none, midway, above] {Thm. \ref{thm:PKEDerand:FFPCPA}};
		\end{tikzpicture}
	}\end{center}
\caption{Summary of our results.
	Top: ''Ths. X/Y`` indicates that we provide a ROM theorem X (in \cref{sec:ROM}) and a QROM theorem Y (in \cref{sec:QROM}).
	Bottom: Breaking down \FFPCPA security of \PKEDerand (\cref{sec:PKEDerand:FFPCPA}).
	Solid (dashed) arrows indicate tight (non-tight) reductions in the QROM.
	We want to emphasize that \cref{thm:SimDec:ROM,thm:SimDec:QROM} have comparably mild tightness loss:
	The loss is linear in the number of decryption queries.
	The QROM loss for \cref{thm:INDPKEToINDCPAKEM,thm:OWPKEToINDCPAKEM} is like the one for previously known reductions.	
}
\label{fig:Intro}
\end{figure}

\subheading{Acknowledgements.}
We would like to thank Dominique Unruh for valuable discussions about the semi-classical one-way to hiding lemma and Manuel Barbosa for pointing out the use of heuristics in bounds for delta.

\ifCameraReady \else %
	
	\ifTightOnSpace
\else
		
		\section{Preliminaries.}\label{sec:prelims}
		
		For convenience, we recall the formal definition of the Fujisaki-Okamoto transformation with explicit rejection (as already described above) in \cref{sec:prels:FO},
		and standard definitions for public-key encryption and key encapsulation algorithms in \cref{sec:Prels:PKE:Security}.
		
		For a finite set $S$, we denote the sampling of a uniform random element $x$ by $x \uni S$,
		and we denote deterministic computation of an algorithm \Ad{A} on input $x$ by $y := \Ad{A}(x)$.
		By $\bool{B}$ we denote the bit that is 1 if the Boolean statement $B$ is true,
		and otherwise 0.

\ifTightOnSpace
	\section{The Fujisaki-Okamoto transformation with explicit rejection} \label{sec:prels:FO}
\else
	\subsection{The Fujisaki-Okamoto transformation with explicit rejection} \label{sec:prels:FO}
\fi

This section recalls the definition of \FOExplicitMess.
To a public-key encryption scheme $\PKE = (\KG, \Encrypt, \Decrypt)$
with message space $\MSpace$, randomness space $\RSpace$, and 
hash functions $\RO{G}:  \MSpace \rightarrow \RSpace$
and
$\RO{H}: \{0,1\}^* \rightarrow \{0,1\}^n$,
we associate 
\begin{eqnarray*}
	\KemExplicitMess & := & \FOExplicitMess[\PKE,\RO{G},\RO{H}] :=  (\KG, \Encaps, \Decaps)\enspace .
\end{eqnarray*}
Its constituting algorithms are given in \cref{fig:Def-FOExplicit}.
\FOExplicitMess uses the underlying scheme \PKE in a derandomized way by using $\RO{G}(m)$ as the encryption coins (see line \ref{line:Def-FO:Derandomise})
and checks during decapsulation whether the decrypted plaintext does re-encrypt to the ciphertext (see line \ref{line:Def-FO:Reencrypt}).
This building block of \FOExplicitMess, i.e., the derandomisation of \PKE and performing a reencryption check, is incorporated in the following transformation \Tone:
\begin{eqnarray*}
	\PKEDerand & := & \Tone[\PKE,\RO{G}] :=  (\KG, \EncryptDerand, \DecryptDerand) \enspace ,
\end{eqnarray*}
with its constituting algorithm given in \cref{fig:Def-Derandomized-PKE}.

\begin{figure}[b]\begin{center} 
		
	\nicoresetlinenr
	
	\fbox{\small
		
		\begin{minipage}[t]{3.7cm}	
			\underline{$\Encaps(\pk)$}
			\begin{nicodemus}
				\item $m \uni \MSpace$
				\item $c := \Encrypt(\pk,m; \RO{G}(m))$ \label{line:Def-FO:Derandomise}
				\item $K:=\RO{H}(m)$ \label{line:KeyDerivationModeEncaps}
				\item \pcreturn $(K, c)$
			\end{nicodemus}
		\end{minipage}
		
		\quad 
		
		\begin{minipage}[t]{6.1cm}	
			
			\underline{$\Decaps(\sk,c)$}
			\begin{nicodemus}
				\item $m' := \Decrypt(\sk,c)$
				\item \pcif $m' = \bot$ \pcor $c \neq \Encrypt(\pk, m'; \RO{G}(m'))$\label{line:Def-FO:Reencrypt}
				\item \quad \pcreturn $\bot$
				
				\item \pcelse
				\item \quad \pcreturn $K:=\RO{H}(m')$  
			\end{nicodemus}
			
		\end{minipage}
	}	
\end{center}
	\caption{Key encapsulation mechanism $\KemExplicitMess = (\KemGen,\Encaps, \Decaps)$,
		obtained from $\PKE= (\KG, \Encrypt, \Decrypt)$ by setting  $\KemExplicitMess \coloneqq \FOExplicitMess[\PKE, \RO{G}, \RO{H}]$.}
	\label{fig:Def-FOExplicit}
\end{figure}

\begin{figure}[tb]\begin{center}
		
	\nicoresetlinenr
		
	\fbox{\small
			
		\begin{minipage}[t]{3.7cm}	
			\underline{$\Encrypt^{\RO G}(\pk)$} 
			\begin{nicodemus}
				\item $m \uni \MSpace$
				\item $c := \Encrypt(\pk,m; \RO{G}(m))$
				\item \pcreturn $c$
			\end{nicodemus}
		\end{minipage}
		
		\;
		
		\begin{minipage}[t]{5.7cm}	
			
			\underline{$\Decrypt^{\RO G}(\sk,c)$}
			\begin{nicodemus}
				\item $m' := \Decrypt(\sk,c)$
				\item \pcif $m' = \bot$ \pcor $c \neq \Encrypt(\pk, m'; \RO{G}(m'))$
				\item \quad \pcreturn $\bot$ %
				\item \pcelse
				\item \quad \pcreturn $m'$  %
			\end{nicodemus}
		\end{minipage}
	
	}	
		
	\caption{Derandomized \PKE scheme $\PKE^\RO{G}=(\KG, \Encrypt^{\RO G}, \Decrypt^\RO{G})$,
		obtained from \PKE $=(\KG,\Encrypt,\Decrypt)$ by encrypting a message $m$ with randomness $\RO G(m)$ for a random oracle $\RO G$, and incorporating a re-encryption check during $\Decrypt^{\RO G}$.}
	\label{fig:Def-Derandomized-PKE}
	\end{center}
\end{figure}

\ifEprint %
	\subsection{Security Notions for Public-Key Encryption} \label{sec:Prels:PKE:Security}
\else %
	\section{Security Notions for Public-Key Encryption} \label{sec:Prels:PKE:Security}
\fi

We also consider all security games in the (quantum) random oracle model,
where \PKE and adversary \Ad{A} are given access to (quantum) random oracles.
(How we model quantum access is made explicit in \cref{sec:prels:compressed}.)

\ifEprint %
	\subsubsection{Definitions for \PKE}
\else %
	\subsection{Definitions for \PKE}
\fi

\begin{definition}[$\gamma$-spreadness]
	We say that \PKE is $\gamma$-spread iff for all key pairs $(\pk, \sk) \in \supp(\KG)$
	and all messages $m \in \MSpace$ it holds that
	\[ \max_{c \in \CSpace} \Pr[\Encrypt(\pk, m)= c] \leq 2^{- \gamma} \enspace , \]
	where the probability is taken over the internal randomness \Encrypt.
\end{definition}

We also recall two standard security notions for public-key encryption:
\underline{O}ne-\underline{W}ayness under
\underline{C}hosen \underline{P}laintext \underline{A}ttacks (\OWCPA)
and
\underline{Ind}istinguishability under \underline{C}hosen-\underline{P}laintext \underline{A}ttacks (\INDCPA).
\begin{definition}[\OWCPA, \INDCPA]
	Let $\PKE = (\KG,\Encrypt,\Decrypt)$ be a public-key encryption scheme with message space \MSpace.
	We define the \OWCPA game as in \cref{fig:Def:PKE:passive} and the \OWCPA \textit{advantage function of an adversary \Ad{A} against \PKE} as
	\[ \Adv^{\OWCPA}_{\PKE}(\Ad{A}) := \Pr [\OWCPA^{\Ad{A}}_\PKE \Rightarrow 1 ] \enspace .\]

	Furthermore, we define the 'left-or-right' version of \INDCPA by defining games $\INDCPA_b$, where $b\in \lbrace 0,1 \rbrace$ (also in \cref{fig:Def:PKE:passive}),
	and the \INDCPA \textit{advantage function of an adversary $\Ad{A} = (\Ad{A}_1, \Ad{A}_2)$ against \PKE}
	(where \(\Ad{A}_2\) has binary output)
	as
	\[ \Adv^{\INDCPA}_{\PKE}(\Ad{A}) := |\Pr [ \INDCPA_0^{\Ad{A}} \Rightarrow 1 ] - \Pr [ \INDCPA_1^{\Ad{A}} \Rightarrow 1 ]| \enspace . \]
	
	\begin{figure}[h!]\begin{center}\fbox{\small
		
		\nicoresetlinenr
		
		\begin{minipage}[t]{3.5cm}	
			\underline{{\bf Game} \OWCPA}
			\begin{nicodemus}
				\item $(\pk, \sk) \leftarrow \KG$
				\item $m^* \uni \MSpace$
				\item $c^* \leftarrow \Encrypt(\pk,m^*)$
				\item $m' \leftarrow \Ad{A}(\pk, c^*)$
				\item \pcreturn $\bool{m' = m^*}$
			\end{nicodemus}
		\end{minipage}
		
		\quad
		
		\begin{minipage}[t]{3.7cm}
			\underline{{\bf Game} $\INDCPA_b$}
			\begin{nicodemus}
				\item $(\pk, \sk) \leftarrow \KG$
				\item $(m^*_0, m^*_1, \state) \leftarrow \Ad{A}_1(\pk)$
				\item $c^* \leftarrow \Encrypt(\pk, m^*_b)$
				\item $b' \leftarrow \Ad{A}_2(\pk, c^*, \state)$
				\item \pcreturn $b'$
			\end{nicodemus}	
		\end{minipage}%
	}
	\end{center}
		\caption{Games \OWCPA and $\INDCPA_b$ for \PKE.}
		\label{fig:Def:PKE:passive}
	\end{figure}
	
\end{definition}

\ifEprint %
	\subsubsection{Standard notions for \KEM}
\else %
	\subsection{Standard notions for \KEM}
\fi

We now define \underline{Ind}istinguishability under \underline{C}hosen-\underline{P}laintext \underline{A}ttacks (\INDCPA)
and under \underline{C}hosen-\underline{C}iphertext \underline{A}ttacks (\INDCCA). 

\begin{definition}[\INDCPA, \INDCCA]\label{def:KEM:CPACCA}
	Let $\KEM = (\KemGen, \Encaps, \Decaps)$ be a key encapsulation mechanism with key space \KeySpace.
	For $\atk \in \lbrace \CPA, \CCA\rbrace$, we define \INDATKKEM games as in \cref{fig:Def:KEM:CPACCA}, where
	\[\mathsf{O}_\atk := \left\{
	\begin{array}{ll}
		-								&	\atk = \CPA \\
		\textnormal{\oracleDecaps}		&	\atk = \CCA
	\end{array}	\right. \enspace .	\]
	
	We define the \INDATKKEM \textit{advantage function of an adversary \Ad{A} against \KEM} as
	\[ \Adv^{\INDATKKEM}_{\KEM}(\Ad{A}) := |\Pr[\INDATKKEM^{\Ad{A}} \Rightarrow 1] - \nicefrac{1}{2}| \enspace. \]
	
	\begin{figure}[h!]\begin{center}\fbox{\small
				
		\nicoresetlinenr
				
		\begin{minipage}[t]{3.9cm}
					\underline{{\bf Game} \INDATKKEM}
					\begin{nicodemus}
						\item $(\pk,\sk) \leftarrow \KemGen$
						\item $b \uni \bits$
						\item $(K_0^*,c^*) \leftarrow \Encaps(\pk)$
						\item $K_1^* \uni \KeySpace$
						\item $b'\leftarrow \Ad{A}^{\mathsf{O}_\atk}(\pk, c^*,K_b^*)$
						\item \pcreturn $\bool{b' = b}$
					\end{nicodemus}%
				\end{minipage}%
				\;
				\begin{minipage}[t]{3.2cm}
					\underline{$\oracleDecaps(c \neq c^*)$}
					\begin{nicodemus}
						\item $K := \Decaps(\sk,c)$
						\item \pcreturn $K$
					\end{nicodemus}%
				\end{minipage}%
			}
		\end{center}
		\caption{Game \INDATKKEM for \KEM, where $\atk \in \lbrace \CPA, \CCA\rbrace$ and
			$\mathsf{O}_\atk$ is defined in \cref{def:KEM:CPACCA}.}
		\label{fig:Def:KEM:CPACCA}
	\end{figure}
	
\end{definition}
 		
	\fi

\fi

\section{ROM reduction}\label{sec:ROM}
This section substantiates the upper half of \cref{fig:Intro} in the 
\ifTightOnSpace
\else
random oracle model.
\fi
The first step of  common security reductions for the FO transformation consists of simulating the decapsulation oracle without using the secret key.
This simulation allows transforming an $\INDCCAKEM$-adversary $\Ad{A}$ against $\KemExplicitMess := \FOExplicitMess[\PKE, \RO{G},\RO{H}]$
into an $\INDCPAKEM$-adversary $\tilde{\Ad{A}}$ against the same \KemExplicitMess.
The oracle simulation, however, will not accurately simulate the behaviour of \Decaps for ciphertexts that trigger decryption errors.
We will show that from an adversary capable of distinguishing between the real decapsulation oracle and its simulation,
we can construct an adversary \Ad{B} that is able to extract failing plaintexts for the derandomised version \PKEDerand of \PKE %
\ifTightOnSpace\else 
	(as defined in \cref{fig:Def-Derandomized-PKE} on page \pageref{fig:Def-Derandomized-PKE})%
\fi.
In more detail, we formalise extraction of failing plaintexts as the winning condition of two \underline{F}ind \underline{F}ailing \underline{P}laintext (\FFP) games, which we formally define in \cref{def:FFPATK} (also see \cref{fig:def:FFPATK}).
For $\mathsf{ATK}\in\{\CPA, \CCA\}$, an adversary $\Ad{B}$ playing the  $\FFP$-$\mathsf{ATK}$ game for a deterministic encryption scheme \PKE gets access to the same oracles as in the respective $\IND$-$\mathsf{ATK}$ game, outputs a message $m$,
and wins if $\Decrypt(\Encrypt(m))\neq m$. 
(Here, and in the following, we sometimes omit the arguments $pk$ and $sk$, respectively.)
For such messages $m$ we say that $m$ is a \emph{failing plaintext}, or shorter, that $m$ \emph{fails}.
\ifTightOnSpace \else 
	We will first show in \cref{thm:FFP1} that any attacker against the \INDCCAKEM security of $\FOExplicitMess[\PKE, \RO{G},\RO{H}]$ can be used to construct an \INDCPAKEM attacker against $\FOExplicitMess[\PKE, \RO{G},\RO{H}]$ and an attacker against the correctness of \PKEDerand that has access to \oracleDecrypt,
	i.e., an attacker that succeeds in game $\FFPCCA$.
	The maximum winning probability in the \FFPCCA game is still quite an unwieldy object.
	In particular, clever strategies have been devised to make adaptive use of a decryption oracle towards finding failing plaintexts. Fortunately, we can use the same strategy underlying our simulation of \oracleDecaps once more to show in \cref{thm:FFP2} that any successful \FFPCCA adversary can be used to construct an adversary succeeding in the \FFPCPA game, meaning that it is sufficient to analyse the success probability of attackers trying to come up with failing plaintexts, having nothing on their hands but the public key.
	It is then shown how \cite{TCC:HofHovKil17} can be used to argue that \INDCPAKEM security of $\FOExplicitMess[\PKE, \RO{G},\RO{H}]$ can be based on either \OWCPA or \INDCPA security of \PKE, with the latter implication being tight up to a factor of 3.
	Lastly, we discuss that the main result of this section also works if we consider the implicitly rejecting variant $\FOImplicitMess[\PKE, \RO{G},\RO H]$ instead of $\FOExplicitMess[\PKE, \RO{G},\RO{H}]$ (see \cref{remark:FOexplicitToFOImplicit}).		
\fi
The final bounds we obtain are essentially similar to the ones in \cite{TCC:HofHovKil17} except for involving a different correctness definition, see the discussion after \cref{remark:FOexplicitToFOImplicit}.

\begin{definition}[\FFPATK] \label{def:FFPATK}
	Let $\PKE=(\KG,\Encrypt,\Decrypt)$ be a deterministic public-key encryption scheme. 
	For $\atk \in \lbrace \CPA, \CCA \rbrace$,
	we define \FFPATK games as in \cref{fig:def:FFPATK},
	where \ifTightOnSpace
\else
		\[\mathsf{O}_\atk := \left\{
	\begin{array}{ll}
		-								&	\atk = \CPA \\
		\textnormal{\oracleDecrypt}		&	\atk = \CCA
	\end{array}
	\right. \enspace .
	\]
		\fi
We define the \FFPATK \textit{advantage function of an adversary \Ad{A} against \PKE} as
	\[ \Adv^{\FFPATK}_{\PKE}(\Ad{A}) 	:= \Pr [\FFPATK^{\Ad{A}}_\PKE \Rightarrow 1 ] \enspace .\]

	\begin{figure}[tb] \begin{center} \fbox{\small
				
		\nicoresetlinenr
		
		\begin{minipage}[t]{4.0cm}	
			
			\underline{{\bf Game} \FFPATK}
			\begin{nicodemus}
				\item $(\pk, \sk) \leftarrow \KG$
				\item $m \leftarrow \Ad{A}^{\textsc{O}_\atk, \RO{G}}(\pk)$
				\item $c := \Encrypt(\pk,m)$
				\item $m' := \Decrypt(\sk,c)$
				\item \pcreturn $\bool{m' \neq m}$
			\end{nicodemus}
			
		\end{minipage}
		
		\quad
		
		\begin{minipage}[t]{4.0cm}
			\underline{$\oracleDecrypt(c)$}
			\begin{nicodemus}
				\item $m := \Decrypt(\sk,c)$
				\item \pcreturn $m$
			\end{nicodemus}
		\end{minipage}
	}
	\end{center}
	\caption{
		Games \FFPATK for a deterministic \PKE, where $\atk \in \lbrace \CPA, \CCA \rbrace$.
		$\mathsf{O}_\atk$ is the decryption oracle present in the respective \INDATK game (see \cref{def:FFPATK}) and \RO{G} is a random oracle, provided if it is used in the definition of \PKE.
	}
	\label{fig:def:FFPATK}
	\end{figure}
	
\end{definition}

\noindent Note that in neither \FFPATK game, the adversary has access to the secret key.
In particular, the \FFPCPA game only differs from the correctness game \COR defined in \cite{TCC:HofHovKil17} in exactly this fact, as game \COR additionally provides the secret key.
We note that an adversary winning either \FFPATK game for a deterministic scheme \PKE can be used to win in game \COR.

We begin by introducing two simulations of the \Decaps oracle, oracle \oracleDecapsSim
and a variant \oracleDecapsSimFail of \oracleDecapsSim.
\oracleDecapsSimFail extracts failing plaintexts from adversarial decapsulation queries,
and is simulatable by \FFP adversaries with access to the decryption oracle \oracleDecrypt for \PKEDerand.
Both simulations of the \Decaps oracle make use of a list $\mathcal{L}$ of previous queries to $\RO{G}$ and their respective encryptions.
For this to work, we replace \RO{G} with a modification $\RO{G'}$ that keeps track of all issued queries and compiles $\mathcal{L}$.
The original \Decaps oracle and its simulations are defined in \cref{fig:SimDecaps}, using the following conventions.
For a set of pairs $ \mathcal{L} \subset \mathcal X \times \mathcal Y$, we assume that a total order is chosen on $\mathcal X$ and $\mathcal Y$.
We denote by $\mathcal{L}^{-1}(y)$ the first preimage of $y$. Formally, we define $\mathcal{L}^{-1}(y)$ by setting
\begin{equation}\label{eq:PickPreimageFromList}
	\mathcal{L}^{-1}(y) \coloneqq \begin{cases}
		x		& \text{if }(x,y)\in \mathcal{L} \text{ and } x\le x' \text{ for all } x' \text{ s. th. } (x',y)\in \mathcal{L}\\
		\bot	&\nexists \ x \text{ s. th. } (x,y)\in \mathcal{L}.
	\end{cases}
\end{equation}

The simulation \oracleDecapsSim can, however, only \emph{reverse} encryptions that were already computed by the adversary (with a query to oracle \RO{G'}) \emph{before} their query to oracle \oracleDecapsSim, which is where the spreadness of \PKE comes into play:
If $\gamma$ is large, it becomes unlikely that the attacker can guess an encryption $c = \Encrypt(\pk, m; \RO{G}(m))$ without a respective query to \RO{G}.
\oracleDecapsSim will furthermore answer inconsistently if the reversion (in other words, the preimage) of $c$ differs from its decryption, meaning that $c$ belongs to a failing plaintext that can be recognized by the failure-extracting variant \oracleDecapsSimFail.

\begin{figure}[tb]\begin{center}\makebox[\textwidth][c]{
	\nicoresetlinenr
		\centering
	\resizebox{\textwidth}{!}{
	\fbox{\small
		\begin{minipage}[t]{4.3cm}	
			
			\underline{$\oracleDecaps(c)$}
			\begin{nicodemus}
				\item $m' := \Decrypt(\sk,c)$
				\item \pcif $m'=\bot$
					\item \quad $\pcreturn K {\coloneqq} \bot$
				\item \pcelse
					\item \quad  $c' := \Encrypt(\pk, m'; \RO{G}(m'))$
					\item \quad \pcif $c \neq c'$
						\item \quad\quad \pcreturn $\bot$
					\item \quad \pcelse
						\item \quad\quad \pcreturn $\RO{H}(m')$  
			\end{nicodemus}
			
			\ \\
		
			\underline{$\RO{G'} (m)$}
			\begin{nicodemus}
				\item $r := \RO{G}(m)$
				\item $c \coloneqq \Encrypt(\pk, m; r)$
				\item $\ListQueriesG\coloneqq\ListQueriesG\cup \{(m, c)\}$
				\item \pcreturn $r$
			\end{nicodemus}
			
		\end{minipage}
		\;
		
		\begin{minipage}[t]{4.5cm}	
			
			\underline{$\oracleDecapsSim(c \neq c^*)$}
			\begin{nicodemus}
				\item $m \coloneqq \ListQueriesG^{-1}(c)$ \label{line:SimDecaps:Extract1}
				\item \pcif $m=\bot$
					\item \quad $\pcreturn K {\coloneqq} \bot$
				\item \pcelse
				\item \quad \pcreturn $K \coloneqq \RO{H}(m)$
			\end{nicodemus}
		
			\ \\
		
		\underline{$\oracleDecrypt(c\neq c^*)$}
		\begin{nicodemus}
			\item $m' \coloneqq \Decrypt(\sk, c)$
			\item \pcif $m' = \bot $ 
				\item \quad \pcreturn $\bot$
			\item \pcelse
				\item \quad  \pcif $\Encrypt(\pk, m'; \RO{G}(m')) \neq c $ \label{line:SimDecaps:Reenc}
					\item \quad \quad \pcreturn $\bot$
				\item \quad \pcelse \pcreturn $m'$
		\end{nicodemus}
		\end{minipage}
		
		\;
		
		\begin{minipage}[t]{4.9cm}	
			
			\underline{$\oracleDecapsSimFail(c\neq c^*)$}
			\begin{nicodemus}
				\item $m\coloneqq \ListQueriesG^{-1}(c)$ \label{line:SimDecaps:Extract2}
				\item $m' \coloneqq \oracleDecrypt(c)$
				\item \pcif $m \neq \bot \pcand m \neq m'$
					\item \quad $\ListFail \coloneqq \ListFail \cup \{ m \}$ 
				\item \pcif $m=\bot$
					\item \quad $\pcreturn K {\coloneqq} \bot$
				\item \pcelse
					\item \quad \pcreturn $K \coloneqq \RO{H}(m)$
			\end{nicodemus}

		\end{minipage}
		
	}	
	}}
	\caption{
		Simulation \oracleDecapsSim of oracle \oracleDecaps for \KemExplicitMess, failing-plaintext-extracting version \oracleDecapsSimFail of \oracleDecapsSim, and decryption oracle \oracleDecrypt for \PKEDerand.
		Oracles \oracleDecapsSim and \oracleDecapsSimFail use in lines \ref{line:SimDecaps:Extract1} and \ref{line:SimDecaps:Extract2} the notation introduced in Equation \eqref{eq:PickPreimageFromList}. Note that \RO{G'} only differs from \RO{G} by compiling list $\ListQueriesG$ (which we assume to be initialized to $\emptyset$).
	}
	\label{fig:SimDecaps}
	\end{center}
\end{figure}

\begin{theorem}[$\FOExplicitMess\lbrack\PKE\rbrack$ \INDCPA and \PKEDerand \FFPCCA \RightarrowROM $\FOExplicitMess\lbrack\PKE\rbrack$ $\INDCCA$]\label{thm:FFP1}\label{thm:INDandFFPtoCCA:ROM}
	Let \PKE be a (randomised) \PKE scheme that is $\gamma$-spread, and let $\KemExplicitMess \coloneqq \FOExplicitMess[\PKE, \RO{G},\RO H]$.
	Let \Ad{A} be an \INDCCAKEM-adversary (in the ROM) against \KemExplicitMess, making at most \numberDecapsQueries many queries to its decapsulation oracle \oracleDecaps %
	.
	Then there exist an \INDCPAKEM adversary $\tilde{\Ad{A}}$ and an \FFPCCA adversary \Ad{B} against \PKEDerand such that
	\begin{equation}\label{eq:FFP1}
		\Adv^{\INDCCAKEM}_{\KemExplicitMess }(\Ad{A})
			\le \Adv^{\INDCPAKEM}_{\KemExplicitMess }\left(\tilde{\Ad{A}}\right)
				+\Adv^{\FFPCCA}_{\PKEDerand}\left(\Ad{B}\right)
				+\numberDecapsQueries\cdot 2^{-\gamma}.
	\end{equation}
	\ifTightOnSpace \else Adversary \fi $\tilde{\Ad A}$ makes $q_\RO{G}$ queries to \RO{G} and $q_\RO{H}+\numberDecapsQueries$ queries to \RO{H}, \ifTightOnSpace \else adversary \fi \Ad B makes $q_\RO{G}$ queries to \RO{G} and $\numberDecapsQueries$ decryption queries, and both adversaries run in about the time of \Ad{A}.
{%
}
\end{theorem}

\begin{proof}
Let \Ad{A} be an adversary against \KemExplicitMess.
We define $\tilde{\Ad{A}}$ as the \INDCPAKEM adversary against \KemExplicitMess that runs $b' \leftarrow \Ad{A}^{\RO G', \RO H, \oracleDecapsSim}$ and returns $b'$.
We furthermore define our \FFPCCA adversary \Ad{B} against \PKEDerand as follows:
\Ad{B} runs $\Ad{A}^{\RO G', \RO H, \oracleDecapsSimFail}$, using its own \FFPCCA oracle \oracleDecrypt to simulate \oracleDecapsSimFail.
As soon as $\oracleDecapsSimFail$ adds a plaintext $m$ to \ListFail, \Ad{B} aborts \Ad{A} and returns $m$.
If \Ad{A} finishes and \ListFail is still empty, \Ad{B} returns $\bot$.

First, we will relate \Ad{A}'s success probability to the one of $\tilde{\Ad{A}}$.
Note that unless $\tilde{\Ad{A}}$'s simulation \oracleDecapsSim of the decapsulation oracle fails,
$\tilde{\Ad{A}}$ perfectly simulates the game to \Ad{A} and wins if \Ad{A} wins.
Let $\EventDecapsSimDiffers$ be the event that $\Ad{A}$ makes a decryption query $c$
such that $\Decaps(\sk, c) \neq \oracleDecapsSim(c)$. 
We bound

\ifTightOnSpace
\else
\begin{align}
	\frac 1 2 + \Adv^{\INDCCAKEM}_{\KemExplicitMess}(\Ad{A})
		= & \Pr\left[\Ad{A} \mathrm{\ wins}\right]\\
		= & \Pr\left[\Ad{A} \mathrm{\ wins} \wedge \neg \EventDecapsSimDiffers \right]
			+ \Pr\left[\Ad{A} \mathrm{\ wins} \wedge \EventDecapsSimDiffers \right]\\
		= & \Pr\left[\tilde{\Ad{A}} \mathrm{\ wins} \wedge \neg \EventDecapsSimDiffers \right]
			+\Pr\left[\Ad{A} \mathrm{\ wins} \wedge \EventDecapsSimDiffers \right]\\
		\le & \Pr\left[\tilde{\Ad{A}} \mathrm{\ wins}\right]+\Pr\left[\EventDecapsSimDiffers\right]\\
		=&\frac 1 2+ \Adv^{\INDCPAKEM}_{\KemExplicitMess}\left(\tilde{\Ad{A}}\right)+\Pr\left[\EventDecapsSimDiffers\right].
\end{align}
\fi

To analyze the probability of event  $\EventDecapsSimDiffers$, we note that
it
covers several cases:
\begin{itemize}
	\item[-]  Original oracle $\oracleDecaps(c)$ rejects, whereas simulation $\oracleDecapsSim(c)$ does not,
		meaning that $c$ is an encryption belonging to a previous query $m$ to \RO{G'}, but fails the reencryption check performed by $\oracleDecaps(c)$. Since the latter means that either $m' \coloneqq  \Decrypt(\sk, c) = \bot$ or that $\Encrypt(\pk, m'; \RO{G}(m')) \neq c = \Encrypt(\pk, m; \RO{G}(m))$, this cases only occurs if $\Decrypt(\sk, c) \neq m$, meaning $m$ fails.
 	\item[-] Neither oracle rejects, but the return values differ, 
        i.e., $c$ is an encryption belonging to a previous query $m$ to \RO{G'}, but decrypts to some message $m' \neq m$.
	\item[-] $\oracleDecapsSim(c)$ rejects, whereas $\oracleDecaps(c)$ does not, 
	i.e., 
	while $c$ would pass the reencryption check, its decryption $m$ has not yet been queried to \RO{G'}.
\end{itemize}

In either of the former two cases, \RO{G'} has been queried on a failing plaintext $m$ and the decapsulation oracle has been queried on its encryption $c$, meaning that the failing plaintext can be found and recognized by \Ad{B} since \Ad{B} can use its own \FFPCCA oracle \oracleDecrypt to simulate \oracleDecapsSimFail.
We will denote the last case by \EventGuessedCT since \Ad{A} has to find a guess for a ciphertext $c$ that passes the reencryption check, meaning it is indeed of the form $c = \Encrypt(\pk, m; \RO{G'}(m))$ for $m := \Decrypt(\sk, c)$, while not having queried \RO{G'} on $m$ yet.
Whenever \EventDecapsSimDiffers occurs, \Ad{B} succeeds unless \EventGuessedCT occurs. In formulae,
\ifTightOnSpace
\else
\begin{align*}
	\Pr\left[\EventDecapsSimDiffers\right] =& \Pr\left[\EventDecapsSimDiffers\wedge\neg\EventGuessedCT\right] + \Pr\left[\EventDecapsSimDiffers\wedge \EventGuessedCT\right]\\
	\le &\Adv^{\FFPCCA}_{\PKEDerand}\left(\Ad{B}\right)+\Pr\left[\EventGuessedCT\right] .
\end{align*}
\fi

Together with Lemma \ref{lem:pr-guess} below, this yields the desired bound. 
\qed
\end{proof}

We continue by bounding the probability of event \EventGuessedCT.
We will also need to analyze a very similar event in \cref{thm:FFP2}, in which we revisit the \FFPCCA attacker \Ad{B} against \PKEDerand, and where we will simulate \Ad{B}'s oracle \oracleDecrypt via an oracle \oracleDecryptSim (see \cref{fig:SimDecrypt}).
Therefore, we generalize the definition of event \EventGuessedCT accordingly.
\begin{restatable}{lemma}{GuessROM}\label{lem:pr-guess}
	Let \PKE be $\gamma$-spread,
	and let \Ad{A} be an adversary expecting \ifTightOnSpace\else random \fi oracles \RO{G}, \RO{H} as well as either a decapsulation oracle \oracleDecaps for \ifTightOnSpace\else $\KemExplicitMess \coloneqq \FOExplicitMess[\PKE, \RO{G},\RO H]$ \fi
	or a decryption oracle \oracleDecrypt for \PKEDerand,
	issuing at most \numberDecOrDecapsQueries queries to the latter.
	When run with \RO{G'} and simulated oracle \oracleDecapsSim (or \oracleDecryptSim, respectively),
	there is only a small probability that original oracle \oracleDecaps (\oracleDecrypt) would not have rejected,
	but simulation \oracleDecapsSim (\oracleDecryptSim) does. 
	Concretely, we have
	\begin{equation}
		\Pr\left[\EventGuessedCT\right]\le \numberDecOrDecapsQueries\cdot 2^{-\gamma}.
	\end{equation}
\end{restatable}

\ifTightOnSpace\else %
\ifTightOnSpace
	\section{Proof of \cref{lem:pr-guess}} \label{sec:proof:GuessGROM}
	
	For easier reference, we repeat the statement of \cref{lem:pr-guess}.
	
	\GuessROM*
	
\fi

\begin{proof}
	The event \EventGuessedCT, i.e. the case that \oracleDecapsSim (\oracleDecryptSim)
	rejects on a ciphertext $c$ where \oracleDecaps (\oracleDecrypt) does not,
	requires that $c = \Encrypt(\pk, m; \RO G(m))$ for $m \coloneqq \Decrypt(\sk, c)$,
	and that \RO{G'} was not yet queried on $m$.
	Let $c$ be any ciphertext queried by the adversary for which \oracleDecaps does not reject,
	and let $m \coloneqq \Decrypt(\sk,c)$. We can bound
	\begin{align*}
		\Pr \left[ \oracleDecapsSim(c) = \bot \right]
		\le & \Pr[\Encrypt(\pk, m, \RO G' (m)) = c \wedge \RO G' \text{ not yet queried on } m] \\ 
		\le & \Pr_{r \uni \RSpace}\left[\Encrypt(\pk, m; r) = c\right]
		\le 2^{-\gamma},
	\end{align*}
	where the penultimate step used that $\RO{G}'$ has the same distribution as random oracle \RO{G} and that $\RO{G}(m)$ has not yet been sampled, and the last step used that \PKE scheme is $\gamma$-spread.
	Applying a union bound, we conclude that
	\begin{align*}
		\Pr\left[\EventGuessedCT\right]\le \numberDecOrDecapsQueries\cdot 2^{-\gamma}.
	\end{align*}
	\vspace{-1.1cm}{}\\
	\qed
\end{proof} %
\fi

So far, we have shown that whenever an \INDCCA adversary \Ad{A}'s behaviour is significantly changed by being run with simulation \oracleDecapsSim instead of the real oracle \oracleDecaps, we can use \Ad{A} to find a failing plaintext,
assuming access to the \FFPCCA decryption oracle \oracleDecrypt for \PKEDerand.
\ifTightOnSpace
\else
We now proceed by showing 
\fi
that \oracleDecrypt can be simulated via oracle \oracleDecryptSim (see \cref{fig:SimDecrypt}) without the secret key,
thereby being able to construct an \FFPCPA adversary from any \FFPCCA adversary that succeeds with the same probability up to (at most) a multiplicative factor equal to the number of decryption queries the \FFPCCA adversary makes.

\begin{theorem}[\PKEDerand \FFPCPA \RightarrowROM \PKEDerand \FFPCCA]\label{thm:FFP2}\label{thm:SimDec:ROM}
	Let \PKE be -$\gamma$-spread, and let \Ad{B} be an \FFPCCA adversary
	against \PKEDerand\ifTightOnSpace\else (in the ROM)\fi, issuing at most $\numberDecQueries$ many decryption queries. Then there exists an \FFPCPA adversary $\tilde{\Ad{B}}$ such that 
	\begin{equation}\label{eq:guessing-qB}
		\Adv^{\FFPCCA}_{\PKEDerand}(\Ad{B})
			\le (\numberDecQueries+1) \cdot \Adv^{\FFPCPA}_{\PKEDerand}\left(\tilde{\Ad{B}}\right)
				+\numberDecQueries \cdot2^{-\gamma} \enspace .
	\end{equation}
	Adversary $\tilde{\Ad{B}}$ makes at most the same number of queries to \RO{G} as \Ad{B} and runs in about the time of \Ad{B} .
\end{theorem}
\begin{proof}
	To simulate \oracleDecrypt, we use a similar strategy as in the proof of Theorem \ref{thm:FFP1}. We define the events $\EventDecapsSimDiffers$ and $\EventGuessedCT$ in the same way as in the proof of Theorem \ref{thm:FFP1}, except now with respect to the adversary $\Ad{B}$ and oracles \oracleDecrypt (\oracleDecryptSim) instead of \oracleDecaps (\oracleDecapsSim). If our simulation does not fail, then a reduction can simulate the \FFPCCA game to \Ad{B} and use \Ad{B}'s output to win its own \FFPCPA game. The simulation will fail if either \EventGuessedCT happens (with probability at most $\numberDecOrDecapsQueries\cdot 2^{-\gamma}$ due to \cref{lem:pr-guess}), or \EventDecapsSimDiffers, while \EventGuessedCT does not, meaning that the failing message triggering \EventDecapsSimDiffers can be extracted from \ListQueriesG.
	Our reduction $\tilde{\Ad{B}}$ combines both approaches (using \Ad{B}'s output and \ListQueriesG).
	Since $\tilde{\Ad{B}}$ has no knowledge of the secret key, it cannot determine which message will let it succeed and hence has to guess.

	Assume without loss of generality that $\Ad{B}$ makes exactly \numberDecQueries many queries to oracle \oracleDecrypt.
	Consider the adversary $\tilde{\Ad{B}}^{\RO G}$ in \cref{fig:SimDecrypt}. $\tilde{\Ad{B}}$  samples $i\leftarrow\{1,...,\numberDecQueries+1\}$ and either runs $\Ad{B}^{\RO G', \oracleDecryptSim}$ until its $i$-th query to \oracleDecryptSim or until the end if $i=\numberDecQueries+1$. To implement \RO{G'}, $\tilde{\Ad{B}}$ uses its oracle \RO{G}. Simulation \oracleDecryptSim is defined in \cref{fig:SimDecrypt} and works analogous to \oracleDecapsSim in the previous proof.
	Finally, $\tilde{\Ad{B}}$ outputs query preimage $\ListQueriesG^{-1}(c_i)$, where $c_i$ is  \Ad{B}'s $i$-th query to decryption oracle \oracleDecryptSim, unless $i=\numberDecQueries+1$, in which case $\tilde{\Ad{B}}$ outputs the output of \Ad{B}.
	
	\begin{figure}[tb]\begin{center}
			
			\nicoresetlinenr
			\resizebox{\textwidth}{!}{
                                \fbox{\small
				
				\begin{minipage}[t]{3.2cm}	
					
					\underline{$\oracleDecryptSim(c)$}
					\begin{nicodemus}
						\item $m\coloneqq \ListQueriesG^{-1}(c)$
						\item \pcreturn $m$
					\end{nicodemus}
\vspace{.1in}
					
					\underline{$\RO{G'} (m)$}
					\begin{nicodemus}
						\item $c \coloneqq \Encrypt(m; \RO{G}(m))$
						\item $\ListQueriesG\coloneqq\ListQueriesG\cup \{(m, c)\}$
						\item \pcreturn $\RO{G}(m)$
					\end{nicodemus}
					
				\end{minipage}
				
				\quad
				
				\begin{minipage}[t]{8.5cm}	
					
					\underline{$\tilde{\Ad{B}}^{\RO G}$}
					\begin{nicodemus}
						\item $i \uni \{1,...,\numberDecQueries+1\}$
						\item \pcif $i < \numberDecQueries+1$
							\item \quad Run $\Ad{B}^{\RO G', \oracleDecryptSim}$(\pk) until $i$-th query $c_i$ to \oracleDecryptSim
							\item \quad $m := \ListQueriesG^{-1}(c_i)$
						\item \pcelse 
							\item \quad $m \leftarrow \Ad{B}^{\RO G', \oracleDecryptSim}(\pk)$
						\item \pcreturn $m$
					\end{nicodemus}
				\end{minipage}
			
			}		
		}
	\caption{Simulation \oracleDecryptSim of oracle \oracleDecrypt for \PKEDerand, which is defined analogously to \oracleDecapsSim (see Figure \ref{fig:SimDecaps}),
				and \FFPCPA adversary $\tilde{\Ad{B}}$. For the reader's convenience, we repeat the definition of \RO{G'}. }
			\label{fig:SimDecrypt}
	\end{center}
	\end{figure}

	Using the same chain of inequalities as in the proof of \cref{thm:FFP1}, and again using Lemma \ref{lem:pr-guess}, we obtain
	\begin{equation}\label{eq:smalladvB}
			\Adv^{\FFPCCA}_{\PKEDerand}(\Ad{B})\le \Pr\left[\Ad{B}\text{ wins }\wedge \neg\EventDecapsSimDiffers\right]+\Pr\left[\EventDecapsSimDiffers\wedge\neg\EventGuessedCT\right]+\numberDecQueries\cdot 2^{-\gamma}.
	\end{equation}

	Adversary $\tilde{\Ad{B}}$ perfectly simulates game \FFPCCA unless \EventDecapsSimDiffers occurs, and wins with probability \nicefrac{1}{\numberDecQueries+1} if \Ad{B} wins by returning a failing plaintext or if \Ad{B} issues a decryption query that triggers \EventDecapsSimDiffers but not \EventGuessedCT.
	
	\begin{equation}\label{eq:largeadvBpr}
		\Adv^{\FFPCPA}_{\PKEDerand}\left(\tilde{\Ad{B}}\right)
			= \frac{1}{\numberDecQueries+1} \cdot \left(
				\Pr\left[\Ad{B}\text{ wins }\wedge \neg\EventDecapsSimDiffers\right] + \Pr\left[\EventDecapsSimDiffers\wedge\neg\EventGuessedCT\right]\right)
	\end{equation}
	
	Combining Equations \eqref{eq:smalladvB} and \eqref{eq:largeadvBpr} yields the desired bound.
\qed
\end{proof}

\ifTightOnSpace\else
Combining Theorems \ref{thm:FFP1} and \ref{thm:FFP2}, we obtain the following straightforwardly.
\begin{corollary}[$\FOExplicitMess\lbrack\PKE\rbrack$ \INDCPA and \PKEDerand \FFPCPA \RightarrowROM $\FOExplicitMess\lbrack\PKE\rbrack$ $\INDCCA$]\label{cor:FFP3}
	\ifTightOnSpace
		Let \PKE  and \Ad{A} be as in \cref{thm:FFP1}.%
	\else
		Let \PKE be $\gamma$-spread, and let $\KemExplicitMess  \coloneqq \FOExplicitMess[\PKE, \RO{G},\RO H]$. 
		Let \Ad{A} be an \INDCCAKEM adversary (in the ROM) against \KemExplicitMess,
		issuing at most $q_\RO{G}$ many queries to its oracle \RO{G}, $q_\RO{H}$ many queries to its oracle \RO{H}, and at most \numberDecQueries many queries to its decapsulation oracle \oracleDecaps.
	\fi
	Then there exist an \INDCPAKEM adversary $\tilde{\Ad{A}}$ and an \FFPCPA adversary \Ad{B} such that
	\begin{equation}\label{eq:FFP3}
		\Adv^{\INDCCAKEM}_{\KemExplicitMess}(\Ad{A})\le \Adv^{\INDCPAKEM}_{\KemExplicitMess}\left(\tilde{\Ad{A}}\right)
			+(\numberDecQueries+1) \cdot\Adv^{\FFPCPA}_{\PKEDerand}\left(\Ad{B}\right)+2\numberDecQueries\cdot 2^{-\gamma}
			\enspace .
	\end{equation}
	Adversary $\tilde{\Ad A}$ makes $q_\RO{G}$ queries to \RO{G} and $q_\RO{H}+\numberDecapsQueries$ queries to \RO{H}, adversary \Ad{B} makes $q_\RO{G}$ queries to \RO{G}, and both run in about the time of \Ad{A}.
\end{corollary}
We remark that the factor $2$ in front of the additive term $\numberDecQueries\cdot 2^{-\gamma}$ is an artefact of our modular proof (in terms of Theorems \ref{thm:FFP1} and \ref{thm:FFP2}). It is straightforward to show that the bound of \cref{cor:FFP3} can be proven without the factor of $2$, when directly analyzing the composition of the reductions from Theorems \ref{thm:FFP1} and \ref{thm:FFP2}.
\fi

\ifTightOnSpace
\else
	Next, we observe in \cref{thm:PKEpassiveToINDCPAKEM:ROM} that \INDCPA security of 
	\ifTightOnSpace \KemExplicitMess \else $\FOExplicitMess[\PKE, \RO{G},\RO H]$ \fi
	can be based on passive security of \PKE. While \cref{thm:PKEpassiveToINDCPAKEM:ROM} is implicitly contained in \cite{TCC:HofHovKil17},
	we make explicit in \ifCameraReady the full version \else \cref{sec:INDCPAHHK} \fi how it can be easily obtained.
	
	\begin{restatable}[\PKE \OWCPA or \INDCPA \RightarrowROM $\FOExplicitMess\lbrack\PKE\rbrack$ \INDCPA] {theorem}{INDCPAHHK}\label{thm:PKEpassiveToINDCPAKEM:ROM}
		Let $\KemExplicitMess \coloneqq \FOExplicitMess[\PKE, \RO{G},\RO H]$ for some PKE scheme \PKE.
		For any \INDCPA adversary \Ad{A} against $\KemExplicitMess$, issuing at most 
		$q_\RO{G}$ many queries to its oracle \RO{G} and $q_\RO{H}$ many queries to its oracle \RO{H},
		there exist an \OWCPA adversary $\Ad{B_{\OWCPA}}$ and an \INDCPA adversary $\Ad{B_{\INDCPA}}$ of roughly the same running time such that
		\[\Adv^{\INDCPA}_{\KemExplicitMess}(\Ad{A}) \leq (q_\RO{G} + q_\RO{H} +1 ) \cdot \Adv^{\OW}_{\PKE}(\Ad{B_{\OWCPA}})\ifTightOnSpace \text{ and}\fi  \]
		and 
		\[\Adv^{\INDCPA}_{\KemExplicitMess}(\Ad{A}) \leq 3 \cdot \Adv^{\INDCPA}_{\PKE}(\Ad{B_{\INDCPA}}) 
		+ \frac{2\cdot(q_\RO{G} + q_\RO{H} )+1}{|\MSpace|} \enspace .
		\]
	\end{restatable}
	
	Combining \cref{cor:FFP3} and \cref{thm:PKEpassiveToINDCPAKEM:ROM}, we obtain the following straightforward
\fi

\begin{corollary}[\PKE \OWCPA or \INDCPA and \PKEDerand \FFPCPA \RightarrowROM $\FOExplicitMess\lbrack\PKE\rbrack$ $\INDCCA$]\label{cor:PKEpassiveAndFFPCPAtoINDCPAKEM:ROM}
	\ifTightOnSpace
\else
		Let \PKE be a (randomized) \PKE scheme that is $\gamma$-spread, and let $\KemExplicitMess \coloneqq \FOExplicitMess[\PKE, \RO{G},\RO H]$.
		Let \Ad{A} be an \INDCCAKEM adversary (in the ROM) against \KemExplicitMess, making at most $q_\RO{RO}$ many queries to its random oracles \RO{G} and \RO{H}, and \numberDecQueries many queries to its decapsulation oracle \oracleDecaps.
	\fi
	Then there exist a \OWCPA adversary $\Ad{B_{\OWCPA}}$ and an \INDCPA adversary $\Ad{B_{\INDCPA}}$ such that
	\begin{align*}
		\Adv^{\INDCCAKEM}_{\KemExplicitMess}(\Ad{A})
			\le & (q_\RO{RO} + \numberDecQueries + 1) \cdot \Adv^{\OW}_{\PKE}(\Ad{B_{\OWCPA}})
		\\
		&\quad\ +(\numberDecQueries+1) \cdot\Adv^{\FFPCPA}_{\PKEDerand}\left(\Ad{C}\right)+2\numberDecQueries\cdot 2^{-\gamma}
	\end{align*}
	and
	\begin{align*}
		\Adv^{\INDCCAKEM}_{\KemExplicitMess}(\Ad{A})
		& \le  3 \cdot \Adv^{\INDCPA}_{\PKE}(\Ad{B_{\INDCPA}}) + \frac{2\cdot(q_\RO{RO} + \numberDecQueries )+1}{|\MSpace|} \\
		& +(\numberDecQueries+1) \cdot\Adv^{\FFPCPA}_{\PKEDerand}\left(\Ad{B}\right)+2\numberDecQueries\cdot 2^{-\gamma}.
	\end{align*}
	\ifTightOnSpace\else Adversary \fi \Ad{C} makes $q_\RO{G}$ queries to \RO{G}, and all adversaries run in about the time of \Ad{A}.
\end{corollary}

When comparing our bounds with the respective bounds from \cite{TCC:HofHovKil17}, we note that our bounds are still in the same asymptotic ball park and differ from the bounds in \cite{TCC:HofHovKil17} essentially by replacing the worst-case correctness term \deltaWorstCase (there denoted by $\delta$) present in \cite{TCC:HofHovKil17} by $\Adv^{\FFPCPA}_{\PKEDerand}\left(\Ad{B}\right)$, and having an additional term in $\gamma$ even for \KemImplicitMess. We believe that the additional $\gamma$-term could be removed by doing a direct proof for \KemImplicitMess, but redoing the whole proof for this variant was outside the scope of this work.
We will further analyze $\Adv^{\FFPCPA}_{\PKEDerand}\left(\Ad{B}\right)$ in \cref{sec:PKEDerand:FFPCPA}.

\begin{remark}[Obtaining the results for $\FOImplicitMess\lbrack\PKE\rbrack$]\label{remark:FOexplicitToFOImplicit}
	We can use the results from \cite{TCC:BHHHP19} to furthermore show that the bounds given in \cref{cor:PKEpassiveAndFFPCPAtoINDCPAKEM:ROM} also hold if $\KemExplicitMess \coloneqq \FOExplicitMess[\PKE, \RO{G},\RO H]$ is replaced with $\KemImplicitMess \coloneqq \FOExplicitMess[\PKE, \RO{G},\RO H]$:
	In more detail, it follows directly from \cite[Theorem 3]{TCC:BHHHP19} that
	for any \INDCCAKEM attacker \Ad{A} against \KemImplicitMess,
	there exists an \INDCCAKEM attacker \Ad{B} against \KemExplicitMess such that
	\ifTightOnSpace
\else
	\[\Adv^{\INDCCAKEM}_{\KemImplicitMess}(\Ad{A}) \leq  \Adv^{\INDCCAKEM}_{\KemExplicitMess}(\Ad{B}) \enspace , \]
	\fi
	and \cref{cor:PKEpassiveAndFFPCPAtoINDCPAKEM:ROM} does not contain any terms relative to \KemExplicitMess itself,
	it only contains terms relative to the underlying schemes \PKE and \PKEDerand.
\end{remark} 
\section{Compressed oracles and extraction}\label{sec:compressed-prels} \label{sec:prels:compressed}

We want to generalize the ROM results obtained in \cref{sec:ROM} to the QROM.
To this end, we will use an extension of the compressed oracle technique~\cite{C:Zhandry19} that was introduced in \cite{DFMS21}
\ifTightOnSpace\else and that we will now quickly recap. \fi
To describe the technique, we start with the observation that for each input value $x$, its oracle value $\RO{O}(x)$ is a uniformly distributed random variable that can equivalently be sampled by measuring a uniform superposition in the computational basis.
It was shown in~\cite{C:Zhandry19} how a quantum-accessible random oracle $\RO{O}: X \rightarrow Y$ can %
be simulated by preparing a database $D$ with an entry $D_x$ for each input value $x$, with each $D_x$ being initialized as a uniform superposition of all elements of $Y$, and omitting the ``oracle-generating'' measurements until after the algorithm accessing \RO{O} has finished.
In \cite{DFMS21}, this oracle simulation was generalized to obtain %
an \emph{extractable} oracle simulator \SupOr (for \underline{e}xtractable \underline{C}ompressed \underline{O}racle) that has two interfaces, the random oracle interface \SRO and an %
 extraction interface $\SE_{f}$, defined relative to a function $f: X \times Y \rightarrow T$. %
 Informally, 
$\SE_{f}$ takes as input a classical value $t$. Consider the classical procedure of going through a lexicographically ordered list of lazy-sampled input output pairs $(x,y)$ and outputting the first one such that %
	$f(x, y) = t$. %
	$\SE_{f}$ performs the quantum analogue of that: a measurement that partially collapses the oracle database, just enough so that the classical procedure would yield one particular outcome  $x$ for all parts of the superposition.
	After the measurement, $D$ is thus in a state such that the superposition held in database entry $D_x$ only contains possibilities $y$ for $\SRO(x)$ such that $f(x,y) = t$, and no entry $D_{x'}$ for any $x'< x$ will have any possibilities $y'$ left such that also $f(x',y') = t$. 
		Whenever it is clear from context which function $f$ is used, we simply write \SE instead of $\SE_{f}$.

In general, $\SE_{f}$ can extract preimage entries from the %
``database'' $D$ during the runtime of an adversary instead of only after the adversary terminated. This allows for adaptive behaviour of a reduction, based on an adversary's queries.
In \cite{DFMS21}, it was already used for the same purpose we need it for -- the simulation of a decapsulation oracle, by having \SE extract a preimage plaintext from the ciphertext on which the decapsulation oracle was queried. 
We will denote oracles %
 modelled as \underline{e}xtractable \underline{q}uantum-accessible \underline{RO}s by \augQRO{f},
and a proof that uses an \augQRO{f} will be called \emph{a proof in the \augQROM{f}}.

We will now make this description more formal, closely following notation and conventions from \cite{DFMS21}.
Like in \cite{DFMS21}, we \ifTightOnSpace\else keep the formalism as simple as possible by describing \fi an inefficient variant of the oracle that is not (yet) ``compressed''.
Efficient simulation %
is possible via a standard sparse encoding, see %
\cite[Appendix A]{DFMS21}.
The simulator \SupOr for a random function $\mathsf O:\{0,1\}^m\to\{0,1\}^n$  is a stateful oracle with a state stored in a quantum register $D=D_{0^m}\ldots D_{1^m}$, where for each \ifTightOnSpace \else input value \fi $x \in \{0,1\}^m$, register $D_x$ has $n+1$ qubits used to store superpositions of $n$-bit output strings $y$, encoded as $0y$, and an additional symbol $\bot$, encoded as $10^n$. We adopt the convention that an operator expecting $n$ input qubits acts on the last $n$ qubits when applied to \ifTightOnSpace \else one of the registers \fi $D_x$. The compressed oracle %
has the following three components.
\begin{itemize}
	\item The initial state of the oracle, $\ket{\phi}=\ket\bot^{2^m}$
	\item A quantum query with query input register $X$ and output register $Y$ is answered using the oracle unitary $O_{XYD}$ defined by
		\begin{equation}
			O_{XYD}\ket x_X=\ket x_X\otimes \left(F_{D_x}\CNOT^{\otimes n}_{D_x:Y}F_{D_x}\right),
		\end{equation}
		where $F\ket\bot=\ket{\phi_0}$, $F\ket{\phi_0}=\ket \bot $ and $F\ket\psi=\ket\psi$ for all $\ket\psi$ such that $\bracket{\psi}{\bot}=\bracket\psi{\phi_0}=0$,
		with $\ket{\phi_0}=\ket +^{\otimes n}$ being the uniform superposition. The CNOT operator here is responsible for XORing the function value (stored in $D_x$, now in superposition) into the query algorithm's output register.
	\item A \emph{recovery algorithm} that recovers a standard QRO $\RO O$:	%
	apply $F^{\otimes 2^m}$ to $D$ and measure it to obtain the function table of $\mathsf O$.
\end{itemize}

\ifTightOnSpace \else 
	In section \ref{sec:QROM:OWTH}, we will use the superposition oracle to analyze algorithms that make parallel (quantum) queries to a random oracle. For a standard quantum oracle for a function \RO{H}, an algorithm that makes $w$ parallel queries sends $2w$ quantum regisers $X_i, Y_i$, $i=1,...,w$ to the oracle. The query is then processed by applying the oracle unitary $U_H$ to each pair $X_i, Y_i$. We can think of this parallel-query oracle as being implemented by a simulator with query access to the non-parallel oracle for \RO{H}:
	upon input regisers $X_i, Y_i$, $i=1,...,w$ the simulator sends the register pairs $X_i, Y_i$ to its own oracle sequentially. Using this trivial reformulation, it is clear how parallel queries can be handled when \RO{H} is a random function and the oracle for \RO{H} is simulated using the compressed oracle. 
\fi

We now make our description of the extraction interface \SE formal: Given a random oracle $\mathsf O: \{0,1\}^m\to\{0,1\}^n$, let $f: \{0,1\}^m\times\{0,1\}^n\to \{0,1\}^\ell$ be a function. We define a family of measurements $\left(\mathcal M^t\right)_{t\in\{0,1\}^\ell}$. The measurement $\mathcal M^t$ has measurement projectors $\{\Sigma^{t,x}\}_{x\in\{0,1\}^m\cup\{\emptyset\}}$ defined as follows. For $x\in\{0,1\}^m$, the projector selects the case where $D_x$ is the first (in lexicographical order) register that contains $y$ such that $f(x,y)=t$, i.e.
\begin{equation}\label{eq:extraction-measurement}
	\Sigma^{t,x}=\bigotimes_{x'<x}\bar\Pi^{t,x'}_{D_x'}\otimes \Pi^{t,x}_{D_x},\  \text{ with }\ \ 
	\Pi^{t,x}=\sum_{\substack{y\in\{0,1\}^n:\\ f(x,y)=t}}\proj{y}
\end{equation}
and $\bar\Pi=\mathds{1}-\Pi$. \ifTightOnSpace\else The remaining projector corresponds to \fi the case where no register contains such a $y$, i.e.
\begin{equation}
	\Sigma^{t,\emptyset}=\bigotimes_{x'\in\{0,1\}^m}\bar\Pi^{t,x'}_{D_x'}.
\end{equation}
As an example, say we model a random oracle \RO{H} as such an \augQRO{f}. Using $f(x,y) := \bool{\RO{H}(x) = y}$,
$M^1$ allows us to extract a preimage of $y$.

\SupOr %
is initialized with the inital state of the compressed oracle. \SRO is %
quantum-accessible %
and applies the compressed oracle query unitary $O_{XYD}$.
\SE is \ifTightOnSpace\else a classical oracle interface that, on input $t$, \fi applies $\mathcal M^t$ to \SupOr's internal state \ifTightOnSpace \else (i.e. the state of the compressed oracle) \fi and returns the result. \ifTightOnSpace\else The simulator \SupOr has several \fi useful properties that were characterized in \cite[Theorem 3.4]{DFMS21}, for convenience included \ifTightOnSpace\else below\fi. These characterisations are in terms of the quantity 
\ifTightOnSpace
\else 
\begin{align}
	\Gamma(f)=\max_{t}\Gamma_{R_{f,t}}\text{, with}\nonumber\\ 
	R_{f,t}(x,y):\Leftrightarrow f(x,y)= t\text{ and}\nonumber\\
	\Gamma_R := \max_{x} | \{y \mid R(x,y) \} | .\label{eq:Gamma_R}
\end{align} 
\fi
For $f=\Encrypt(\cdot;\cdot)$, %
the encryption function of a PKE that takes as \ifTightOnSpace\else first input a message $m$ and as second input an encryption randomness $r$\fi, we have $\Gamma(f)=2^{-\gamma}|\RSpace|$ if \PKE is $\gamma$-spread. In this case, $\SE(c)$ outputs a plaintext $m$ such that $\Encrypt(m, \SRO(m))=c$, or $\bot$ if the ciphertext $c$ has not been computed using $\SRO$ before.
\ifCameraReady \else %
	\ifTightOnSpace
\else 
\ifTightOnSpace %
	\section{More details about the extractable QRO simulator \SupOr} \label{sec:eCO}
	In this section we include some more details about the extractable QRO simulator \SupOr for the reader's convenience.	

	In section \ref{sec:QROM:OWTH}, we will use the superposition oracle to analyze algorithms that make parallel (quantum) queries to a random oracle. For a standard quantum oracle for a function \RO{H}, an algorithm that makes $w$ parallel queries sends $2w$ quantum regisers $X_i, Y_i$, $i=1,...,w$ to the oracle. The query is then processed by applying the oracle unitary $U_H$ to each pair $X_i, Y_i$. We can think of this parallel-query oracle as being implemented by a simulator with query access to the non-parallel oracle for \RO{H}:
	upon input regisers $X_i, Y_i$, $i=1,...,w$ the simulator sends the register pairs $X_i, Y_i$ to its own oracle sequentially. Using this trivial reformulation, it is clear how parallel queries can be handled when \RO{H} is a random function and the oracle for \RO{H} is simulated using the compressed oracle. 
\fi 

We now state the parts of \cite[Theorem 3.4]{DFMS21} that we will use in our proofs.
\begin{lemma}[Part of theorem 3.4 in \cite{DFMS21}]\label{lem:extractable-sim-properties}
	The extractable RO simulator \SupOr described above, with interfaces \SRO and \SE, satisfies the following properties. 
	\begin{itemize}\vspace{-1ex}\setlength{\parskip}{0.5ex}
		\item[1.\!] If $\SE$ is unused, $\SupOr$ is perfectly indistinguishable from a random oracle.  \\[-2ex]
		\item[2.a] Any two subsequent independent queries to $\SRO$ %
		commute. In particular,  two subsequent {\em classical} $\SRO$-queries with the same input $x$ give identical responses. 
		\item[2.b] Any two subsequent independent queries to  $\SE$ %
		commute. In particular,  two subsequent %
		$\SE$-queries with the same input $t$ give identical responses. 
		\item[2.c] Any two subsequent independent queries to $\SE$ and $\SRO$ %
		$8\sqrt{2\Gamma(f)/2^n}$-almost-commute. \\[-2ex]
	\end{itemize}
	Furthermore, the total runtime and quantum memory footprint of $\SupOr$, when using the sparse representation of the compressed oracle, are bounded as
	\begin{align*}
		\Time(\SupOr,q_{RO}, q_E)&= O\bigl(q_{RO} \cdot q_E\cdot \mathrm{Time}[f] + q_{RO}^2\bigr),\text{ and }\\
		\QMem(\SupOr,q_{RO}, q_E)&=O\bigl( q_{RO}\bigr)  .
	\end{align*}
	\ifTightOnSpace
	where $q_E(q_{RO})$ is the number of queries to $\SE(\SRO)$%
	.
	\else
	where $q_E$ and $q_{RO}$ are the number of queries to $\SE$ and $\SRO$, respectively%
	.
	\fi
\end{lemma} 	\fi
\fi

\tinka{I think this might go to the intro:
\sloppy For a fixed function $f$, we can now formally consider an \augmentedORextended quantum-accessible random oracle model \augQROM{f}, where the random oracle is instantiated with the \SRO interface of the simulator \SupOr, and where adversaries in addition have classical query access to \SE. While this is not a realistic model of a cryptographic hash function, it can serve as a useful model for \emph{intermediate} results.  In particular, we will use this model in the following way. We want to base the QROM-\CCA-security of $\FO[\PKE]$ on quantum problems defined for $\PKE$. To that end, we simulate the decapsulation oracle using $\SE$-queries. This simplifies things (no decaps oracle), but at the same time complicates things (the adversary is now in the $\augQROM{\Encrypt}$). We then further reduce the \augQROM{\Encrypt}-\CPA-security  to security properties of the underlying scheme $\PKE$. Just as done in reductions that use the (regular) QROM, the $\augQROM{\Encrypt}$ can then be simulated by the reduction. Among the security properties of the underlying scheme \PKE that we need is the hardness of finding \emph{non-generic} failing plaintext randomness pairs of \PKE, i.e. finding failing plaintext randomness pairs by making non-trivial use of the public key. This is the security property that replaces the decryption failure probability that features in existing security bounds for the FO transformation.
}\cm{It might actually help to have this in the intro for the full version. Where could we put it?}

\section{QROM reduction}\label{sec:QROM}

In this section, we generalize the reductions from \cref{sec:ROM} to the \ifTightOnSpace\else quantum-accessible random oracle model\fi.
To do so, we give in \cref{fig:QuantumSimDecaps} the quantum analogues of the simulated decapsulation oracles \oracleDecapsSim and \oracleDecapsSimFail from \cref{fig:SimDecaps}, which were (essentially) developed in \cite{DFMS21}.
We have to adapt our simulations since the ROM simulations from \cref{fig:SimDecaps} use book-keeping techniques and therefore cannot be easily implemented in the standard QROM. Instead, we use the formalism described in \cref{sec:compressed-prels}, i.e., we use a simulation of a quantum-accessible random oracle and \emph{make use of the additional extraction interface} \SE:
While the simulations in \cref{fig:SimDecaps} had access to a list \ListQueriesG that could be used to extract potential ciphertext preimages, the simulations in \cref{fig:QuantumSimDecaps} can now extract them by accessing extractor \SE (see lines~\ref{line:QuantumSimDecaps:Extract1} and~\ref{line:QuantumSimDecaps:Extract2}).
The rest of the simulation is exactly as before.
Using the notation from \cref{sec:compressed-prels}, we denote the modelling of the ROM as extractable by \augQROM{\Encrypt},
as we extract preimages relative to function $f=\Encrypt(\pk, \cdot, \cdot)$, with the message being $f$'s first and the randomness being $f$'s second input.

	\begin{figure}[tb] \begin{center} \fbox{\small
			
			\nicoresetlinenr
			
			\begin{minipage}[t]{4.0cm}	
				
				\underline{{\bf Game} \FFPATK}
				\begin{nicodemus}
					\item $(\pk, \sk) \leftarrow \KG$
					\item $m \leftarrow \Ad{A}^{\textsc{O}_\atk, \SupOr}(\pk)$
					\item $c := \Encrypt(\pk,m)$
					\item $m' := \Decrypt(\sk,c)$
					\item \pcreturn $\bool{m' \neq m}$
				\end{nicodemus}
				
			\end{minipage}
			
			\quad
			
			\begin{minipage}[t]{4.0cm}
				\underline{$\oracleDecrypt(c)$}
				\begin{nicodemus}
					\item $m := \Decrypt(\sk,c)$
					\item \pcreturn $m$
				\end{nicodemus}
			\end{minipage}
			
		}
	\end{center}
	\caption{
		Games \FFPATK for a deterministic \PKE, where $\atk \in \lbrace \CPA, \CCA \rbrace$, in the \augQROM{f}.
		Like in its classical counterpart (see \cref{fig:def:FFPATK}, page~\pageref{fig:def:FFPATK}),
		$\mathsf{O}_\atk$ is the decryption oracle present in the respective \INDATK game (see \cref{def:FFPATK} on page~\pageref{def:FFPATK}).
		The only difference is that the random oracle \RO{G} is now modelled as an extractable superposition oracle \SupOr. 
	}
	\label{fig:def:FFPATK:QROM}
\end{figure}

While \cref{sec:ROM} concluded by showing in \cref{thm:PKEpassiveToINDCPAKEM:ROM} how to base \INDCPAKEM security of
\ifTightOnSpace\else $\FOExplicitMess[\PKE, \RO{G},\RO{H}]$ \fi
on passive security of \PKE in the ROM, we need to develop an additional tool to do the same in the \augQROM{\Encrypt}.
Therefore, we split this section as follows:
\cref{sec:QROM:CCA-to-CPA-KEM} ends with
\ifTightOnSpace\else \INDCCA security of $\FOExplicitMess[\PKE, \RO{G},\RO{H}]$ being based on \fi
\INDCPA security of
\ifTightOnSpace\else $\FOExplicitMess[\PKE, \RO{G},\RO{H}]$ \fi and \FFPCPA security of \PKEDerand%
\ifTightOnSpace\else.
Note that the notions on which we base \INDCCA security are now in the \augQROM{\Encrypt}.\fi
We give the \augQROM{f} definition of \FFPATK in \cref{fig:def:FFPATK:QROM}.
\cref{sec:QROM:OWTH} develops the necessary \augQROM{\Encrypt} tools to further analyze \INDCPA security of \ifTightOnSpace\else $\FOExplicitMess[\PKE, \RO{G},\RO{H}]$\fi.
Concretely, \cref{sec:QROM:OWTH} provides an \augQROM{\Encrypt}-compatible variant of the one-way to hiding (\OWTH) lemma for semi-classical oracles as introduced in \cite{C:AmbHamUnr19}.
\ifTightOnSpace\else
	Intuitively, the \augQROM{\Encrypt}-\OWTH lemma states that input depending on particular random oracle values $\SRO(x)$ (like, e.g., $\RO{G}(m^*)$) can be replaced with input that replaced all involved oracle values with fresh uniform randomness. The change goes unnoticed unless one of the $x$ can be detected in the oracle queries. 
	\cref{sec:QROM:OWTH} is given in a general way and might prove to be of independent interest.
\fi 
Equipped with the results from \cref{sec:QROM:OWTH}, we show in \cref{sec:QROM:CPA-to-passive} that \emph{also in the \augQROM{\Encrypt}}, \INDCPA security of $\FOExplicitMess[\PKE, \RO{G},\RO{H}]$ can be based on passive security of \PKE. 

\subsection{From $\INDCPA_{\FO[\PKE]}$ and $\FFPCCA_\PKEDerand$ to $\INDCCA_{\FO[\PKE]}$} \label{sec:QROM:CCA-to-CPA-KEM}

\begin{figure}[tb]\begin{center}\makebox[\textwidth][c]{
		
		\nicoresetlinenr
		\resizebox{\textwidth}{!}{
		\fbox{\small
			
			\begin{minipage}[t]{4.3cm}	
				
				\underline{$\oracleDecaps(c\neq c^*)$}
				\begin{nicodemus}
					\item $m' := \Decrypt(\sk,c)$
					\item \pcif $m'=\bot$
					\item \quad $\pcreturn K {\coloneqq} \bot$
					\item \pcelse
					\item \quad  $c' := \Encrypt(\pk, m'; \RO{G}(m'))$
					\item \quad \pcif $c \neq c'$
					\item \quad\quad \pcreturn $\bot$
					\item \quad \pcelse
					\item \quad\quad \pcreturn $\RO{H}(m')$  
				\end{nicodemus}
				
				\ \\
				
				\underline{$\RO G'$, input registers $X,Y$}
				\begin{nicodemus}
					\item Apply $\SRO_{XYD}$
					\item \pcreturn registers $XY$
				\end{nicodemus}
				
			\end{minipage}
			\;
			
			\begin{minipage}[t]{3.1cm}	
				
				\underline{$\oracleDecapsSim(c\neq c^*)$}
				\begin{nicodemus}
					\item $m\leftarrow \SE(c)$ \label{line:QuantumSimDecaps:Extract1}
					\item \pcif $m=\bot$
					\item \quad $\pcreturn \bot$
					\item \pcelse
					\item \quad \pcreturn $\RO{H}(m)$
				\end{nicodemus}

			\end{minipage}
		
			\;
			
			\begin{minipage}[t]{4.5cm}	
				
				\underline{$\oracleDecapsSimFail(c\neq c^*)$}
				\begin{nicodemus}
					\item $m\leftarrow \SE(c)$ \label{line:QuantumSimDecaps:Extract2}
					\item $m' \coloneqq \oracleDecrypt(c)$\label{line:QuantumSimDecaps:Decaps}
					\item \pcif $m \neq \bot \pcand m \neq m'$
					\item \quad $\ListFail \coloneqq \ListFail \cup \{ m \}$ 
					\item \pcif $m=\bot$
						\item \quad $\pcreturn \bot$
					\item \pcelse
						\item \quad \pcreturn $\RO{H}(m)$
				\end{nicodemus}
					
			\ifTightOnSpace
			\else
				\ \\
				
				\underline{$\oracleDecrypt(c)$}
				\begin{nicodemus}
					\item $m' \coloneqq \Decrypt(\sk, c)$
					\item \pcif $m' = \bot$
					\item \quad \pcreturn $\bot$
					\item \pcelse 
						\item \quad  \pcif $\Encrypt(\pk, m'; \RO{G}(m')) \neq c$
							\item \quad \quad \pcreturn $\bot$
					\item \quad \pcelse
					\item \quad \quad  \pcreturn $m'$
				\end{nicodemus}
			\fi
			\end{minipage}
			
		}	}   }
		\caption{Simulated and failing-plaintext-extracting versions of the decapsulation oracle \oracleDecaps for  $\FOExplicitMess[\PKE, \RO{G}, \RO{H}]$, using the extractable QRO simulator \SupOr from \cite{DFMS21} (see \cref{sec:compressed-prels}). The simulations of \oracleDecaps are exactly like the ROM ones in \cref{fig:SimDecaps} except for how they extract ciphertext preimages in lines \ref{line:QuantumSimDecaps:Extract1} and \ref{line:QuantumSimDecaps:Extract2}. We assume \SupOr to be freshly initialized before \oracleDecapsSim or \oracleDecapsSimFail is used for the first time in a security game, and extraction interface \SE is defined with respect to function $f =\Encrypt(\pk, \cdot; \cdot)$\ifTightOnSpace\else, where $\Encrypt$ is the encryption algorithm of \PKE.\fi %
		}
		\label{fig:QuantumSimDecaps}
	\end{center}
\end{figure}

We begin by proving a quantum analogue of \cref{thm:FFP1}.
\begin{theorem}[{\small $\FOExplicitMess\lbrack\PKE\rbrack$ \INDCPA and \PKEDerand \FFPCCA \RightarrowaugQROM{\Encrypt} $\FOExplicitMess\lbrack\PKE\rbrack$ $\INDCCA$}]%
\label{thm:QFFP1}\label{thm:INDandFFPtoCCA:QROM}%
Let $\PKE$ be a (randomized) \PKE that is $\gamma$-spread, and $\KemExplicitMess \coloneqq \FOExplicitMess[\PKE, \RO{G},\RO H]$.
	Let \Ad{A} be an \INDCCAKEM-adversary (in the QROM) against \KemExplicitMess, making at most \ifTightOnSpace\else \numberDecapsQueries many queries to its decapsulation oracle \oracleDecaps, and making $q_{\RO G}$, $q_{\RO H}$ queries to its respective random oracles .\fi Let furthermore $d$ and $w$ be the combined query depth and query width of \Ad A's random oracle queries.
	Then there exist an \INDCPAKEM adversary $\tilde{\Ad{A}}$ and an \FFPCCA adversary \Ad{B} against \PKEDerand, both in the $\augQROM{\Encrypt}$, such that
	\ifTightOnSpace
\else
	\begin{align}
		\Adv^{\INDCCAKEM}_{\KemExplicitMess}(\Ad{A})\le &\Adv^{\INDCPAKEM}_{\KemExplicitMess}\left(\tilde{\Ad{A}}\right)+\Adv^{\FFPCCA}_{\PKE^{\RO G}}\left(\mathcal B\right)\nonumber\\
		&+12 \numberDecapsQueries(q_{\RO G}+4\numberDecapsQueries)\cdot 2^{-\gamma/2}%
		.\label{eq:QFFP1}
	\end{align}
\fi
	The adversary $\tilde{\Ad{A}}$ makes $q_{\RO G}+q_{\RO H}+\numberDecapsQueries$ queries to $\SRO$ with a combined depth of $d+\numberDecapsQueries$ and a combined width of $w$,  and \numberDecapsQueries queries to \SE. Here, \SRO simulates $\RO G\times \RO H$. The adversary $\mathcal B$ makes \numberDecapsQueries many queries to \oracleDecrypt and \SE and $q_{\RO G}$ queries to $\SRO$%
	,
	and neither $\tilde{\Ad{A}}$ nor $\mathcal B$ query \SE on the challenge ciphertext. The running times of the adversaries $\tilde{\Ad A}$ and $\Ad B$ are bounded as $\Time (\tilde{\Ad A})=\Time (\Ad A)+O(\numberDecapsQueries)$ and $\Time (B)=\Time (\Ad A)+O(\numberDecapsQueries)$.

\end{theorem}

Before proving the theorem, we \ifTightOnSpace \else briefly \fi point out similarities and differences to the ROM counterpart, \cref{thm:FFP1}.  First note that the bounds look very similar. The only difference lies in the additive error term that depends on the spreadness parameter $\gamma$. In the above theorem, this additive error term $O(\numberDecapsQueries q_\RO{G} 2^{-\gamma/2})$ is much larger than the term $O(\numberDecapsQueries 2^{-\gamma})$ present in \cref{thm:FFP1}.
\ifTightOnSpace\else This larger additive loss \fi originates from dealing with the fact that the extraction technique used to simulate the \Decaps oracle inflicts an error onto the simulation of the QRO.
We expect that for many real-world schemes, the additive security loss of $O(\numberDecapsQueries q_\RO{G} 2^{-\gamma/2})$ is still small enough to be neglected, and calculate the term for two example cases in \cref{sec:spreadness}.
Another important difference between \cref{thm:QFFP1} and \cref{thm:FFP1} is of course that the adversaries $\tilde{\Ad{A}}$ and \Ad B are now in the non-standard $\augQROM{\Encrypt}$. Looking ahead, we provide further reductions \ifTightOnSpace\else in \cref{sec:QROM:CPA-to-passive} \fi culminating in \cref{cor:main-result} which gives a standard-QROM \ifTightOnSpace\else \INDCCAKEM security \fi bound for \KemExplicitMess in terms of (standard model) security properties of \PKE.

\begin{proof}
We prove this theorem via a number of hybrid games, drawing some inspiration from the reduction for the entire FO transformation given in \cite{DFMS21}.

\bfgame{0} is $\INDCCAKEM_{\KemExplicitMess}(\Ad{A})$.

\bfgame{1} is like \bfgame{0}, except for two modifications: The quantum-accessible random oracle \RO{G} is replaced by $\RO G'$ as defined in \cref{fig:QuantumSimDecaps}
\ifTightOnSpace\else (i.e., it is simulated using an \augQRO{\Encrypt}), \fi
and after the adversary has finished,
\ifTightOnSpace\else we compute oracle preimages for all ciphertexts on which \oracleDecaps was queried, i.e., \fi
we compute $\hat{m_i} \coloneqq \SE(c_i)$ for all $i=1,...,\numberDecapsQueries$, where $c_i$ is the input to the adversary's $i$th decapsulation query. By property 1 in \cite[Lem. 3.4]{DFMS21}/\cref{lem:extractable-sim-properties}, $\RO G'$ perfectly simulates \RO{G} until the first \SE-query, and since the first \SE-query occurs only after \Ad{A} finishes, we have
\begin{equation}
	\Adv^{\INDCCAKEM}_{\KemExplicitMess}(\Ad{A})=\Adv^\bfgame{0}=\Adv^\bfgame{1} \enspace .
\end{equation}

\bfgame{2} is like \bfgame{1}, except that $\hat{m_i} \coloneqq \SE(c_i)$ is computed right after \Ad{A} submits $c_i$ instead of computing it in the end. Note that \bfgame{2} can be obtained from  \bfgame{1} by first swapping the \SE call that produces $\hat m_1$ with all $\SRO$ calls that happen after the adversary submits $c_1$, including the calls inside \oracleDecaps, then continuing with the \SE-call that produces $\hat m_2$, etc. \ifTightOnSpace\else We will now use that \SRO and \SE almost-commute: \fi
By property 2.c \ifTightOnSpace\else and possibly 2.b) \fi of \cite[Lem. 3.4]{DFMS21}/\cref{lem:extractable-sim-properties} and since $\Gamma(\Encrypt(\cdot;\cdot))=2^{-\gamma}|\RSpace|$ for $\gamma$-spread PKE schemes, we have that
\begin{equation}
	\left|\Adv^\bfgame{1}-\Adv^\bfgame{2}\right|\le 8\sqrt 2%
	\numberDecapsQueries(q_{\RO G}+\numberDecapsQueries) \cdot 2^{-\gamma/2} \enspace .
\end{equation}

\bfgame{3} is the same as \bfgame{2}, except that \Ad{A} in run with access to the oracle \oracleDecapsSim instead of \oracleDecaps, meaning that upon a decapsulation query on $c_i$, \Ad{A} receives $\oracleDecapsSim(c_i)=\RO{H}(\hat m_i)$ instead of $\oracleDecaps(c_i) = \Decaps(\sk, c_i)$ (using the convention $\RO{H}(\bot) \coloneqq \bot$).
We still let the game also compute $\oracleDecaps(c_i)$, as $\oracleDecaps$ makes queries to $\SRO$ which can influence the behavior of \SE in subsequent queries. (Note that the reencryption step of \oracleDecaps triggers a call to \RO{G'}, which in turn uses \SRO.)
We define \Ad{B} exactly as in the proof of \cref{thm:FFP1}, except that it uses the oracles \RO{G'} and \oracleDecapsSimFail defined in \cref{fig:QuantumSimDecaps}:
\Ad{B} runs $\Ad{A}^{\RO G', \RO H, \oracleDecapsSimFail}$, using its own \FFPCCA oracle \oracleDecrypt to simulate \oracleDecapsSimFail and answering $\RO H$ queries by simulating a fresh compressed  oracle.\footnote{We remark that a $t$-wise independent function for sufficiently large $t=O(q_{\RO H}+\numberDecapsQueries)$ also suffices, which is more efficient as it doesn't require (nearly as much) quantum memory.} As soon as \oracleDecapsSimFail adds a plaintext $m$ to \ListFail, \Ad{B} aborts \Ad{A} and returns $m$. If \Ad{A} finishes and \ListFail is still empty, \Ad{B} returns $\bot$.

Let $\EventDecapsSimDiffers$ be the event that \Ad{A} makes a decryption query $c$ in \bfgame{2} such that $\oracleDecaps(c)\neq \oracleDecapsSim(c)$. 
\ifTightOnSpace
\else
Like in the respective proof step for \cref{thm:FFP1}, we bound
\begin{align*}
&	\frac 1 2 + \Adv^\bfgame{2} =  \Pr\left[\Ad{A} \textnormal{ wins in } \bfgame{2}\right]\\
	=& \Pr\left[\Ad{A} \textnormal{ wins in } \bfgame{2} \wedge \neg\EventDecapsSimDiffers\right]
		+\Pr\left[\Ad{A} \textnormal{ wins in } \bfgame{2}\wedge\EventDecapsSimDiffers\right]\\
	=& \Pr\left[\Ad{A} \textnormal{ wins in } \bfgame{3}\wedge\neg\EventDecapsSimDiffers\right]
		+\Pr\left[\Ad{A} \textnormal{ wins in } \bfgame{2}\wedge\EventDecapsSimDiffers\right]\\
	\le&\Pr\left[\Ad{A} \textnormal{ wins in } \bfgame{3}\right]+\Pr\left[\EventDecapsSimDiffers\right]\\
	=&\frac 1 2 +\Adv^\bfgame{3}+\Pr\left[\EventDecapsSimDiffers\right] \enspace .
\end{align*}
\fi

Again, event $\EventDecapsSimDiffers$ encompasses three cases: For some decapsulation query $c$,
\begin{itemize}
	\item[-]  Original decapsulation oracle $\oracleDecaps(c)$ rejects, but the simulation
		\ifTightOnSpace
\else
			$\oracleDecapsSim(c)$ does not, the latter meaning that $\oracleDecapsSim(c) = \RO{H}(\hat m_i)$
			for $\hat{m} \coloneqq \SE(c)$.
			By construction of the oracles this implies that while $\hat{m}$ encrypts to $c$, $c$ does not decrypt to $\hat{m}$ (under \PKEDerand, right after).
			(Otherwise, $\oracleDecaps(c)$ would not reject.) Hence, this case only occurs if $c$'s preimage $\hat{m}$ fails.			
		\fi	
	\item[-] Neither oracle rejects, but the return values differ, i.e., calling $\SE(c)$ in line \ref{line:QuantumSimDecaps:Extract1}
		yielded something different than $\Decrypt(\sk, c)$.
		Like above, this implies that preimage $\hat{m} \coloneqq \SE(c)$ fails 
	\item[-] $\oracleDecaps(c)$ does not reject, while $\oracleDecapsSim(c)$ does,
		i.e., $\hat{m} \coloneqq \SE(c)$ in line \ref{line:QuantumSimDecaps:Extract1} yielded $\bot$,
		but the re-encryption check inside the \oracleDecaps call in line \ref{line:QuantumSimDecaps:Decaps} checked out, 
		meaning that $\Encrypt(\pk, m,\SRO(m)=c$ for $m\coloneqq \Decrypt(\sk, c)$.
		(Equivalently, the latter means that $\oracleDecrypt(c) = m$.)
		\ifTightOnSpace\else Intuitively, this case again implies that \Ad{A} managed to compute a valid encryption without the respective oracle query on $m$.\fi
\end{itemize}
In the above, any statements about \SupOr calls that are not actually performed by the adversary or an oracle are assumed to be made right after the query $c$ and do not cause any measurement disturbance in that case. \cm{There were todos in the above, I took care of them by adding this sentence.}

We will again denote the last case by \EventGuessedCT. Whenever \EventDecapsSimDiffers occurs, \Ad{B} succeeds unless only case \EventGuessedCT occurs: If $\EventDecapsSimDiffers\wedge \neg\EventGuessedCT$ occurs, then a failing plaintext is extractable from the ciphertext that triggered $\EventDecapsSimDiffers\wedge \neg\EventGuessedCT$ (this time due to access to \SE),
and the plaintext is recognisable as failing by \Ad{B} due to its \FFPCCA oracle \oracleDecrypt. In formulae,
\begin{align*}
		\Pr\!\left[\EventDecapsSimDiffers\right]\!
	=\! \Pr\!\left[\EventDecapsSimDiffers \!\wedge\! \neg\EventGuessedCT\right]
	\!+\!\Pr\!\left[\EventDecapsSimDiffers\!\wedge \!\EventGuessedCT \right]
	\!\le\! \Adv^{\FFPCCA}_{\PKEDerand}\!\left[\Ad{B} \right] \!+\! \Pr\left[\EventGuessedCT\right]\!.
\end{align*}

In summary, we can bound the difference in advantages between \bfgame{2} and \bfgame{3} as
\begin{align*}
	\left|\Adv^\bfgame{2}-\Adv^\bfgame{3}\right|\le&\Adv^{\FFPCCA}_{\PKE^{\RO G}}\left(\Ad{B} \right)+\Pr\left[\EventGuessedCT\right].
\end{align*}

The following two steps are in a certain sense symmetric to the steps for \textbf{Games 0-2}:
\Ad{A} playing \bfgame{3} can almost be simulated without using the \oracleDecaps oracle, except that \oracleDecaps is still invoked before each call of \oracleDecapsSim, without the result ever being used.
This is an artifact from \bfgame{2}.
Omitting the \oracleDecaps invocations might introduce changes in \Ad{A}'s view, as these invocations might influence the behavior of \SE in subsequent queries.
We therefore define \bfgame{4} like \bfgame{3}, %
	except that the \oracleDecaps invocations are postponed until after \Ad{A} finishes.
By a similar argument as for the transition from \bfgame{1} to \bfgame{2}, we obtain
\begin{equation*}
	\left|\Adv^\bfgame{3}-\Adv^\bfgame{4}\right|\le 8\sqrt 2 \numberDecapsQueries^22^{-\gamma/2} \enspace .
\end{equation*}
Finally, \bfgame{5} is like \bfgame{4}, except that the computations of $\oracleDecaps(c_i)$ are omitted entirely.
In game 4, all invocations of \oracleDecaps already happened after the execution of \Ad{A}, hence this omission does not influence \Ad{A}'s success probability\ifTightOnSpace\else \ and
	\begin{equation*}
		\Adv^\bfgame{4} = \Adv^\bfgame{5} \enspace .
	\end{equation*}	
\fi

Let $\tilde{\Ad{A}}$ be an \INDCPAKEM adversary against \KemExplicitMess in the $\augQROM{\Encrypt}$, simulating \bfgame{5} to \Ad{A}:
$\tilde{\Ad{A}}$ has access to a single extractable oracle whose oracle interface \SRO simulates the combination of \RO{G} and \RO{H}, i.e., \SRO simulates $\RO G\times \RO H$.
(We decided to combine \RO{G} and \RO{H} into one oracle to simplify the subsequent analysis of the \INDCPA advantage against  \KemExplicitMess that will be carried out in \cref{sec:QROM:CPA-to-passive}.)
$\tilde{\Ad{A}}$ runs $b' \leftarrow \Ad{A}^{\RO G', \RO H, \oracleDecapsSim}$ and returns $b'$.
The simulation of \Ad{A}'s oracles using \SRO is straightforward (preparing the redundant register in uniform superposition, querying the combined oracle, and uncomputing the redundant register), but for completeness,
\ifTightOnSpace
\else we now explain the technique in more detail:  %
\ifTightOnSpace

	\section{Simulating $\RO{G} \times \RO{H}$ as an The Fujisaki-Okamoto transformation with explicit rejection} \label{sec:SimulatingGtimesH}
	
	This section discusses why combining oracles \RO{G} and \RO{H} into a single oracle \SupOr is merely of a conceptual nature,
	as it does not introduce any differences in the probabilities:
\fi
For any algorithm \Ad{A} expecting an \augQRO{\Encrypt}-modelled oracle \RO{G} and a QRO \RO{H},
one can define an algorithm $\tilde{A}$ with access to a single oracle \SupOr whose oracle interface \SRO represents $\RO{G} \times \RO{H}$ and whose extraction interface is only relative to \RO{G}, that perfectly simulates \Ad{A}'s view.
Whenever \Ad{A} issues a query to \RO{G}, $\tilde{A}$ prepares an additional output register $H_{\textnormal{out}}$ for \RO{H} in a uniform superposition, queries \SRO, uncomputes $H_{\textnormal{out}}$ by applying the Hadamard transform to $H_{\textnormal{out}}$,
and forwards the input-output registers belonging to \RO{G} to \Ad{A}. The same idea with reversed oracle roles can be used to answer queries to \RO{H}.
The extraction oracle \SE represents an extraction interface for $\Encrypt(\pk, \cdot;\cdot)$ with respect to \RO{G}:
This is possible as the oracle database for $\RO{G} \times \RO{H}: \MSpace \rightarrow \RSpace \times \KeySpace$ consists of registers $D_m$,
of which each register $D_m$ now consist of
one register $R_m$ to accommodate a superposition of elements in \RSpace (or $\bot$) and
one register $K_m$ to accommodate a superposition of elements in \KeySpace (or $\bot$).
The projectors of the measurements performed by extraction interface \SE can hence be defined in a way such that when \SE is queried on some ciphertext $c$, 
they select the message $m$ where $D_m$ is the first (in lexicographical order) register whose register $R_m$ contains an $r$ such that $\Encrypt(\pk, m; r) = c$. \fi

We now have
\ifTightOnSpace
\else 
	\begin{equation}\label{explanation:GtimesH}
		\Adv^\bfgame{5}=\Adv^{\INDCPAKEM}_{\KemExplicitMess}(\tilde{\Ad{A}}).
	\end{equation}
\fi

Collecting the terms from the hybrid transitions, using \cref{lem:pr-guess-augQROM} below, and bounding $\numberDecapsQueries 2^{-\gamma}\le \numberDecapsQueries^2 2^{-\gamma/2}$ yields the desired bound. The statements about query numbers, width and depth, as well as the runtime, are straightforward.
\qed\end{proof}

Like in \cref{sec:ROM}, we continue by bounding the probability of event \EventGuessedCT, and \cref{lem:pr-guess-augQROM} below is the \augQROM{\Encrypt} analogue of \cref{lem:pr-guess}.
Again, we will soon revisit \FFPCCA attacker \Ad{B} against \PKEDerand, and we will simulate \Ad{B}'s oracle \oracleDecrypt via an oracle \oracleDecryptSim (see \cref{fig:SimDecrypt:QROM}) that differs from \oracleDecrypt if an event equivalent to \EventGuessedCT occurs.
Therefore, we again generalize the definition of event \EventGuessedCT accordingly.

\begin{figure}[tb]\begin{center}
		
		\nicoresetlinenr
		
		\fbox{\small
			
			\begin{minipage}[t]{3cm}	
				
				\underline{$\oracleDecryptSim(c)$}
				\begin{nicodemus}
					\item $m \leftarrow \SE(c)$
					\item \pcreturn $m$
				\end{nicodemus}
			\end{minipage}
			
			\quad
			
			\begin{minipage}[t]{3.7cm}	
				
				\underline{\RO{G'}, input registers $X,Y$}
				\begin{nicodemus}
					\item Apply $\SRO_{XYD}$
					\item \pcreturn registers $XY$
				\end{nicodemus}
				
			\end{minipage}
			
		}		
		\caption{Simulation \oracleDecryptSim of oracle \oracleDecrypt for \PKEDerand. For the reader's convenience, we repeat the definition of \RO{G'}.}
		\label{fig:SimDecrypt:QROM}
	\end{center}
\end{figure}

\begin{lemma}\label{lem:pr-guess-augQROM}
	\ifTightOnSpace
\else
		Let \PKE be $\gamma$-spread,
		and let \Ad{A} be an \augQROM{\Encrypt} adversary that expects random oracles \RO{G}, \RO{H} as well as either a decapsulation oracle \oracleDecaps for \ifTightOnSpace \KemExplicitMess \else $\KemExplicitMess \coloneqq \FOExplicitMess[\PKE, \RO{G},\RO H]$ \fi
		or a decryption oracle \oracleDecrypt for \PKEDerand,
		issuing at most \numberDecOrDecapsQueries queries to the latter.
	\fi
	Let \Ad{A} be run with $\RO G'$ and \oracleDecaps or \oracleDecapsSim (\oracleDecrypt or \oracleDecryptSim), but for each query $c_i$, both $\hat m_i= \oracleDecryptSim(c_i)$ and $m_i=\oracleDecrypt(c_i)$  are computed in that order, regardless of which of the two oracles \oracleDecaps and \oracleDecapsSim (\oracleDecrypt and \oracleDecryptSim) \Ad{A} has access to.
	Then \EventGuessedCT, the event that $\hat m_i = \bot$ while $m_i \neq \bot$, is very unlikely. 
	Concretely,
	\begin{equation}
		\Pr\left[\EventGuessedCT\right]\le 2\numberDecOrDecapsQueries\cdot  2^{-\gamma}.
	\end{equation}
\end{lemma}
\begin{proof} 
	We begin by bounding the probability that for some fixed $i\in\{1,...,\numberDecOrDecapsQueries\}$ we have $\hat m_i=\bot$ but $m_i\neq \bot$. From the definitions of \oracleDecaps and \oracleDecapsSim, as well as the definitions of the interfaces \SRO and \SE, we obtain 
	\ifTightOnSpace
	\else
	the expression\fi
	\begin{align}
		\sqrt{\Pr[\hat m_i=\bot\wedge m_i\neq \bot]}=&\sqrt{\Pr[\hat m_i=\bot\wedge \Encrypt(m_i, \SRO(m_i))=c_i]}\nonumber\\
		=&\left\|\Pi_Y^{c,x}O_{XYF}\Sigma^{c,\emptyset}_F\ket{ m_i}_X\ket 0_Y\ket{\psi_i}_{FE}\right\| \label{eq:prob2op}
	\end{align}
	Here, $\ket{\psi_i}$ is the adversary-oracle state before \Ad{A} submits the query $c_i$ and the projectors $\Pi_Y^{c,x}$ and $\Sigma^{c,\emptyset}$ are with respect to $f=\Encrypt$ (see \cref{eq:extraction-measurement}). We begin by simplifying the expression on the right hand side. We have $O_{XYF}\ket{m_i}_X=F_{F_{m_i}}\CNOT^{\otimes n}_{F_{m_i}:Y}F_{F_{m_i}}\otimes \ket{m_i}_X$ and $\Pi_Y\CNOT^{\otimes n}_{F_{m_i}:Y}\ket 0_Y=\CNOT^{\otimes n}_{F_{m_i}:Y}\Pi_{F_{m_i}}\ket 0_Y$ for any projector $\Pi$ that is diagonal in the computational basis. We can thus simplify
	\begin{align}
		&\left\|\Pi_Y^{c,x}O_{XYF}\Sigma^{c,\emptyset}_F\ket{ m_i}_X\ket 0_Y\ket{\psi_i}_{FE}\right\|=\left\|\Pi_{F_{m_i}}^{c,x}F_{F_{m_i}}\Sigma^{c,\emptyset}\ket{ m_i}_X\ket 0_Y\ket{\psi_i}_{FE}\right\|\nonumber\\
		\le &\left\|F_{F_{m_i}}\Pi_{F_{m_i}}^{c,x}\Sigma^{c,\emptyset}_F\ket{ m_i}_X\ket 0_Y\ket{\psi_i}_{FE}\right\|+\left\|\left[\Pi^{c,x}, F\right]\right\|\nonumber\\
		\le &\left\|F_{F_{m_i}}\Pi_{F_{m_i}}^{c,x}\Sigma^{c,\emptyset}_F\ket{ m_i}_X\ket 0_Y\ket{\psi_i}_{FE}\right\|+\sqrt 2\cdot 2^{-\gamma/2}\label{eq:commutator-bound}
	\end{align}
	where we have applied the two observations and omitted any final unitary operators in the first equality, and the last inequality is due to Lemma 3.3 in \cite{DFMS21}. But the remaining norm term vanishes as
	\begin{align}\label{eq:orth-projectors}
		\Pi_{F_{m_i}}^{c,x}\Sigma^{c,\emptyset}_F=(\Pi^{c,x}\bar \Pi^{c,x})_{F_{m_i}}\otimes (\bar \Pi^{c,x})^{\otimes |\mathcal M|-1}_{F_{\mathcal M\setminus \{m_i\}}}=0.
	\end{align}
	Combining \cref{eq:prob2op,eq:commutator-bound,eq:orth-projectors} and squaring the resulting inequality yields
	\begin{equation}
		\Pr[\hat m_i=\bot\wedge m_i\neq \bot]\le 2\cdot 2^{-\gamma}.
	\end{equation}
	 Collecting the terms and applying a union bound over the $\numberDecOrDecapsQueries$ decapsulation queries yields the desired bound.
\qed\end{proof}

So far, we have shown that whenever an \INDCCA adversary \Ad{A}'s behaviour is significantly changed by being run with simulation \oracleDecapsSim instead of the real oracle \oracleDecaps, we can use \Ad{A} to find a failing plaintext,
assuming access to the decryption oracle \oracleDecrypt provided in the \FFPCCA game. 
We continue by proving an $\augQROM{\Encrypt}$-analogue of \cref{thm:FFP2}, i.e.,
we show that \oracleDecrypt can be simulated via oracle \oracleDecryptSim (see \cref{fig:SimDecrypt:QROM}) without the secret key,
thereby being able to construct an \FFPCPA adversary from any \FFPCCA adversary (both in the \augQROM{\Encrypt}).

\begin{theorem}[\PKEDerand \FFPCPA \RightarrowaugQROM{\Encrypt} \PKEDerand \FFPCCA]\label{thm:QFFP2}\label{thm:SimDec:QROM}
	\ifTightOnSpace
\else
		Let \PKE be $\gamma$-spread, and let \Ad{B}  be an \FFPCCA adversary in the \augQROM{\Encrypt} against \PKEDerand that makes at most \ifTightOnSpace \numberDecQueries, $q_{\SRO}$ and $q_{\SE}$ to its oracles $\oracleDecaps$, $\SRO$ and $\SE$, \else \numberDecQueries many decryption queries, and at most $q_{\SRO}$ and $q_{\SE}$ to the two interfaces of the \augQROM{\Encrypt}, \fi respectively.
	\fi
	Then there exist an \FFPCPA adversary $\tilde{\Ad{B}}$ in the \augQROM{\Encrypt} such that 
	\begin{equation}\label{eq:guessing-qB-quantum}
		\Adv^{\FFPCCA}_{\PKEDerand}(\Ad{B})
		\le (\numberDecQueries+1) \Adv^{\FFPCPA}_{\PKEDerand}(\tilde{\Ad{B}})
		+12%
		\numberDecQueries(q_{\RO G}+4\numberDecQueries)2^{-\gamma/2}%
	\end{equation}
	The adversary $\tilde{\Ad{B}}$ makes $q_{\SRO}$ queries to \SRO and $q_{\SE}+\numberDecQueries$ queries to \SE, and its runtime satisfies $\Time (\tilde{\Ad B})=\Time (\Ad B)+O(\numberDecapsQueries)$.
\end{theorem}
\begin{proof}
	On a high level, the proof works as follows. Analogous to \cref{thm:QFFP1}, we simulate \oracleDecrypt by \oracleDecryptSim. As we wish to remove the usage of \oracleDecrypt entirely, however, we cannot use it to determine at which \oracleDecryptSim query a failure occurs. We thus resort to guessing that information.
\ifTightOnSpace
\else

\fi
	On a technical level this proof follows the proof of \cref{thm:QFFP1} with deviations similar as in the proof of \cref{thm:FFP2}.
	Let \oracleDecryptSim be the simulation defined in \cref{fig:SimDecrypt:QROM}. Let \bfgame{0} be the \FFPCCA-game, and let \bfgameseq{1}{5} be defined based on \bfgame{0}  like in the proof of \cref{thm:QFFP1}. Like in the proof of \cref{thm:QFFP1}, we have 
	\ifTightOnSpace
\else
	\begin{align}
		&\Adv^\bfgame{0}\le \Adv^\bfgame{5}+ 12%
		\numberDecQueries(q_{\RO G}+2\numberDecQueries)2^{-\gamma/2}+\Pr[\EventDecapsSimDiffers]\nonumber\\
		\le &\Adv^\bfgame{5}+ 12%
		\numberDecQueries(q_{\RO G}+2\numberDecQueries)2^{-\gamma/2}+\Pr[\EventDecapsSimDiffers\wedge \neg \EventGuessedCT]+\Pr[\EventGuessedCT].\label{eq:QFFP2-firstineq}
	\end{align}
\fi
	Assume without loss of generality that $\Ad{B}$ makes exactly \numberDecQueries many queries to the oracle for $\DecryptDerand$ (if it does not, we modify $\Ad{B}$ by adding a number of useless decryption queries in the end). We define an \FFPCPA adversary $\tilde{\Ad{B}}^{\SupOr}$ defined exactly like the classical one in \cref{fig:SimDecrypt} (except that it has  quantum access to its oracles), i.e., $\tilde{\Ad{B}}$  samples $i\leftarrow\{1,...,\numberDecQueries+1\}$ and runs $\Ad{B}^{\RO G', \oracleDecryptSim}$ until the $i$-th query, or until the end if $i=\numberDecQueries+1$. Finally, $\tilde{\Ad{B}}$ outputs $m_i$, the output of $\Ad{B}^{\RO G', \oracleDecryptSim}$'s $i$-th decryption query, unless $i=\numberDecQueries+1$, in which case $\tilde{\Ad{B}}$ outputs the output of  $\Ad{B}^{\RO G', \oracleDecryptSim}$. %
	 By construction, 
\ifTightOnSpace
\else
	\begin{align}
	\Adv^{\FFPCPA}_{\PKEDerand}(\tilde{\Ad{B}})\ge &\frac{1}{\numberDecQueries+1}\left(\Adv^{\bfgame{5}}+\Pr[\EventDecapsSimDiffers\wedge \neg \EventGuessedCT]\right)\label{eq:QFFP2-secondineq}
	\end{align}
\fi
(note that all instances of $\Adv^{\mathbf{Game\ i}}$ are for $\Ad{B}$ playing \textbf{Game i}.)  Combining \cref{eq:QFFP2-firstineq,eq:QFFP2-secondineq,lem:pr-guess-augQROM} yields the desired bound. The statement about $\tilde{\Ad B}$'s running time and number of queries is straightforward.
\qed\end{proof}

Combining Theorems \ref{thm:QFFP1} and \ref{thm:QFFP2}, we obtain the \ifTightOnSpace\else $\augQROM{}$-analogue of \cref{cor:FFP3}.\fi
\begin{corollary}[$\FOExplicitMess\lbrack\PKE\rbrack$ \INDCPA and \PKEDerand \FFPCPA \RightarrowaugQROM{\Encrypt} $\FOExplicitMess\lbrack\PKE\rbrack$ $\INDCCA$] \label{cor:QFFP3}
		Let \PKE and \Ad A be as in \cref{thm:QFFP1}. %
	Then there exist an \INDCPAKEM adversary $\tilde{\Ad{A}}$ and an \FFPCPA adversary \Ad{B}, both in the \augQROM{\Encrypt},  such that
	\begin{align}
		\Adv^{\INDCCAKEM}_{\KemExplicitMess}(\Ad{A})
			\le & \Adv^{\INDCPAKEM}_{\KemExplicitMess}\left(\tilde{\Ad{A}}\right)
			+ (\numberDecQueries+1) \Adv^{\FFPCPA}_{\PKEDerand}\left({\Ad{B}}\right)
		\nonumber\\&+24%
		\numberDecapsQueries(q_{\RO G}+4\numberDecapsQueries)2^{-\gamma/2}%
		\label{eq:QFFP3}
	\end{align}
	Both \ifTightOnSpace \else adversaries \fi $\tilde{\Ad A}$ and $\Ad B$ make $q_{\RO G}+q_{\RO H}+\numberDecapsQueries$ queries to $\SRO$, with a \ifTightOnSpace\else combined depth of $d+\numberDecQueries$ and a combined width of $w$\fi,  and \numberDecQueries queries to \SE%
	. %
	The running times of  $\tilde{\Ad{A}}$ and $\Ad B$ satisfy
	$\Time (\tilde{\Ad A})=\Time (\Ad A)+O(\numberDecapsQueries)$  and $\Time (\Ad B)=\Time (\Ad A)+O(\numberDecapsQueries)$.
\end{corollary}
Again, we remark that the additive error terms are a factor of $2$ larger due to our modular proof (in terms of Theorems \ref{thm:QFFP1} and \ref{thm:QFFP2}). It is straightforward to show that the bound of \cref{cor:QFFP3} can be proven without the factor of $2$, when directly analyzing the composition of the reductions from Theorems \ref{thm:QFFP1} and \ref{thm:QFFP2}.

While the additive error term \ifTightOnSpace\else that depends on the spreadness parameter $\gamma$ \fi 
improves by roughly a power 2 over the corresponding term in the security bound of \cite{DFMS21}, the only known concrete bound for \FOExplicitMess, we remark that we do not expect it to be tight. It turns out, however, that many relevant schemes have abundantly randomized ciphertexts. In \cref{sec:spreadness}, we bound the spreadness parameter for some schemes where this was relatively easy to do: the alternate candidates in the NIST post-quantum cryptography competition Frodo and HQC.

\subsection{Semi-classical OWTH in the \augQROM{f}} \label{sec:QROM:OWTH}

To further analzye \INDCPAKEM security of $\FOExplicitMess[\PKE, \RO{G},\RO{H}]$, in the \augQROM{\Encrypt},
we want to apply an \augQROM{\Encrypt} argument to show that keys encapsulated by $\FOExplicitMess[\PKE, \RO{G},\RO{H}]$ are random-looking unless the adversary can be used to attack the underlying scheme \PKE.
\ifTightOnSpace\else In slightly more detail, w\fi e will need to argue that the challenge key $K^*:=\RO{H}(m^*)$ and the encryption randomness $\RO{G}(m^*)$ used for challenge ciphertext $c^*$ can be replaced with fresh random values, in the \augQROM{\Encrypt}.
To \ifTightOnSpace\else theoretically justify this argument, this section develops \fi
\augQROM{f} generalizations of the semi-classical \OWTH theorems from \cite{C:AmbHamUnr19}.

We will first describe how we model this 'replacing with fresh randomness' on a subset $\reproSet \subset \mathcal X$ for superposition oracle, and how our approach generalizes previous approaches.
Previous work (like \cite{C:AmbHamUnr19}) used two oracles \RO{O_0} and \RO{O_1} that only differ on some set \reproSet,
while algorithm \Ad{A}'s input is always defined relative to oracle \RO{O_0}.
In the case where \Ad{A}'s oracle is \RO{O_1}, the input uses fresh randomness from the adversary's point of view.
Here we meet the first \augQROM{\Encrypt}-related roadblock:
Superposition oracles have the property that initially, each value $\SRO(x)$ is in \emph{quantum superposition},
which complicates equating two oracles everywhere but on \reproSet. 
As it suffices for our purpose, we define the 'resampling' set \reproSet as follows:
We assume \Ad{A}'s input \inputVar to be classical,  generated by an algorithm \GenInput with classical access to \SupOrNotProg.
We can then define \reproSet as the set of all inputs $x$ queried by \GenInput,
e.g., for input $(c^*, K^*) \coloneqq (\Encrypt(\pk, m^*; \RO{G}(m^*))), \RO{H}(m^*))$, \reproSet is $\lbrace m^* \rbrace$.) \cm{If \GenInput can be an adaptive algorithm, do we need to include all inputs in \reproSet that are queried by \GenInput for \emph{some} RO?}
Apart from how we model \reproSet, we proceed as in \cite{C:AmbHamUnr19}:
Use \SupOrNotProg to generate \Ad{A}'s input and replace \Ad{A}'s access to \SupOrNotProg with access to \SupOrProg, an independent extractable compressed oracle.

{%
}
{%
}

Clearly, if \GenInput does not query \SupOrNotProg, the two oracles \SupOrNotProg and \SupOrProg are perfectly indistinguishable to \Ad{A}.
But what if \Ad{A}'s input depends on \SupOrNotProg?
\cite{C:AmbHamUnr19} related \Ad A's distinguishing advantage to the probability of ``\FIND'',
the event that an element of \reproSet is detected in \Ad{A}'s queries to the QRO via a quantum measurement.
This result, however, is in the (plain) QROM, and \FIND is not be the only distinction opportunity in the \augQROM{f} as there are now two oracle interfaces, \SRO and \SE.
As an example, let \Ad{A} have input $(x, t \coloneqq f(x, \SRONotProg(x)))$ for some oracle input value $x$.
\emph{Without any \SRO query}, \Ad{A} can tell the two cases apart by querying \SE on $t$:
Querying \SENotProg on $t$ results in output $x$ with overwhelming probability, while querying \SEProg on $t$ yields output $\bot$.
Extraction queries hence have to be taken into account.

Before stating this section's main theorems, we will describe our approach more formally.
Borrowing the notation from \cite{C:AmbHamUnr19}, we define \emph{`punctured'  versions} \puncture{\SupOr}{\reproSet} of \ifTightOnSpace \else extractable superposition oracles \fi \SupOr:
When an \SRO query is performed, we first apply a \emph{'semi-classical'} oracle \orSemiClassical{S},
and then oracle unitary \OrUnit.
Intuitively, \orSemiClassical{S} marks if an element of \reproSet was found in one of the query registers.
\ifTightOnSpace\else (The plural is used since we consider parallel queries.) \fi
Formally, \orSemiClassical{S} acts on the query input registers $X_1, \cdots X_w$ and \ifTightOnSpace\else an additional \fi `flag' register $F$ that holds one qubit per oracle query,
by first mapping \ifTightOnSpace
\else 
\[ \ket{x_1, \cdots x_w}_{X_1 \cdots X_w} \otimes \ket{b}_F
	\mapsto \ket{x_1, \cdots x_w}_{X_1 \cdots X_w} \otimes \ket{b \oplus \bool{x_1 \in \reproSet \vee \cdots \vee x_w \in \reproSet}}_F \enspace ,\]
\fi
and then measuring register $F$ in the computational basis.

Like in \cite{C:AmbHamUnr19},  we denote the event that any measurement of $F$ returns 1 by \FIND.
In that case, the query has collapsed to a superposition of states where at least one input register only contains elements of \reproSet.
If \FIND does not occur, then all oracle queries collapsed to states not containing any elements of \reproSet,
and in consequence, set \reproSet defining \Ad{A}'s input is effectively removed from the query input domain.
In this case, the only way to distinguish between \SupOrNotProg and \SupOrProg is to perform an extraction query where \SENotProg might returns an element of \reproSet. We will call this event \EVENTEXT.
If neither \FIND nor \EVENTEXT occur, the two scenarios are indistinguishable to \Ad{A}.

The following helper lemma formalizes the above reasoning and extends it to some other probability distances: \cref{eq:OWTH:EqualEventNoFindNoExtract} formalizes that if \Ad{A} neither triggers \FIND
\ifTightOnSpace\else (and hence never sees a random oracle value on \reproSet) \fi
nor \EVENTEXT\ifTightOnSpace\else (meaning no extraction is performed an on a critical point)\fi,
 its behaviour in the two cases is the same: %
 arbitrary events will be equally likely in both cases.
\ifTightOnSpace\else\cref{eq:OWTH:DiffEventNoFind} states  that if \Ad{A} does not trigger \FIND, any event will only become more likely in the resampled scenario than in the honest scenario if \EVENTEXT happens.
During the proof of one of this section's main theorems, we need to also reason about the probability of \FIND in the two cases.
\cref{eq:OWTH:DiffFind} states that the likelihood of \FIND only differs in the two scenarios if \EVENTEXT happens.
(To make this statement more intuitive, consider an adversary with input $(x, t \coloneqq f(x, \SRONotProg(x)))$ that first performs an extraction query on $t$ and then queries the oracle on the result.)
\fi
The proof of \cref{lem:OWTH:DiffEvents} is mostly reworking the probabilities by reasoning about the cases and eliminating the case where neither \FIND nor \EVENTEXT occurs.
It is given in 
\ifCameraReady
the full version.
\else 
	\cref{sec:QROM:OWTH:DiffEvents} (page~\pageref{sec:QROM:OWTH:DiffEvents}).
\fi
\begin{restatable}{lemma}{DiffEvent}\label{lem:OWTH:DiffEvents}
	Let \SupOrNotProg and \SupOrProg be two extractable superposition oracles from $\mathcal X$ to $\mathcal Y$
	for some function $f:\mathcal X\times \mathcal Y\to\mathcal T$, and let \GenInput be an algorithm with classical output \inputVar, having access to \SupOrNotProg.
	Let \reproSet be the set of elements $x \in \mathcal X$ whose oracle values are needed to compute \inputVar,
	and let $\commitmentSet_{\reproSet} \coloneqq \lbrace t \mid \exists x \in \reproSet \textnormal{ s.th. } t = f(x, \SupOrNotProg(x)) \rbrace$.
	Let \FIND be the event that flag register $F$ is ever measured to be in state 1 during a call to \Ad{A}'s punctured oracle,
	and let \EVENTEXT be the event that \Ad{A} performs an extraction query on any $t \in \commitmentSet_{\reproSet}$.
	Let \EVENT be an arbitrary (classical) event.
	Then 
	\ifTightOnSpace
\else
	\begin{align} \label{eq:OWTH:EqualEventNoFindNoExtract}
		\Pr[\EVENT  & \wedge \neg \FIND \wedge \neg \EVENTEXT:  \Ad{A}^{\puncture{\SupOrNotProg}{\reproSet}}]
		\nonumber \\
		 & = \Pr[\EVENT \wedge \neg \FIND \wedge \neg \EVENTEXT: \Ad{A}^{\puncture{\SupOrProg}{\reproSet}}]
		\enspace ,
	\end{align}
	\begin{align} \label{eq:OWTH:DiffEventNoFind}
		|
			\Pr[\EVENT  \wedge \neg \FIND : \Ad{A}^{\puncture{\SupOrNotProg}{\reproSet}}]
			-  \Pr[\EVENT \wedge \neg \FIND : & \Ad{A}^{\puncture{\SupOrProg}{\reproSet}}]
		|
		\nonumber \\
		& \leq \Pr[\EVENTEXT: \Ad{A}^{\puncture{\SupOrNotProg}{\reproSet}}]
		\enspace ,
	\end{align}
	\begin{align}\label{eq:OWTH:DiffFind}
		|
			\Pr[\FIND : \Ad{A}^{\puncture{\SupOrNotProg}{\reproSet}}]
			-  \Pr[\FIND : \Ad{A}^{\puncture{\SupOrProg}{\reproSet}}]
		|
		\leq \Pr[\EVENTEXT: \Ad{A}^{\puncture{\SupOrNotProg}{\reproSet}}]
		\enspace ,
	\end{align}
\fi
	where all probabilities are taken over the coins of \GenInput and the internal randomness of \Ad{A}
	and we used $\Ad{A}^{\RO{O_0}}$ as a shorthand for $\Ad{A}^{\RO{O_0}}(\inputVar)$.
	
\end{restatable}

The following theorem relates the distinguishing advantage between \SupOrNotProg and \SupOrProg to the probability that \FIND or \EVENTEXT occur. Intuitively, the theorem states that no algorithm \Ad{A} will recognize the reprogramming unless \Ad{A} %
makes a random oracle or an extraction query related to its input.
\cref{thm:OWTH:Dist-to-Find} is the \augQROM{f} counterpart of \cite[Th. 1, 'Semi-classical O2H']{C:AmbHamUnr19}.
Its proof is given in 
\ifCameraReady
	the full version.
\else 
	\cref{sec:QROM:OWTH:DistToFind} (page~\pageref{sec:QROM:OWTH:DistToFind}).
\fi
In the special case where \EVENTEXT never happens, %
e.g., when extraction queries are triggered by an oracle simulation like \oracleDecapsSim that forbids critical inputs, 
we obtain the same bound as  \cite[Th. 1]{C:AmbHamUnr19},
but in the \augQROM{f}.

\begin{restatable}[Semi-classical \OWTH in the \augQROM{f}: Distinguishing to Finding]{theorem}{DistToFind} \label{thm:OWTH:Dist-to-Find}
	Let \SupOrNotProg, \SupOrProg, \GenInput, \reproSet, \FIND and \EVENTEXT be like in \cref{lem:OWTH:DiffEvents}.
	We define the \OWTH distinguishing advantage function of \Ad{A} as
	\[ \Adv^{\OWTH}_{\augQRO{f}}(\Ad{A})
		:=	|\Pr[1 \leftarrow \Ad{A}^{\SupOrNotProg}(\inputVar)]  - \Pr[1 \leftarrow \Ad{A}^{\SupOrProg}(\inputVar)] |
	\enspace, \]
	where the probabilities are \ifTightOnSpace \else taken \fi over the coins of \GenInput and the \ifTightOnSpace\else internal \fi randomness of \Ad{A}.
	For any algorithm \Ad{A} of query depth $d$ with respect to \SRO,
	we have that
	\begin{align} \label{eq:OWTH:Dist-to-Find:General}
		\Adv^{\OWTH}_{\augQRO{f}}(\Ad{A})
			\leq  & 4 \cdot \sqrt{d \cdot \Pr[\FIND : \Ad{A}^{\puncture{\SupOrProg}{\reproSet}}]}
		\nonumber \\
		& + 2 \cdot (\sqrt{d}+1)  \cdot \sqrt{\Pr[\EVENTEXT: \Ad{A}^{\SupOr^0}]} + \Pr[\EVENTEXT: \Ad{A}^{\SupOr^1}]
		\enspace .
	\end{align}
\ifTightOnSpace\else In the special case where \fi
	$\Pr[\EVENTEXT: \Ad{A}^{\puncture{\SupOrNotProg}{\reproSet}}] = \Pr[\EVENTEXT: \Ad{A}^{\puncture{\SupOrProg}{\reproSet}}] = 0$,  we obtain
	\begin{align} \label{eq:OWTH:Dist-to-Find:NoExtract}
		\Adv^{\OWTH}_{\augQRO{f}}(\Ad{A}) \leq 4 \cdot \sqrt{d \cdot \Pr[\FIND : \Ad{A}^{\puncture{\SupOrProg}{\reproSet}}]}
		\enspace .
	\end{align}

\end{restatable}

In many cases, a desired reduction will not know the 'resampled' set \reproSet. \ifTightOnSpace\else  We therefore proceed by giving \cref{thm:OWTH:Find-to-Extract} which \fi  relates the probability of \FIND to the advantage of a preimage extractor \ifTightOnSpace\else algorithm \fi \Extractor that extracts an element of \reproSet without knowing \reproSet: \Extractor will \ifTightOnSpace \else simply \fi run \Ad{A} with the unpunctured oracle \SupOr and measure one of its queries to generate its output.
In one of our proofs, we %
additionally need to puncture on a set different from \reproSet.
We therefore prove \cref{thm:OWTH:Find-to-Extract} for \emph{arbitrary} sets $\reproSet''$ .

\begin{restatable}[Semi-classical \OWTH in the \augQROM{f}: Finding to Extracting]{theorem}{FindToExtract} \label{thm:OWTH:Find-to-Extract}
	Let \Ad{A} be an algorithm with access to an extractable superposition oracle \SupOr from $\mathcal X$ to $\mathcal Y$
	for some function $f:\mathcal X\times \mathcal Y\to\mathcal T$,
	with query depth $d$ with respect to \SRO, and let \GenInput 
	be like in \cref{lem:OWTH:DiffEvents}.
	Let \FIND be the event that flag register $F$ is ever measured to be in state 1 during a call to \Ad{A}'s punctured oracle,
	where the puncturing happens on a set $\reproSet''$.	
	
	Let \Extractor be the algorithm that on input \inputVar chooses $i \uni \lbrace 1, \cdots d \rbrace$,
	runs $\Ad A^{\SupOr}(\inputVar)$ until \ifTightOnSpace\else (just before) \fi the $i$-th query to \SRO; then measures all query input registers in the computational basis and outputs the set \extractionSet of measurement outcomes.
	Then
	\begin{equation}\label{eq:OWTH:Find-to-Extract}
		\Pr[\FIND : \Ad{A}^{\puncture{\SupOr}{\reproSet''}}]
		\leq 4 d \cdot \Pr[\reproSet'' \cap \extractionSet \neq \emptyset: \extractionSet \leftarrow \Extractor]
		\enspace .
	\end{equation}
\end{restatable}

The proof directly follows from \cite[Th. 2, 'Search in semi-classical oracle']{C:AmbHamUnr19} since \cite[Th. 2]{C:AmbHamUnr19} gives
\ifTightOnSpace
\else a bound with the same ride-hand side as in \fi \cref{thm:OWTH:Find-to-Extract} for algorithms \Ad{B} accessing a semi-classical oracle \orSemiClassical{\reproSet''} itself (rather than some oracle punctured on $\reproSet''$). An algorithm $\Ad{B}^{\orSemiClassical{\reproSet''}}$ hence can perfectly simulate \puncture{\SupOr}{\reproSet''} to \Ad{A} by simulating \SupOr and having the puncturing done by its own oracle \orSemiClassical{\reproSet''}.
\ifTightOnSpace
\else
\ifTightOnSpace %
	\section{Proof of \cref{thm:OWTH:Find-to-Extract} (Finding to Extracting)} \label{sec:QROM:OWTH:FindToExtract}
	
	For easier reference, we repeat the statement of \cref{thm:OWTH:Find-to-Extract}.
	
	\FindToExtract*
	
\else

\fi

\begin{proof}
Given an algorithm $\Ad A^{\puncture{\SupOr}{\reproSet}}$, we define an algorithm $\Ad{B}^{\orSemiClassical{\reproSet''}}$ as follows:
$\Ad{B}^{\orSemiClassical{\reproSet''}}$ initializes a fresh extractable superposition oracle simulation \SupOr.
After generating \Ad{A}'s input \inputVar,
\Ad{B} runs $\Ad A^{\puncture{\SupOr}{\reproSet}}$ by simulating \puncture{\SupOr}{\reproSet} as follows:
Extraction queries are simply answered using \SE, and random oracle queries with query registers $XY$ are answered by first performing a query to its own oracle $\orSemiClassical{\reproSet''}$ with these registers and then applying $\SRO$.

Since \Ad{B} perfectly simulates \puncture{\SupOr}{\reproSet} to \Ad{A} and since \Ad{B}'s queries to \orSemiClassical{\reproSet''} are exactly \Ad{A}'s queries to \puncture{\SupOr}{\reproSet},
\begin{equation}
	\Pr[\FIND : \Ad{A}^{\puncture{\SupOr}{\reproSet''}}] = \Pr[\FIND : \Ad{B}^{\orSemiClassical{\reproSet''}}]
	\enspace .
\end{equation}

Applying \cite[Th. 2]{C:AmbHamUnr19} to \Ad{B} yields

\begin{equation}
	\Pr[\FIND : \Ad{B}^{\orSemiClassical{\reproSet''}}]
	\leq 4 d \cdot \Pr[\reproSet'' \cap \extractionSet \neq \emptyset: \extractionSet \leftarrow \Extractor'(\Ad{B})]
	\enspace ,
\end{equation}

where $\Extractor'$ randomly measures one of \Ad{B}'s queries to generate its output.
Unwrapping \Ad{B} into $\Extractor'$ defines the theorem's extractor \Extractor that randomly measures one of \Ad{A}'s queries to generate its output.

\begin{equation}
	\Pr[\reproSet'' \cap \extractionSet \neq \emptyset: \extractionSet \leftarrow \Extractor'(\Ad{B})]
		= \Pr[\reproSet'' \cap \extractionSet \neq \emptyset: \extractionSet \leftarrow \Extractor(\Ad{A})]
	\enspace .
\end{equation}

Collecting the probabilities yields the desired bound.
\qed
\end{proof} \fi

\ifTightOnSpace\else In the case that \fi the input \inputVar of \Ad{A} is independent of $\reproSet''$, we
\ifTightOnSpace
\else
	furthermore get the following extraction bound. \cref{thm:OWTH:Find-to-Extract-Independent} is the \augQROM{f} counterpart of \cite[Cor. 1]{C:AmbHamUnr19}.
\fi

\begin{corollary}[Semi-classical \OWTH in the \augQROM{f}: Extracting independent values]  \label{thm:OWTH:Find-to-Extract-Independent}
	If $\reproSet$ and $\inputVar$ are independent, then for any algorithm $\Ad A^{\SupOr}$ issuing $q$ many queries to \SRO in total,
	\begin{equation*}
		\Pr[\FIND : \Ad{A}^{\puncture{\SupOr}{\reproSet''}}] \leq 4q \cdot p_{\textnormal{max}}
		\enspace ,
	\end{equation*}
	where $p_{\textnormal{max}} := \max_{x in X} \Pr_{\reproSet''} [x \in S]$.
	As a special case, we obtain that
	\ifTightOnSpace
\else
	\begin{equation} \label{eq:OWTH:Find-to-Extract-Independent}
		\Pr[\FIND : \Ad{A}^{\puncture{\SupOr}{\lbrace x \rbrace}}] \leq \frac{4q}{|X|}
		\enspace ,
	\end{equation}
\fi
	 for $\reproSet'' = \lbrace x \rbrace$ with uniformly chosen $x \in X$, assuming that $x$ was not needed to generate the input to \Ad{A}.	
\end{corollary}
\ifTightOnSpace\else The proof is the same as in \cite{C:AmbHamUnr19}: W.l.o.g., we can assume that \Ad{A} does not perform parallel queries,
meaning that $q=d$ and the probability of \Extractor succeeding is the probability that \Extractor outputs an element $x \in \reproSet''$ that is independent of its input.
Hence $\Pr[\reproSet'' \cap \extractionSet \neq \emptyset: \extractionSet \leftarrow \Extractor(\inputVar)] \leq p_{\textnormal{max}}$,
and the corollary follows from \cref{thm:OWTH:Find-to-Extract}.\fi

\subsection{From $\INDCPA_{\PKE}$ or $\OWCPA_{\PKE}$ to $\INDCPA_{\FO[\PKE]}$}\label{sec:QROM:CPA-to-passive}

We will now use the \OWTH results from \cref{sec:QROM:OWTH} to show that the \INDCPA security of $\FOExplicitMess[\PKE, \RO{G},\RO{H}]$ can be based on the passive security of \PKE.
In \cref{thm:INDPKEToINDCPAKEM}, we base \INDCPA security of $\FOExplicitMess[\PKE, \RO{G},\RO{H}]$ on the \INDCPA security of \PKE, and \ifTightOnSpace \else for the sake of completeness, \fi we base it on \OWCPA security of \PKE in \cref{thm:OWPKEToINDCPAKEM}.
The obtained bounds are the same as their known plain QROM counterparts.

\begin{restatable}[\PKE \INDCPA \RightarrowaugQROM{\Encrypt} $\FO\lbrack\PKE\rbrack$ $\INDCPAKEM$]{theorem}{INDPKEToINDCPAKEM} \label{thm:INDPKEToINDCPAKEM}
	Let \Ad{A} be an \INDCPA adversary against  \ifTightOnSpace\else $\KemExplicitMess \coloneqq \FOExplicitMess[\PKE, \RO{G},\RO{H}]$ \fi in the \augQROM{\Encrypt},
	issuing $q$ many queries to \SRO in total, with a query depth of $d$,
	and $q_E$ many queries to \SE, where none of them is with its challenge ciphertext.
	Then there exists an \INDCPA adversary $\Ad{B_{\INDCPA}}$ against \PKE such that
	\ifTightOnSpace
\else 
	\[\Adv^{\INDCPAKEM}_{\KemExplicitMess}(\Ad{A}) \leq 4 \cdot \sqrt{ d  \cdot \Adv^{\INDCPA}_{\PKE}(\Ad{B_{\INDCPA}})}
		+ \frac{8q}{\sqrt{\left|\MSpace\right|}}\enspace \ifTightOnSpace , \else . \fi 
	\]
	\fi
	\ifTightOnSpace\else The running time and quantum memory footprint of $\Ad{B_{\INDCPA}}$ satisfy \fi 
	 $\Time(\Ad{B_{\INDCPA}})=\Time(\Ad{A})+\Time(\SupOr,q, q_E)$ and $\QMem(\Ad{B_{\INDCPA}})=\QMem(\Ad{A})+\QMem(\SupOr,q, q_E)$.
\end{restatable}

Note that forbidding extraction queries to \SE on $c^*$ is no limitation in
\ifTightOnSpace
\else 
in the context of the overall result: In the bigger picture,
\fi
\SE queries are only triggered by an \INDCCA adversary querying its simulated oracle \oracleDecapsSim, and \oracleDecapsSim rejects queries on $c^*$ \ifTightOnSpace\else right away. \fi

To summarise the proof, we first define a \game{1} like the \INDCPAKEM game \ifTightOnSpace \else for \KemExplicitMess\fi, except that encryption randomness $r^* := \RO{G}( m^*)$ and honest KEM key $K_0 := \RO{H}(m^*)$ are replaced with fresh uniform randomness.
In \game{1}, the forwarded KEM key is a uniformly random key either way, the advantage of \Ad{A} in \game{1} hence is 0.
It remains to bound the distinguishing advantage between the \INDCPAKEM game and \game{1}.
We apply \ifTightOnSpace\else the 'Distinguishing to Finding' \fi \cref{thm:OWTH:Dist-to-Find}
which bounds this distinguishing advantage in terms of the probability of event $\FIND_{m^*}$, the event that $m^*$ is detected in the adversary's random oracle queries.
To further bound $\Pr[\FIND_{m^*}]$, we use \INDCPA security of \PKE to replace \Ad{A}'s ciphertext input $c^*$ with an encryption of an independent message.
As $m^*$ now is independent of \Ad{A}'s input,
$\FIND_{m^*}$ is highly unlikely for large enough message spaces.
(This uses \ifTightOnSpace\else the 'independent values' \fi \cref{thm:OWTH:Find-to-Extract-Independent} %
.)

\ifTightOnSpace \else %
\ifTightOnSpace %
	\section{Proof of Theorems \ref{thm:INDPKEToINDCPAKEM} and \ref{thm:OWPKEToINDCPAKEM} (From $\INDCPA_{\PKE}$ or $\OWCPA_{\PKE}$ to $\INDCPA_{\FO[\PKE]}$ in the \augQROM{\Encrypt})}\label{sec:QROM:CPA-to-passive:proof}

	We will now use the \OWTH results from \cref{sec:QROM:OWTH} to prove \cref{thm:INDPKEToINDCPAKEM} and \cref{thm:OWPKEToINDCPAKEM}. We will first prove \cref{thm:INDPKEToINDCPAKEM}, which we repeat for easier reference:
	
	\INDPKEToINDCPAKEM*
	
	As sketched in the main body, the proof consists of the following steps:
	First, replace the encryption randomness $r^* := \RO{G}( m^*)$ and the honest KEM key $K_0 := \RO{H}(m^*)$ in \game{1} with fresh uniform randomness.
	Since in \game{1}, the forwarded key is uniformly random either way, an adversary has no distinguishing advantage at all and it remains to upper bound the distinguishing advantage between the \INDCPA game and \game{1}.
	To this end, we use the reprogramming theorem (\cref{thm:OWTH:Dist-to-Find}, see page~\pageref{thm:OWTH:Dist-to-Find}).
	\cref{thm:OWTH:Dist-to-Find} bounds the distinguishing advantage in terms of the probability of an event $\FIND_{m^*}$, the event that $m^*$ is detected in the adversary's random oracle queries.
	To upper bound $\Pr[\FIND_{m^*}]$, we assume \INDCPA security of \PKE to argue that the challenge ciphertext $c^*$ can be replaced with an encryption of another random message.
	After this change, $m^*$ is independent of \Ad{A}'s input, therefore $\FIND_{m^*}$ becomes highly unlikely for large enough message spaces. The last argument is made more formal using theorem \cref{thm:OWTH:Find-to-Extract-Independent} (given on page~\pageref{thm:OWTH:Find-to-Extract-Independent}).
	
\fi

\begin{proof}

	Let \Ad{A} be an adversary against the \INDCPA security of $\KemExplicit = \FOExplicitMess[\PKE, \RO{G},\RO{H}]$,
	issuing random oracle queries to both its oracles of query depth $d$, and $q$ many in total. 
	Consider the two games given in \cref{fig:games:INDPKEToINDCPAKEM}.
	
	\begin{figure}[htb] \begin{center} \fbox{\small
				
				\nicoresetlinenr
				
				\begin{minipage}[t]{5cm}	
					
					\underline{\bfgameseq{0}{1}} 
					\begin{nicodemus}
						
						\item $(\pk,\sk) \leftarrow \KG$
						\item $b \uni \bits$
						\item $m^* \uni \MSpace$
						\item $(r^*, K_0^*) := \SRO ( m^*)$ \gcom{$G_0$}
						\item $(r^*, K_0^*) \uni \RSpace \times \KeySpace$ \label{line:reprogramOracleValues} \gcom{$G_1$} 
						\item $c^* := \Encrypt(\pk, m^*; r^*)$ \label{line:Encrypt}
						\item $K_1^* \uni \KeySpace$
						\item $b'\leftarrow \Ad{A}^{\SupOr}(\pk, c^*,K_b^*)$
						\item \pcreturn $\bool{b' = b}$
						
					\end{nicodemus}
					
				\end{minipage}
	}	
	\end{center}
		\caption{\gameseq{0}{1} for the proof of \cref{thm:INDPKEToINDCPAKEM}.}			
		\label{fig:games:INDPKEToINDCPAKEM}
	\end{figure}
	
	\bfgame{0} essentially is game $\INDCPAKEM_{\KemExplicit}$, the only difference is that we combined oracles \RO{G} and \RO{H} into a single oracle \SupOr.
	\ifTightOnSpace
		Again, we point to \cref{sec:SimulatingGtimesH} that explains why 
	\else
		As discussed in the proof of \cref{thm:INDandFFPtoCCA:QROM}
		(before \cref{explanation:GtimesH}, see page~\pageref{explanation:GtimesH}),
	\fi
	this change is merely of a conceptual nature, simplifying our later reasoning about the synchronous reprogramming of $\RO{G}$ and $\RO{H}$ on $m^*$.
	\[
		\Adv^{\INDCPAKEM}_{\KemExplicit}(\Ad{A}) = |\Adv^\bfgame{0} - \frac{1}{2} |\enspace .
	\]

	\newcommand{\PrFindInGameOne}{\ensuremath{
			\Pr[\FIND_{m^*} \textnormal{ in } \bfgame{1}^{\puncture{\SupOr}{ \lbrace m^* \rbrace}} ]
	}\xspace}

	{In \bfgame{1}, we replace oracle values $r^* := \RO{G}( m^*)$ and $K_0 := \RO{H}(m^*)$
		with fresh random values (see line \ref{line:reprogramOracleValues}).
		Since $K_b^*$ is now an independent random value regardless of the challenge bit,
		\[
			\Adv^\bfgame{1} = \frac{1}{2} \enspace .
		\]

		We will now apply \cref{thm:OWTH:Dist-to-Find} to relate \Ad{A} being able to distinguish between \game{0} and \game{1} to the probability that \Ad{A}'s queries contain $m^*$,
		or more precisely, the probability that \Ad{A} would trigger event $\FIND_{m^*}$ in \game{1},
		would it be run with the punctured oracle $\puncture{\SupOr}{ \lbrace m^* \rbrace}$ that additionally measures whether any of \Ad{A}'s random oracle queries contained $m^*$ and in that case sets flag $\FIND_{m^*}$ to 1.
		We claim that
		\begin{equation}\label{eq:FO:DistToFind}
			|\Adv^\bfgame{0} - \Adv^\bfgame{1} |
				\leq  4 \cdot \sqrt{d \cdot \PrFindInGameOne } \enspace .
		\end{equation}
		
		To verify this claim, we identify each \game{b} (where $b \in \lbrace 0, 1 \rbrace$) with one of the \OWTH games defined in \cref{thm:OWTH:Dist-to-Find} as follows:
		As \GenInput, we define the algorithm that samples a key pair and a random message $m^*$,
		queries \SupOrNotProg on $m^*$ to obtain $r^*$ and $K_0^*$, and outputs as \inputVar the public key as well as $c^* \coloneqq \Encrypt(\pk, m^*; r^*)$ and $K_0^*$.
		With this identification, set \reproSet from \cref{thm:OWTH:Dist-to-Find} is $\lbrace m^*\rbrace$.
		
		As the \OWTH distinguisher, we define algorithm \Ad{D} that gets \inputVar, picks a random bit $b$ and a random key $K_1^*$ and then forwards $\pk$, $c^*$ and $K_b^*$ to \Ad{A}.
		It forwards all of \Ad{A}'s random oracle and extraction queries to its own respective oracle, and at the end, it returns 1 iff \Ad{A}'s output bit is equal to $b$.
		When \Ad{D} is run with access to \SupOrNotProg, it perfectly simulates \game{0},
		and when \Ad{D} is run with access to \SupOrProg, the input is defined relative to oracle \SupOrNotProg, while the oracle to which \Ad{A}'s queries are forwarded by \Ad{D} is \SupOrProg. Since everything except for the values $r^*$ and $K_0^*$ computed by \GenInput is now independent of the oracle \SupOrNotProg which is furthermore inaccessible to \Ad{D} and \Ad{A}, this is equivalent to simply sampling random values $r^*$ and $K_0^*$ instead, therefore		
		\[
			|\Adv^\bfgame{0} - \Adv^\bfgame{1} | = \Adv^{\OWTH}_{\augQRO{f}}(\Ad{D}) \enspace.
		\]

		Note that \EVENTEXT from \cref{thm:OWTH:Dist-to-Find} corresponds to the event that \Ad{A} queries its extraction oracle \SE on $c^*$, which we ruled out in the theorem statement as a prerequisite.
		Therefore, we can apply the special case bound \cref{eq:OWTH:Dist-to-Find:NoExtract} of \cref{thm:OWTH:Dist-to-Find},
		and since \Ad{D} has exactly the query behaviour of \Ad{A} and triggers \FIND exactly if \Ad{A} triggers \FIND,
		\[
			\Adv^{\OWTH}_{\augQRO{f}}(\Ad{D}) \leq 4 \cdot \sqrt{d \cdot \PrFindInGameOne}
			\enspace .
		\]
		
	}
	
	What we have shown so far is that
	\begin{equation}\label{eq:FO:CPAtoFIND}
		\Adv^{\INDCPAKEM}_{\KemExplicit}(\Ad{A}) \leq  4 \cdot \sqrt{d \cdot \PrFindInGameOne } \enspace .
	\end{equation}

	In order to take the last step towards our reduction, consider the two games given in \cref{fig:games:OWtoINDCPA:INDPKEToINDCPAKEM:StepTwo}.
	
	\begin{figure}[htb] \begin{center} \fbox{\small
				
				\nicoresetlinenr
				
				\begin{minipage}[t]{5.2cm}	
					
					\underline{\bfgameseq{2}{3}} 
					\begin{nicodemus}
						
						\item $(\pk,\sk) \leftarrow \KG$
						\item $m^* \uni \MSpace$
						\item $c^* \leftarrow \Encrypt(\pk, m^*)$ \gcom{$G_2$}
						\item $\tilde{m} \uni \MSpace$ \gcom{$G_3$}
						\item $c^* \leftarrow \Encrypt(\pk, \tilde{m})$ \gcom{$G_3$}
						\item $K^* \uni \KeySpace$
						\item $b'\leftarrow \Ad{A}^{\puncture{\SupOr}{ \lbrace m^* \rbrace}}(\pk, c^*,K^*)$
						\item \pcif $\FIND_{m^*}$ \pcreturn 1
						
					\end{nicodemus}
					
				\end{minipage}
				
				\quad 
				
				\begin{minipage}[t]{6cm}
					
					\underline{\textbf{Reduction} $B^1_{\INDCPA}(\pk)$} 
					\begin{nicodemus}
						
						\item $m^*, \tilde{m} \uni \MSpace$
						
						\item \pcreturn $(m^*, \tilde{m}, \state := m^*)$
						
					\end{nicodemus}
					
					\ \\
					
					\underline{\textbf{Reduction} $B^2_{\INDCPA}(\pk, c^*, \state = m^*)$} 
					\begin{nicodemus}
						
						\item $K^* \uni \KeySpace$
						
						\item $b'\leftarrow \Ad{A}^{\puncture{\SupOr}{ \lbrace m^* \rbrace}}(\pk, c^*,K^*)$
						
						\item \pcif $\FIND_{m^*}$ \pcreturn 1
						
					\end{nicodemus}
					
				\end{minipage}
				
			}	
		\end{center}
		\caption{\gameseq{2}{3} and \INDCPA reduction $B_{\INDCPA} = (B^1_{\INDCPA}, B^2_{\INDCPA})$
				for the proof of \cref{thm:INDPKEToINDCPAKEM}.}			
		\label{fig:games:OWtoINDCPA:INDPKEToINDCPAKEM:StepTwo}
	\end{figure}
	
	{\bfgame{2} exactly formalises \PrFindInGameOne.
		We cleaned up some variables that are not needed any longer - since $r^*$ is uniformly random in \game{1} and since it will be used nowhere but in line \ref{line:Encrypt} (of \game{1}), we can drop it altogether and simply write $c^* \leftarrow \Encrypt(\pk, m^*)$ instead.
		Similarly, since $K_0^*$ is uniformly random in \game{1} (as is $K_1^*$), we do not need to distinguish between 
		$K_0^*$ and $K_1^*$ any longer, thereby also rendering bit $b$ redundant.
	}
	
	\begin{equation}\label{eq:FO:FindToGameTwo}
		\PrFindInGameOne = \Adv^\bfgame{2} \enspace .
	\end{equation}

	{In \bfgame{3}, we replace $c^*$ with an encryption of another random message,
		while sticking with puncturing the oracle on $m^*$.
		With this change, $m^*$ becomes independent of \Ad{A}'s input, and using \cref{eq:OWTH:Find-to-Extract-Independent} from \cref{thm:OWTH:Find-to-Extract-Independent} yields
		\begin{equation}\label{eq:FO:FindToExtractIndependent}
			\Adv^\bfgame{3} \leq \frac{4q}{|\MSpace|} \enspace .
		\end{equation}
		
		To upper bound $|\Adv^\bfgame{2} - \Adv^\bfgame{3}|$,
		consider the reduction given in \cref{fig:games:OWtoINDCPA:INDPKEToINDCPAKEM:StepTwo}.
		Since $B_{\INDCPA}$ perfectly simulates either \game{2} or \game{3}, depending on which message is encrypted in its $\INDCPA$ challenge,
		\begin{equation}\label{eq:FO:DistToCPA}
			|\Adv^\bfgame{2} - \Adv^\bfgame{3}|	=  \Adv^{\INDCPA}_{\PKE}(\Ad{B_{\INDCPA}}) \enspace .
		\end{equation}
		
		Combining equations (\ref{eq:FO:FindToGameTwo}), (\ref{eq:FO:FindToExtractIndependent}) and (\ref{eq:FO:DistToCPA}) yields
		\begin{equation}\label{eq:FO:FindToCPA}
			\PrFindInGameOne \leq \Adv^{\INDCPA}_{\PKE}(\Ad{B_{\INDCPA}}) + \frac{4q}{|\MSpace|} \enspace .
		\end{equation}
	}

	Plugging \cref{eq:FO:FindToCPA} into \cref{eq:FO:DistToFind} and using that $d\leq q$ yields the bound claimed in \cref{thm:INDPKEToINDCPAKEM}. The statement about $B_{\INDCPA}$'s runtime is straightforward.
		
\qed
\end{proof}

\ifTightOnSpace %
\ifTightOnSpace %
	We proceed to proving \cref{thm:OWPKEToINDCPAKEM}, which we repeat for easier reference:
	
	\OWPKEToINDCPAKEM*
\fi

\begin{proof}
Let \Ad{A} again be an adversary against the \INDCPA security of $\KemExplicit$,
issuing random oracle queries of query depth $d$, and $q$ many in total.
Defining \game{0} to \game{2} exactly like in the proof of \cref{thm:INDPKEToINDCPAKEM} and combining \cref{eq:FO:CPAtoFIND} and \cref{eq:FO:FindToGameTwo}, we obtain
\begin{equation} \label{eq:FO:CPAtoGameTwo}
	\Adv^{\INDCPA}_{\KemExplicit}(\Ad{A}) \leq  4 \cdot \sqrt{d \cdot \Adv^\bfgame{2} } \enspace .
\end{equation}
To bound $\Adv^\bfgame{2}$, we use \cref{thm:OWTH:Find-to-Extract}
to relate $ \Adv^\bfgame{2}$ to the \OWCPA advantage of an algorithm that extracts $m^*$ from the oracle queries:
In order to relate $\Adv^\bfgame{2}$ to \OWCPA security using \cref{thm:OWTH:Find-to-Extract}, consider reduction $B_{\OWCPA}$ given in \cref{fig:games:OWtoINDCPA}.
$B_{\OWCPA}$ is exactly the query extractor \Extractor from \cref{thm:OWTH:Find-to-Extract} \emph{until $B_{\OWCPA}$'s last additional step},
where $B_{\OWCPA}$ randomly chooses its output from the candidate list it extracted (in line~\ref{line:pickCandidate}).
Since \game{2} exactly models the probability that \Ad{A} triggers $\FIND_{m^*}$, applying \cref{thm:OWTH:Find-to-Extract} yields
\begin{equation}\label{eq:FO:FindToExtract}
	\Adv^\bfgame{2} \leq 4 d \cdot \Pr[m^* \in \extractionSet : \extractionSet \leftarrow \Extractor(\pk, c^*)] \enspace ,
\end{equation}
where \Extractor is the query extractor from \cref{thm:OWTH:Find-to-Extract}, meaning \extractionSet is the result of running $\Ad A^{\SupOr}(\inputVar)$ until (just before) the $i$-th query, 
measuring all query input registers, and returning as \extractionSet the set of measurement outcomes.
Since $B_{\OWCPA}$ does exactly the same and then picks a random element of \extractionSet,
and since $B_{\OWCPA}$ wins if it randomly picked $m^*$ from \extractionSet,
\begin{equation}\label{eq:FO:ExtractToOW}
	\Pr[m^* \in \extractionSet : \extractionSet \leftarrow \Extractor(\pk, c^*)] \leq |\extractionSet| \cdot \Adv^{\OW}_{\PKE}(\Ad{B_{\OWCPA}}) \enspace .
\end{equation}

Combining equations (\ref{eq:FO:FindToExtract}) and (\ref{eq:FO:ExtractToOW}) yields
\begin{equation}\label{eq:FO:FindToOW}
	\Adv^\bfgame{2} \leq 4 d \cdot w \cdot \Adv^{\OW}_{\PKE}(\Ad{B_{\OWCPA}}) \enspace ,
\end{equation}
where we used that $|\extractionSet|$, the number of parallel queries issued during \Ad{A}'s $i$-th query,
can be upper bounded by $w$, the maximal query width.

Plugging \cref{eq:FO:FindToOW} into \cref{eq:FO:CPAtoGameTwo} yields the bound claimed in \cref{thm:OWPKEToINDCPAKEM}.
Again, the statement about $B_{\OWCPA}$'s runtime is straightforward.

	\begin{figure}[htb] \begin{center} 
	\resizebox{\textwidth}{!}{ 
                \fbox{\small
				
		\nicoresetlinenr
		
		\begin{minipage}[t]{4.7cm}	
			
			\underline{\bfgame{2}} 
			\begin{nicodemus}
				
				\item $(\pk,\sk) \leftarrow \KG$
				\item $m^* \uni \MSpace$
				\item $c^* \leftarrow \Encrypt(\pk, m^*)$
				\item $K^* \uni \KeySpace$
				\item $b'\leftarrow \Ad{A}^{\puncture{\SupOr}{ \lbrace m^* \rbrace}}(\pk, c^*,K^*)$
				\item \pcif $\FIND_{m^*}$ \pcreturn 1
				
			\end{nicodemus}
			
		\end{minipage}
		
		\quad
		
		\begin{minipage}[t]{7.5cm}
			
			\underline{\textbf{Reduction} $B_{\OWCPA}(\pk, c^*)$} 
			\begin{nicodemus}
				
				\item $i \uni \lbrace 1, \cdots, d \rbrace$
				
				\item $K^* \uni \KeySpace$ 
				
				\item Run $\Ad{A}^{\SupOr}(\pk, c^*, K^*)$ \\	\text{\quad }
				until its $i$-th query to $\SRO$
				
				\item $\lbrace m_1', m_2', \cdots \rbrace \leftarrow \Measure$ query input registers
				
				\item $m' \uni \lbrace m_1', m_2', \cdots \rbrace$ \label{line:pickCandidate}
				
				\item \pcreturn $m'$
				
			\end{nicodemus}
			
		\end{minipage}		
			
	}}	
	\end{center}
		\caption{\game{2} and \OWCPA reduction $B_{\OWCPA}$ for the proof of \cref{thm:OWPKEToINDCPAKEM}.}			
		\label{fig:games:OWtoINDCPA}
	\end{figure}
		
\qed
\end{proof}
 \fi

 \fi

\begin{restatable}[\PKE \OWCPA \RightarrowaugQROM{\Encrypt} $\FO\lbrack\PKE\rbrack$ $\INDCPA$]{theorem}{OWPKEToINDCPAKEM} \label{thm:OWPKEToINDCPAKEM}
	\ifTightOnSpace
\else 
		For any \INDCPA adversary \Ad{A} against $\KemExplicitMess \coloneqq \FOExplicitMess[\PKE, \RO{G},\RO{H}]$ in the \augQROM{\Encrypt} that issues $q$ many queries to \SRO in total, with a query depth (width) of $d$ ($w$),
		and $q_E$ many queries to \SE, where none of them is with its challenge ciphertext.
	\fi
	there furthermore exists an \OWCPA adversary $\Ad{B_{\OWCPA}}$ such that
	\[\Adv^{\INDCPA}_{\KemExplicitMess}(\Ad{A}) \leq 8d \cdot \sqrt{ w \cdot \Adv^{\OW}_{\PKE}(\Ad{B_{\OWCPA}}) }
		\ifTightOnSpace\else . \fi 
	\]
	\ifTightOnSpace\else The running time and quantum memory footprint of $\Ad{B_{\OWCPA}}$ satisfy \fi 
	$\Time(\Ad{B_{\OWCPA}})=\Time(\Ad{A})+\Time(\SupOr,q, q_E)$ and $\QMem(\Ad{B_{\OWCPA}})=\Time(\Ad{A})+\QMem(\SupOr,q, q_E)$.
\end{restatable}

\ifTightOnSpace
\else In a nutshell, the proof proceeds by going through \fi
exactly the same steps as the one of \cref{thm:INDPKEToINDCPAKEM},
up to the point where we bound $\Pr[\FIND_{m^*}]$.
To bound $\Pr[\FIND_{m^*}]$, we use \ifTightOnSpace\else the 'Finding to Extracting' \fi \cref{thm:OWTH:Find-to-Extract}
to relate $\Pr[\FIND_{m^*}]$ to the \OWCPA advantage of an algorithm that extracts $m^*$ from \Ad{A}'s oracle queries.

\ifTightOnSpace \else %
\ifTightOnSpace %
	We proceed to proving \cref{thm:OWPKEToINDCPAKEM}, which we repeat for easier reference:
	
	\OWPKEToINDCPAKEM*
\fi

\begin{proof}
Let \Ad{A} again be an adversary against the \INDCPA security of $\KemExplicit$,
issuing random oracle queries of query depth $d$, and $q$ many in total.
Defining \game{0} to \game{2} exactly like in the proof of \cref{thm:INDPKEToINDCPAKEM} and combining \cref{eq:FO:CPAtoFIND} and \cref{eq:FO:FindToGameTwo}, we obtain
\begin{equation} \label{eq:FO:CPAtoGameTwo}
	\Adv^{\INDCPA}_{\KemExplicit}(\Ad{A}) \leq  4 \cdot \sqrt{d \cdot \Adv^\bfgame{2} } \enspace .
\end{equation}
To bound $\Adv^\bfgame{2}$, we use \cref{thm:OWTH:Find-to-Extract}
to relate $ \Adv^\bfgame{2}$ to the \OWCPA advantage of an algorithm that extracts $m^*$ from the oracle queries:
In order to relate $\Adv^\bfgame{2}$ to \OWCPA security using \cref{thm:OWTH:Find-to-Extract}, consider reduction $B_{\OWCPA}$ given in \cref{fig:games:OWtoINDCPA}.
$B_{\OWCPA}$ is exactly the query extractor \Extractor from \cref{thm:OWTH:Find-to-Extract} \emph{until $B_{\OWCPA}$'s last additional step},
where $B_{\OWCPA}$ randomly chooses its output from the candidate list it extracted (in line~\ref{line:pickCandidate}).
Since \game{2} exactly models the probability that \Ad{A} triggers $\FIND_{m^*}$, applying \cref{thm:OWTH:Find-to-Extract} yields
\begin{equation}\label{eq:FO:FindToExtract}
	\Adv^\bfgame{2} \leq 4 d \cdot \Pr[m^* \in \extractionSet : \extractionSet \leftarrow \Extractor(\pk, c^*)] \enspace ,
\end{equation}
where \Extractor is the query extractor from \cref{thm:OWTH:Find-to-Extract}, meaning \extractionSet is the result of running $\Ad A^{\SupOr}(\inputVar)$ until (just before) the $i$-th query, 
measuring all query input registers, and returning as \extractionSet the set of measurement outcomes.
Since $B_{\OWCPA}$ does exactly the same and then picks a random element of \extractionSet,
and since $B_{\OWCPA}$ wins if it randomly picked $m^*$ from \extractionSet,
\begin{equation}\label{eq:FO:ExtractToOW}
	\Pr[m^* \in \extractionSet : \extractionSet \leftarrow \Extractor(\pk, c^*)] \leq |\extractionSet| \cdot \Adv^{\OW}_{\PKE}(\Ad{B_{\OWCPA}}) \enspace .
\end{equation}

Combining equations (\ref{eq:FO:FindToExtract}) and (\ref{eq:FO:ExtractToOW}) yields
\begin{equation}\label{eq:FO:FindToOW}
	\Adv^\bfgame{2} \leq 4 d \cdot w \cdot \Adv^{\OW}_{\PKE}(\Ad{B_{\OWCPA}}) \enspace ,
\end{equation}
where we used that $|\extractionSet|$, the number of parallel queries issued during \Ad{A}'s $i$-th query,
can be upper bounded by $w$, the maximal query width.

Plugging \cref{eq:FO:FindToOW} into \cref{eq:FO:CPAtoGameTwo} yields the bound claimed in \cref{thm:OWPKEToINDCPAKEM}.
Again, the statement about $B_{\OWCPA}$'s runtime is straightforward.

	\begin{figure}[htb] \begin{center} 
	\resizebox{\textwidth}{!}{ 
                \fbox{\small
				
		\nicoresetlinenr
		
		\begin{minipage}[t]{4.7cm}	
			
			\underline{\bfgame{2}} 
			\begin{nicodemus}
				
				\item $(\pk,\sk) \leftarrow \KG$
				\item $m^* \uni \MSpace$
				\item $c^* \leftarrow \Encrypt(\pk, m^*)$
				\item $K^* \uni \KeySpace$
				\item $b'\leftarrow \Ad{A}^{\puncture{\SupOr}{ \lbrace m^* \rbrace}}(\pk, c^*,K^*)$
				\item \pcif $\FIND_{m^*}$ \pcreturn 1
				
			\end{nicodemus}
			
		\end{minipage}
		
		\quad
		
		\begin{minipage}[t]{7.5cm}
			
			\underline{\textbf{Reduction} $B_{\OWCPA}(\pk, c^*)$} 
			\begin{nicodemus}
				
				\item $i \uni \lbrace 1, \cdots, d \rbrace$
				
				\item $K^* \uni \KeySpace$ 
				
				\item Run $\Ad{A}^{\SupOr}(\pk, c^*, K^*)$ \\	\text{\quad }
				until its $i$-th query to $\SRO$
				
				\item $\lbrace m_1', m_2', \cdots \rbrace \leftarrow \Measure$ query input registers
				
				\item $m' \uni \lbrace m_1', m_2', \cdots \rbrace$ \label{line:pickCandidate}
				
				\item \pcreturn $m'$
				
			\end{nicodemus}
			
		\end{minipage}		
			
	}}	
	\end{center}
		\caption{\game{2} and \OWCPA reduction $B_{\OWCPA}$ for the proof of \cref{thm:OWPKEToINDCPAKEM}.}			
		\label{fig:games:OWtoINDCPA}
	\end{figure}
		
\qed
\end{proof}
 \fi

\section{Characterizing $\FFPCPA_{\PKEDerand}$}\label{sec:QROM:CPA-to-passive:FFPCPA} \label{sec:PKEDerand:FFPCPA}

While it may very well be that the maximal success probability in game \FFPCPA for \PKEDerand can already be bounded
for particular instantiations of \PKEDerand without too much technical overhead, even in the \augQROM{\Encrypt}, this section offers an alternative way to bound this probability:
In \cref{thm:PKEDerand:FFPCPA}, we relate the success probability in game \FFPCPA for \PKEDerand to two failure-related success probabilities that are easier to analyze. This reduction separates the \emph{computationally hard} problem of exploiting knowledge of the public key to find failing ciphertexts for \PKE, from the \emph{statistically hard} problem of searching  the QRO $\RO G$ for failing plaintexts $m$ for \PKEDerand \emph{without knowledge of the key}. 

We begin by defining these two new notions related to decryption failures:
In \cref{fig:def:FFPNoKey} we define a new variant of the \FFP game  that differs from game \FFPCPA by providing \Ad{A} not even with the public key. Since the adversary obtains \underline{N}o \underline{K}ey whatsoever, the game is called \FFPNoKey, and we define the advantage of an \FFPNoKey adversary \Ad{A} against \PKE as
\[\Adv^{\FFPNoKey}_{\PKE}(\Ad{A}) \coloneqq \Pr [\FFPNoKey^{\Ad{A}}_\PKE \Rightarrow 1 ] \enspace .\]
Furthermore, we define a \underline{F}ind \underline{n}on-\underline{g}enerically \underline{F}ailing \underline{P}laintext (\FngFPCPA) game, also in \cref{fig:FngFP}. %
In this game, the adversary gets a public key $\pk_0$ as input and is allowed to issue a single message-randomness pair to a \underline{F}ailure \underline{C}hecking \underline{O}racle \FCO that is defined either relative to $(\sk_0,\pk_0)$, the key pair whose public key constitutes \Ad{A}'s input, or relative to a key pair $(\sk_1,\pk_1)$ which is an independent key pair.
We define the advantage of an \FngFPCPA adversary \Ad{A} against \PKE as
\ifTightOnSpace
\else
\[\Adv^{\FngFPCPA}_{\PKE}(\Ad{A}) \coloneqq \left|\Pr [\FngFPCPA^{\Ad{A}}_\PKE \Rightarrow 1 ] - \frac{1}{2} \right| \enspace .\]
\fi
While the game is formalized as an oracle distinguishing game, \Ad{A} can only win the game with an advantage over random guessing if it queries oracle \FCO on a message-randomness pair that fails with a different probability with respect to key pair $(\sk_0,\pk_0)$ than with respect to key pair $(\sk_1,\pk_1)$, \ifTightOnSpace\else the latter being \fi a key pair about which \Ad{B} can only gather information by its query to \FCO. We expect this game to be a more palatable target for both provable security and cryptanalysis compared to $\FFPCPA_{\PKEDerand} $ or  correctness-related games from the existing literature.

\begin{figure}[tb]\begin{center}
		
		\nicoresetlinenr
		
		\fbox{\small
			
			\begin{minipage}[t]{3cm}	
				
				\underline{{\bf Game} \FFPNoKey}
				\begin{nicodemus}
					\item $m \leftarrow \Ad{A}$
					\item $(\pk, \sk) \leftarrow \KG$
					\item $c \coloneqq \Encrypt(\pk, m)$
					\item $m' \coloneqq \Decrypt(\sk,c)$
					\item \pcreturn $\bool{m' \neq m}$
				\end{nicodemus}
				
			\end{minipage}
		
			\quad 
			
			\begin{minipage}[t]{3.5cm}	
				
				\underline{\textbf{Game} $\FngFPCPA$}
				\begin{nicodemus}
					\item $(\sk_0,\pk_0)\leftarrow\KG$
					\item $(\sk_1,\pk_1)\leftarrow\KG$
					\item $b\leftarrow\{0,1\}$
					\item $b'\leftarrow\A^{\FCO_b}(\pk_0)$
					\item \pcreturn $\bool{b=b'}$
				\end{nicodemus}
				
			\end{minipage}
			\;
			
			\begin{minipage}[t]{4cm}	
				
				\underline{$\FCO_b(m;r)$} \gcom{one query}
				\begin{nicodemus}
					\item $c\leftarrow \Encrypt(\pk_b, m;r)$
					\item $m' \coloneqq \Decrypt(\sk_b, c)$
					\item \pcreturn $\bool{m\neq m'}$
				\end{nicodemus}
			\end{minipage}
			
		}	
		\caption{
			Key-independent game \FFPNoKey for deterministic schemes \PKE, and the find non-generically failing ciphertexts game \FngFPCPA, for \PKE. \Ad{A} can make at most one query to $\FCO_{\sk_b}$.
			\tinka{Optional: Switch the \FngFPCPA notion from 'guess' to 'left-or-right' $\Rightarrow$ Even simpler proof of \cref{thm:QFngFPCPA} that does not lose a factor of 2 anymore. (If we switch: also update \cref{sec:final-result}.)}
		}
		\label{fig:FngFP} \label{fig:def:FFPNoKey}
	\end{center}
\end{figure}

\begin{restatable}[\PKE \FngFPCPA and \PKEDerand \FFPNoKey $\Rightarrow$ \PKEDerand \FFPCPA]{theorem}{PKEDerandFFPCPA}\label{thm:QFngFPCPA}\label{thm:PKEDerand:FFPCPA}
	Let \PKE be a public-key encryption scheme.
	For any \FFPCPA adversary \Ad{A} in the \augQROM{\Encrypt} against \PKEDerand making $q_{R}$ and $q_E$ queries to $\SRO$ and $\SE$, respectively, there exist an \FFPNoKey adversary \Ad{C} in the \augQROM{\Encrypt} against \PKEDerand and an \FngFPCPA adversary \Ad{B} against \PKE with %
	\[
		\Adv^{\FFPCPA}_{\PKEDerand}(\Ad{A})\le 2 \Adv^{\FngFPCPA}_{\PKE}(\Ad{B}) + \Adv^{\FFPNoKey }_{\PKEDerand}(\Ad{C}) \enspace .
	\]
The running time of $\Ad C$ is about that of $\Ad A$, $\Time(\Ad B)=\Time(A)+\Time(\SupOr, q_{RO}, q_E)$ and  $\QMem(\Ad B)=\QMem(A)+\QMem(\SupOr, q_{RO}, q_E)$.%
\end{restatable}
\ifTightOnSpace
\else
\ifTightOnSpace
	\section{Proof of \cref{thm:PKEDerand:FFPCPA} (From FngFPCPA and \FFPNoKey to \FFPCPA) } \label{PKEDerand:FFPCPA}
	
	For easier reference, we repeat the statement of \cref{thm:PKEDerand:FFPCPA}.
	
	\PKEDerandFFPCPA*
	
\fi

\begin{proof}
	By definition of the \FFPCPA advantage, we have
	\begin{align*}
		\Adv^{\FFPCPA}_{\PKEDerand}(\Ad{A}) =\Pr_{m\leftarrow \Ad A^{\SupOr}(\pk)}[(m, \SRO(m))\text{ fails wrt. }(\sk, \pk)]
		\enspace .
	\end{align*}	
	To upper bound this probability, we begin by defining \FngFPCPA adversary \Ad{B}:
	On input $\pk$, \Ad{B} runs $\Ad{A}(\pk)$, simulating \SupOr to \Ad{A}. %
	When \Ad{A} finishes by outputting its message $m$, \Ad{B} computes $r \coloneqq \SRO(m)$,
	uses its failure-checking oracle to compute $b' \coloneqq \FCO_b(m,r)$ and outputs $b'$.
	In the case where the challenge bit $b$ of \Ad{B}'s \FngFPCPA game is 0, \Ad{B} perfectly simulates the \FFPCPA game to \Ad{A} and wins iff \Ad{A} wins in game \FFPCPA.
	Therefore,
	\begin{align*}
		\Pr_{m\leftarrow \Ad A^{\SupOr}(\pk)} & [(m, \SRO(m))\text{ fails wrt. }(\sk, \pk)]
			=\Pr[1\leftarrow\Ad B(\pk)|b=0]\\
		&\le \Pr[1\leftarrow\Ad B(\pk)|b=1] + 2\Adv^{\FngFPCPA}_{\PKE}(\Ad{B})
		\enspace ,
	\end{align*}
	where the last line used the definition of the \FngFPCPA advantage.
	
	To upper bound $\Pr[1\leftarrow\Ad B(\pk)|b=1]$, note that this probability formalizes \Ad{A} outputting a message that fails to decrypt, but under an independently drawn key pair $(\sk', \pk')$:
	\begin{align} \label{eq:FFP-no-CPA}
		\Pr[1\leftarrow\Ad B(\pk)|b=1]
		&=\ifTightOnSpace \!\!\!\! \fi\Pr_{m\leftarrow \Ad A^{\SupOr}(\pk)}[(m, \SRO(m))\text{ fails wrt. }(\sk', \pk')]\ifTightOnSpace \else	\enspace \fi,
	\end{align}
	where the probability is taken additionally over $(\sk', \pk')\leftarrow \KG$.
	
	To upper bound this probability, we define \FFPNoKey adversary $\Ad{C}^\SupOr$ against \PKEDerand:
	Upon initialisation, \Ad{C} computes a key pair $(\pk, \sk)$ on its own and runs $\Ad{A}^\SupOr(\pk)$.
	When \Ad{A} finishes by outputting its message $m$, \Ad{C} forwards the message to its own game.
	Since \Ad{C} perfectly simulates the game in \cref{eq:FFP-no-CPA} to \Ad{A} and wins iff \Ad{A} wins,
	\[
		\Pr_{m\leftarrow \Ad A^{\SupOr}(\pk)}[(m, \SRO(m))\text{ fails wrt. }(\sk', \pk')]
		= \Adv^{\FFPNoKey }_{\PKEDerand}(\Ad{C}) \enspace .
	\]
\vspace{-1cm}\\
\qed
\end{proof} \fi

\subsection{Characterizing $\FFPNoKey_{\PKEDerand}$}\label{sec:PKEDerand:FFPNoKey}

In the last section, we have related the success probability of an adversary in game \FFPCPA for \PKEDerand to the success property of an adversary in game \FFPNoKey for \PKEDerand, in the \augQROM{\Encrypt}.
Intuitively, an adversary in game \FFPNoKey will succeed %
if it can find oracle inputs $m$ such that $m$ and $r \coloneqq \SRO(m)$ satisfy a certain predicate%
, i.e., the predicate that $(m,r)$ fails with respect to $\pk$.
To prove the upper bound we provide in \cref{thm:PKEDerand:FFPNoKey}, we therefore generically bound the success probability for a certain search problem in \cref{sec:QROM:FindLargeValues}.
While we note that the search bound might be of independent interest, it in particular allows us to characterize the maximal advantage in game \FFPNoKey in terms of two statistical values for the underlying randomised scheme \PKE.%
\ifTightOnSpace
\else

\fi
We begin with the definitions of \deltaRandKey and $\sigma_\deltaRandKey$: Below, we define the worst-case decryption error rate \deltaRandKey \emph{under independent keys}, and the maximal variance of the decryption error rate $\sigma_\deltaRandKey$.

\begin{definition}[worst-case independent-key decryption error rate, maximal decryption error variance%
	]
	We define the \emph{worst-case decryption error rate under independent keys} $\deltaRandKey$ and the \emph{maximal decryption error variance under independent keys} $\sigma_{\deltaRandKey}$ of a public-key encryption scheme \PKE as
	\begin{align*}
		\deltaRandKey&\coloneqq \max_{m\in\MSpace}[\Pr_{(\sk,\pk),r}[(m,r)\text{ fails}]]
		=\max_{m\in\MSpace}\mathbb E_r[\Pr_{(\sk,\pk)}[(m,r)\text{ fails}]] \enspace ,\text{ and}\\
		\sigma^2_\deltaRandKey& \coloneqq \max_{m\in\MSpace}\mathbb V_r[\Pr_{(\sk,\pk)}[(m,r)\text{ fails}]]
	\end{align*}
\ifTightOnSpace\else	for uniformly random $r$.\fi%
\end{definition}

We want to stress that \emph{\deltaRandKey differs from the worst-case term \deltaWorstCase } that was introduced in \cite{TCC:HofHovKil17} (there denoted by $\delta$) since \deltaWorstCase is defined by
\[ 
	\deltaWorstCase \coloneqq \mathbb E_\KG \max_{m \in \MSpace} \Pr_{r \uni \RSpace} [(m,r)\text{ fails}]
\enspace . \]
Intuitively, \deltaWorstCase is the best possible advantage of an an adversary, trying to find the message most likely to fail for a given key pair, while for \deltaRandKey, %
the key pair will be randomly sampled \emph{after} the adversary had made its choice $m$. On a formal level, it is easy to verify that \deltaWorstCase serves as an upper bound for \deltaRandKey.

\begin{theorem}[Upper bound for \FFPNoKey of \PKEDerand]\label{thm:PKEDerand:FFPNoKey}
	Let \PKE be a public-key encryption scheme with worst-case independent-key decryption error rate \deltaRandKey and decryption error rate variance $\sigma_\deltaRandKey$.
	For any \FFPNoKey adversary \Ad{A} in the \augQROM{\Encrypt} against \PKEDerand, setting $C=304$, we have that
\ifTightOnSpace
\else
	\begin{equation*}
			\Adv^{\FFPNoKey }_{\PKEDerand}(\Ad{A})
			\leq \deltaRandKey + 3\sqrt Cq\sigma_\deltaRandKey
			 					+ 2Cq^2\sigma_\deltaRandKey^2\deltaRandKey\left(-\log\sqrt Cq\sigma_\deltaRandKey\right)
			\enspace ,
	\end{equation*}
\fi
\end{theorem}

In \ifCameraReady the full version\else \cref{sec:envelope} \fi, we give an alternative bound that grows with the logarithm of the number of RO queries, assuming a \emph{Gaussian \ifTightOnSpace\else-shaped \fi tail bound} for the decryption error \ifTightOnSpace \else probability\fi distribution.
\begin{proof}
	The claimed bounds result from applying \cref{cor:optbound-chebyshev,cor:optbound-gaussian-tail} that we give in \cref{sec:QROM:FindLargeValues} below:
	The success probability of \Ad{A} in game \FFPNoKey is the probability that \Ad{A}'s output message fails to decrypt.
	If we define function $F$ by setting $F(m,r) \coloneqq \Pr[(m;r) \text{ fails wrt. }(\sk, \pk)]$,
	we can alternatively describe \Ad{A}'s task as the task to find a superposition oracle input $m$ such that $F(m, \SRO(m))$
	is large, having access to %
	 the extractable oracle simulation \SupOr %
	 .
	We prove general upper bounds for the success probability in finding large values for arbitrary functions $F$ in \cref{sec:QROM:FindLargeValues}, %
	which immediately yields the claimed  bound. %
\qed\end{proof}

\subsection{Finding large values of a function in the \augQROM{f}}\label{sec:QROM:FindLargeValues}

In this section, we provide the technical results for the \augQROM{f} that we need to prove \cref{thm:PKEDerand:FFPNoKey}.
Throughout this section, $f$ is a fixed function such that \augQROM{f} is well-defined.
We begin by providing a bound for the success probability of an %
 algorithm in the \augQROM{f} that searches for a value $x$ that, %
together with its oracle value $\SRO(x)$, satisfies a relation $R$.
In the lemma below \ifTightOnSpace\else that provides this upper bound\fi, we will use the quantity $\Gamma_R$ that was defined in \cref{eq:Gamma_R} (see page~\pageref{eq:Gamma_R}).
\begin{lemma}[Slight generalization of {\cite[Proposition 3.5]{DFMS21}}]\label{lem:search-in-aug-QROM}
	Let $R\subset \mathcal X\times\mathcal Y$ be a relation and $\Ad A^{\SupOr}$ an algorithm with access to \augQRO{f} from $\mathcal X$ to $\mathcal Y$ for some function $f:\mathcal X\times \mathcal Y\to\mathcal T$, making $q$ queries to $\SRO$%
	. Then
\ifTightOnSpace
\else
	\begin{equation}
		\Pr_{x\leftarrow\Ad A^{\SupOr}}[R(x, \SRO(x))]\le 152 (q+1)^2\ifTightOnSpace\Gamma_R|\mathcal Y|^{-1} \else \frac{\Gamma_R}{|\mathcal Y|}\fi ,
	\end{equation}	
\fi
independently of the number of queries $\Ad A$ makes to $\SE$. Here it is understood that $\SRO$ is queried once in the very end to determine $\SRO(x)$.
\end{lemma}
\ifTightOnSpace \else The generalization consists of allowing $\Ad A$ to query $\SE$ as well.\fi
\begin{proof}
The only difference between \cite[Proposition 3.5]{DFMS21} and \cref{lem:search-in-aug-QROM} is that \Ad{A} now additionally has access to \SE.
	The proof is thus the same as %
	 for  \cite[Proposition 3.5]{DFMS21}, with the additional observation that queries to \SE commute with the progress measure operator $M$ for any relation $R$. This is because i) both $M$ and the operator applied upon an \SE query are controlled unitaries controlling on the database register of the compressed oracle database of the \augQRO{f}, and ii) the target registers of $M$ and \SE are disjoint.
\qed
\end{proof}

According to \cref{lem:search-in-aug-QROM}, it is hard to search a random oracle,  even given extraction %
access. We will now use \cref{lem:search-in-aug-QROM} to show that it is also hard to produce an input to the oracle so that the resulting input-output pair has a large value under a function $F$, in expectation%
.
To state a theorem making this intuition precise and quantitative, let $F:X\times Y\to I\subset [0,1]$,
and let $I$ be ordered as $I=\{t_1,...,t_R\}$ with $t_i>t_{i-1}$.
The hardness of the task of finding large values is related to a ``tail bound'' $G(t)$ for %
 the probability of $F(x,r)$ being larger than $t$
.

\begin{theorem}\label{thm:opt-augQROM}
	Let $F$ and $I$ be as above. Let further $G:[0,1]\to[0,1]$ be non-increasing such that $G(t)\ge \Pr_{r\leftarrow Y}[F(x,r)\ge t]$ for all $x$. 
	Let $C \coloneqq 304$, $\Delta G(i) \coloneqq G(t_i)-G(t_{i+1})$ (setting formally $G(t_{R+1})=0$),
	and let $\kappa_q \coloneqq  \min\{i|C q^2G(t_i)\le1\}$.
	Then for any algorithm $\Ad A^{\SupOr}$ making at most $q\ge 1$ %
	queries to $\SRO$, 
	\begin{equation}
			\mathbb E_{x\leftarrow \Ad A^\SupOr}[F(x,\SRO(x))]\le t_{\kappa_q}+C q^2\sum_{i=\kappa_q+1}^Rt_i\Delta G(i)
			\enspace .
	\end{equation}
 $\SRO$ is queried once in the %
end to determine $\SRO(x)$.
\end{theorem}
\begin{proof}
	Let $x\leftarrow \A^\SupOr$. %
	We bound
	\ifTightOnSpace
\else
	\begin{align*}
		\mathbb E\left[F(x, \SRO(x))\right]&=\sum_{i=1}^Rt_i\Pr[F(x,\SRO(x))=t_i]\\
		&=\sum_{i=1}^Rt_i\left(\Pr[F(x,\SRO(x))\ge t_i]-\Pr[F(x, \SRO(x))\ge t_{i+1}]\right)\\
		&=t_1+\sum_{i=2}^R\Pr[F(x, \SRO(x))\ge t_i](t_i-t_{i-1})\\
		&\le t_1+\sum_{i=2}^R\min(1,Cq^2G(t_i))(t_i-t_{i-1})\\
		&=t_{\kappa_q}+Cq^2\sum_{i=\kappa_q+1}^RG(t_i)(t_i-t_{i-1}),
	\end{align*}
\fi
where we have used \cref{lem:search-in-aug-QROM} with the relation $R_{f, \ge t_i}$ defined by $R_{f, \ge t_i}(x,y):\Leftrightarrow f(x,y)\ge t_i$ in the second-to-last line.
\ifTightOnSpace
\else 

\fi
 We further bound
 \ifTightOnSpace
\else
\begin{align*}
	\sum_{i=\kappa_q+1}^RG(t_i)(t_i-t_{i-1})&=-G(t_{\kappa_q+1})t_{\kappa_q}+\sum_{i=\kappa_q+1}^Rt_i\Delta G(i)\\
&\le \sum_{i=\kappa_q+1}^Rt_i\Delta G(i),
\end{align*}
finishing the proof.
\fi
\qed\end{proof}
We provide a corollary for the case where $G$ is given by Chebyshev's inequality%
.
\begin{corollary}\label{cor:optbound-chebyshev}
	Let $F$, $I$, and $C$ be as in \cref{thm:opt-augQROM}, and let the expectation values and variances of $F(x,r)$ for  random $r\leftarrow\mathcal Y$ be bounded as $ \mathbb E_r[F(x,r)]\le\mu$ and $\mathbb V_r[F(x,r)]\le\sigma^2$, respectively.  Then, for an algorithm $\Ad A^{\SupOr}$ making at most $q\ge 1$ quantum queries to $\SRO$,
	\ifTightOnSpace
\else
	\begin{equation}
		\mathbb E_{x\leftarrow A^{\SupOr}}[F(x, \SRO(x))]\le  \mu+3\sqrt Cq\sigma+2Cq^2\sigma^2\mu\log\frac{1}{\sqrt Cq\sigma}.
	\end{equation}
\fi
\end{corollary}
\begin{proof}
	 By Chebyshev's inequality, we can set 
\ifTightOnSpace
\else
\begin{equation}
		G(t)=\frac{\sigma^2}{(t-\mu)^2}.
	\end{equation}
\fi
	We thus obtain $t_{\kappa_q}\le \sqrt C q\sigma+\mu$. We bound
\ifTightOnSpace
\else
	\begin{align}
		\sum_{i=\kappa_q+1}^Rt_i\Delta G(i)=&- \sum_{i=\kappa_q+1}^Rt_i\int_{t_i}^{t_{i+1}}G'(t)\mathrm dt\\
		\le&-\int_{t_{\kappa_q}}^1t G'(t)\mathrm dt\\
		=&2\sigma^2\int_{t_{\kappa_q}}^1\frac t {t-\mu}\mathrm dt\\
		=&2\sigma^2\int_{t_{\kappa_q}-\mu}^{1-\mu}\frac{u+\mu} {u}\mathrm du\\
		=&2\sigma^2\left(1-t_{\kappa_q}+\mu\log\frac{1-\mu}{t_{\kappa_q}-\mu}\right).
	\end{align}
\fi
	We arrive at the bound
	\ifTightOnSpace
\else
	\begin{align*}
		\mathbb E_{x\leftarrow A^{\SupOr}}[F(x,\SRO(x))]&\le \mu+\sqrt Cq\sigma+2Cq^2\sigma^2\left(1+\mu\log\frac{1-\mu}{\sqrt Cq\sigma}\right).
	\end{align*}
\fi
If $\sqrt Cq\sigma\ge 1$, the claimed bound trivially holds, else $\sqrt Cq\sigma\ge Cq^2\sigma^2$ and thus
\ifTightOnSpace
\else
	\begin{align*}
	\mathbb E_{x\leftarrow A^{\SupOr}}[F(x,\SRO(x))]&\le \mu+3\sqrt Cq\sigma+2Cq^2\sigma^2\mu\log\frac{1-\mu}{\sqrt Cq\sigma}.
\end{align*}
\fi
\vspace{-1cm}\\
\qed\end{proof}

\ifTightOnSpace \else %
\ifTightOnSpace %
	
	\section{Alternative bound for  \FFPNoKey based on a stronger tail bound} \label{sec:envelope}
	
	In this appendix, we show how to use a stronger uniform tail bound in place of Chebyshev's inequality to obtain a stronger bound for the adversarial advantage in \FFPNoKey.
	
\else
	
	\subsection{Alternative bound for  \FFPNoKey based on a stronger tail bound} \label{sec:envelope}
	In this subsection, we show how to use a stronger uniform tail bound in place of Chebyshev's inequality to obtain a stronger bound for the adversarial advantage in \FFPNoKey.
	
\fi

We begin by defining the decryption error tail envelope.
\begin{definition}[decryption error tail envelope]
	We define the %
	\emph{decryption error tail envelope} as
	\[ 
	\tau(t) \coloneqq \max_m\Pr_{r\leftarrow \RSpace}\left[\Pr_{(\sk,\pk)}[(m,r)\text{ fails}]\ge t\right]
	\enspace . \]
\end{definition}
We obtain the following stronger bound for \FFPNoKey that scales logarithmically with the adversary's random oracle queries.

\begin{theorem}[Upper bound for \FFPNoKey of \PKEDerand]%
	\label{thm:QFFPNoKey-gauss}
	Let \PKE be a public-key encryption scheme with worst-case random-key decryption error rate \deltaRandKey and decryption error tail envelope $\tau$.
	For any \FFPNoKey adversary \Ad{A} in the \augQROM{\Encrypt} against \PKEDerand, setting $C=304$, we have that
	\begin{equation*}
		\Adv^{\FFPNoKey}_{\PKEDerand}(\Ad{A})
		\leq \deltaRandKey+2\beta^{-1/2}\sqrt{\ln(2C\sqrt\beta )+2\ln(q)}.
	\end{equation*}
\end{theorem}
The above theorem follows directly by an application of \cref{cor:optbound-gaussian-tail} given below. Combining \cref{thm:QFFPNoKey-gauss} with the reductions from \cref{sec:QROM:CCA-to-CPA-KEM,sec:QROM:CPA-to-passive} we get the following alternative to \cref{cor:main-result}.
\begin{corollary}[\PKE \FngFPCPA and pass. secure $\Rightarrow \FOExplicitMess\lbrack\PKE\rbrack$ $\INDCCA$] \label{cor:main-result-2}
	\ifTightOnSpace
	Let \PKE and \Ad{A} be like in \cref{cor:main-result:withFFPCPA}, and let \PKE furthermore have worst-case random-key decryption error rate \deltaRandKey, decryption error rate variance $\sigma_{\deltaRandKey}$ and decryption error tail envelope $\tau$.
	\else
	Let \PKE be a (randomized) \PKE scheme that is $\gamma$-spread and with worst-case random-key decryption error rate \deltaRandKey, decryption error rate variance $\sigma_{\deltaRandKey}$ and decryption error tail envelope $\tau$.
	Let \Ad{A} be an \INDCCAKEM adversary (in the QROM) against $\KemExplicitMess \coloneqq \FOExplicitMess[\PKE, \RO{G},\RO H]$, issuing at most $q_\RO{G}$ many queries to its oracle \RO{G}, $q_\RO{H}$ many queries to its oracle \RO{H}, and at most \numberDecapsQueries many queries to its decapsulation oracle \oracleDecaps. Let $q=q_\RO{G}+q_\RO{H}$, and let $d$ and $w$ be the query depth and query width of the combined queries to $\RO G$ and $\RO H$.
	\fi
	Set $C=304$ and assume $\sqrt Cq_{\RO{G}}\sigma_\deltaRandKey\le 1/2$.
	Then there exist an \INDCPA adversary $\Ad{B}_\IND$, a \OWCPA adversary $\Ad{B}_\OW$ and an \FngFPCPA adversary \Ad{C} against $\PKE$,  such that
	\begin{align}
		\Adv^{\INDCCAKEM}_{\KemExplicitMess}(\Ad{A})
		\le&\widetilde{\Adv}_{\PKE}
		+
		(\numberDecapsQueries+1) \left(2\Adv^{\FngFPCPA}_{\PKE}(\Ad{C})+\eps_{\deltaRandKey}\right)
		+\eps_{\gamma}\label{eq:QFFP3}
	\end{align}
	with
	\ifTightOnSpace
	$\widetilde{\Adv}_{\PKE}$ and $\eps_{\gamma}$ like in \cref{cor:main-result:withFFPCPA}.
	\else
	\begin{equation}
		\widetilde{\Adv}_{\PKE}
		=\begin{cases}
			4 \cdot \sqrt{ \left(d+\numberDecapsQueries\right)  \cdot \Adv^{\INDCPA}_{\PKE}(\Ad{B}_\IND)}
			+ \frac{8\left(q+\numberDecapsQueries\right)}{\sqrt{\left|\MSpace\right|}}&\text{ or}\\
			8\left(d+\numberDecapsQueries\right) \cdot \sqrt{ w \cdot \Adv^{\OW}_{\PKE}(\Ad{B}_\OW) }. &
		\end{cases}
	\end{equation}	
	\fi
	The additive error term $\eps_{\deltaRandKey}$ is given by
	\begin{equation}\label{eq:epsdelta-chebyshev}
		\eps_{\deltaRandKey} \le \deltaRandKey
		+\left(3+2\delta_{rk}\right) \sqrt Cq_{\RO{G}}\sigma_\deltaRandKey
		\enspace \ifTightOnSpace . \else , \fi
	\end{equation}
	\ifTightOnSpace \else
	and the additive error term $\eps_{\gamma}$ is given by
	\begin{equation*}
		\eps_{\gamma}=24%
		\numberDecQueries(q_{\RO G}+2\numberDecQueries)2^{-\gamma/2}+4\numberDecQueries \cdot2^{-\gamma}.
	\end{equation*}
	\fi
	
	Here, $\deltaRandKey, \sigma_\deltaRandKey$ and $\gamma$ are the worst-case random-key decryption error rate,
	the maximal decryption failure variance under random keys, and the ciphertext spreadness parameter, respectively.
	If the Gaussian tail bound 
	\begin{equation*}
		\max_m\Pr_{r\leftarrow \RSpace}\left[\Pr_{(\sk,\pk)}[\Decrypt(\sk, \Encrypt(\pk, m;r))\neq m]\ge t\right]\le \exp\left(-\beta(t-\deltaRandKey)^2\right)
	\end{equation*}
	holds for some parameter $\beta$, the dependency of $\epsilon_\deltaRandKey$ on $q_{\RO{G}}$ can be improved to
	\begin{equation}\label{eq:epsdelta-Gausstail}
		\eps_{\deltaRandKey}\le \deltaRandKey+2\eta_1\sqrt{\ln\left(\eta_2 q^2_{\RO{G}}\right)}%
	\end{equation}
	with $\eta_1=\beta ^{-1/2}$ and $\eta_2= 2C\sqrt\beta $.
	\ifTightOnSpace $\Ad{B}_\IND$'s, $\Ad{B}_\OW$'s and \Ad{C}'s running time \else The running time of the adversaries $\Ad{B}_\IND$, $\Ad{B}_\OW$ and \Ad{C} \fi are all bounded by
	\begin{equation*}
		\Time(A)+\Time(\SupOr, q_{\RO G}+q_{\RO H}+\numberDecapsQueries)+O(\numberDecapsQueries).
	\end{equation*}
\end{corollary}

We continue to prove the corollary of \cref{thm:opt-augQROM} which yields \cref{thm:QFFPNoKey-gauss}
\begin{corollary}\label{cor:optbound-gaussian-tail}
	Let $F$, $I$, and $C$ be as in \cref{thm:opt-augQROM}. Let furthermore $ \mathbb E[F(x,H(x))]\le\mu$ for some $\mu\in[0,1]$ and suppose in addition that we can set $G(t)=c \exp(-\beta (t-\mu)^2)$ with $\beta\ge e/(2C)$ .  Then, for an algorithm $\Ad A^\SupOr$ making at most $q\ge 1$ quantum queries to $\SRO$
	\begin{equation}
		\mathbb E_{x\leftarrow A^{\SupOr}}[F(x,\SRO)]\le \mu+2\beta^{-1/2}\sqrt{\ln(2C\sqrt\beta )+2\ln(q)}
	\end{equation}
\end{corollary}
\begin{proof}
	Here, we directly use \cref{lem:search-in-aug-QROM} for simplicity (a slightly tighter but less pretty bound can be obtained from \cref{thm:opt-augQROM}). For any $a\in[0,1]$, we have
	\begin{equation}
		\mathbb E_{x\leftarrow A^{\SupOr}}[F(x, \SRO(x))]\le a+\Pr_{x\leftarrow A^{\SupOr}}[F(x, \SRO(x))\ge a].
	\end{equation}
	Setting $a=\mu+\hat a$ and using the definition of $G$ as well as \cref{lem:search-in-aug-QROM} (in the same way as in the proof of \cref{thm:opt-augQROM}), we obtain
	\begin{equation}
		\Pr_{x\leftarrow A^{\SupOr}}[F(x, \SRO(x))\ge \mu+\hat a]\le Cq^2\exp(-\beta \hat a^2)
	\end{equation}
	Setting $\hat a=\sqrt{\ln(2Cq^2\sqrt{\beta})/\beta}$ and using $\ln(2Cq^2\sqrt{\beta})\ge 1$, we obtain
	\begin{align}
		\mathbb E_{x\leftarrow A^{\SupOr}}[F(x,\SRO(x))]\le &\mu+\beta^{-1/2}\left(1+\sqrt{\ln(2Cq^2\sqrt{\beta})}\right)\\
		\le &\mu+2\beta^{-1/2}\sqrt{\ln(2C\sqrt\beta )+2\ln(q)},
	\end{align}
	where $\ln$ is the natural logarithm.
	\qed\end{proof}

 \fi

\section{Tying everything together}\label{sec:final-result}

Combining the reductions from \cref{sec:QROM:CCA-to-CPA-KEM,sec:QROM:CPA-to-passive}, we obtain a first corollary that still relies on \FFPCPA of \PKEDerand .

\begin{corollary}[\PKEDerand \FFPCPA and \PKE pass. secure $\Rightarrow \FOExplicitMess\lbrack\PKE\rbrack$ $\INDCCA$] \label{cor:main-result:withFFPCPA}
	\ifTightOnSpace
\else
		Let \PKE be a (randomized) \PKE scheme that is $\gamma$-spread,
		and let \Ad{A} be an \INDCCAKEM adversary (in the QROM) against $\KemExplicit \coloneqq \FOExplicitMess[\PKE, \RO{G},\RO H]$, issuing at most $q_\RO{G}$ many queries to its oracle \RO{G}, $q_\RO{H}$ many queries to its oracle \RO{H}, and at most \numberDecapsQueries many queries to its decapsulation oracle \oracleDecaps. Let $q=q_\RO{G}+q_\RO{H}$, and let $d$ and $w$ be the query depth and query width of the combined queries to $\RO G$ and $\RO H$. 
	\fi
	Then there exist an \INDCPA adversary $\Ad{B}_\IND$, a \OWCPA adversary $\Ad{B}_\OW$ and an \FFPCPA adversary \Ad{C} against \PKEDerand in the \augQROM{\Encrypt} such that
	\begin{align}
		&\Adv^{\INDCCAKEM}_{\KemExplicitMess}(\Ad{A})
		\le\widetilde{\Adv}_{\PKE} + (\numberDecapsQueries+1) \Adv^{\FFPCPA}_{\PKE}(\Ad{C}) +\eps_{\gamma}, \text{ with }\label{eq:QFFP3}\\
	&	\widetilde{\Adv}_{\PKE}
	=\begin{cases}
			4 \cdot \sqrt{ \left(d+\numberDecapsQueries\right)  \cdot \Adv^{\INDCPA}_{\PKE}(\Ad{B}_\IND)}
			+ \frac{8\left(q+\numberDecapsQueries\right)}{\sqrt{\left|\MSpace\right|}}&\text{ or}\\
			8\left(d+\numberDecapsQueries\right) \cdot \sqrt{ w \cdot \Adv^{\OW}_{\PKE}(\Ad{B}_\OW) }. &
		\end{cases}
\end{align}
	The additive error term is given by
	\ifTightOnSpace
\else
	\begin{equation*}
		\eps_{\gamma}=24%
		\numberDecQueries(q_{\RO G}+4\numberDecQueries)2^{-\gamma/2}%
		\enspace .
	\end{equation*}
	\fi
\Ad C makes $q_{\RO G}+q_{\RO H}+\numberDecapsQueries$ queries to \SRO and $\numberDecapsQueries$ to \SE.
\ifTightOnSpace\else The running time of the adversaries $\Ad{B}_\IND$, $\Ad{B}_\OW$ and \Ad{C} \fi are bounded as $\Time(\Ad{B}_{\IND/\OW})=\Time(A)+\Time(\SupOr, q_{\RO G}+q_{\RO H}+\numberDecapsQueries)+O(\numberDecapsQueries)$ and $\Time(\Ad C)=\Time(\Ad A)+O(\numberDecapsQueries)$.
\end{corollary}

Combining \cref{cor:main-result:withFFPCPA} with \cref{thm:PKEDerand:FFPCPA} from \cref{sec:PKEDerand:FFPCPA} and \cref{thm:PKEDerand:FFPNoKey} from  \cref{sec:PKEDerand:FFPNoKey}, we now obtain our main result as a corollary.

\begin{corollary}[\PKE \FngFPCPA and pass. secure $\Rightarrow \FOExplicitMess\lbrack\PKE\rbrack$ $\INDCCA$] \label{cor:main-result}
	\ifTightOnSpace
\else
		Let \PKE be a (randomized) \PKE scheme that is $\gamma$-spread and with worst-case random-key decryption error rate \deltaRandKey, decryption error rate variance $\sigma_{\deltaRandKey}$ and decryption error tail envelope $\tau$.
		Let \Ad{A} be an \INDCCAKEM adversary (in the QROM) against $\KemExplicitMess \coloneqq \FOExplicitMess[\PKE, \RO{G},\RO H]$, issuing at most $q_\RO{G}$ many queries to its oracle \RO{G}, $q_\RO{H}$ many queries to its oracle \RO{H}, and at most \numberDecapsQueries many queries to its decapsulation oracle \oracleDecaps. Let $q=q_\RO{G}+q_\RO{H}$, and let $d$ and $w$ be the query depth and query width of the combined queries to $\RO G$ and $\RO H$.
	\fi
	Set $C=304$ and assume $\sqrt Cq_{\RO{G}}\sigma_\deltaRandKey\le 1/2$.
	\ifTightOnSpace
\else
		Then there exist an \INDCPA adversary $\Ad{B}_\IND$, a \OWCPA adversary $\Ad{B}_\OW$ and an \FngFPCPA adversary \Ad{C} against \PKE such that
	\fi
	\begin{align}
		\Adv^{\INDCCAKEM}_{\KemExplicitMess}(\Ad{A})
			\le&\widetilde{\Adv}_{\PKE}
		+
		(\numberDecapsQueries+1) \left(2\Adv^{\FngFPCPA}_{\PKE}(\Ad{C})+\eps_{\deltaRandKey}\right)
		+\eps_{\gamma}\label{eq:QFFP3}
	\end{align}
with
\ifTightOnSpace
\else
	\begin{equation}
		\widetilde{\Adv}_{\PKE}
		=\begin{cases}
			4 \cdot \sqrt{ \left(d+\numberDecapsQueries\right)  \cdot \Adv^{\INDCPA}_{\PKE}(\Ad{B}_\IND)}
			+ \frac{8\left(q+\numberDecapsQueries\right)}{\sqrt{\left|\MSpace\right|}}&\text{ and}\\
			8\left(d+\numberDecapsQueries\right) \cdot \sqrt{ w \cdot \Adv^{\OW}_{\PKE}(\Ad{B}_\OW) }. &
		\end{cases}
	\end{equation}	
\fi
The additive error term $\eps_{\deltaRandKey}$ is given by
\ifTightOnSpace
\else
\begin{equation}\label{eq:epsdelta-chebyshev}
	\eps_{\deltaRandKey} \le \deltaRandKey
				+\left(3+2\deltaRandKey\right) \sqrt Cq_{\RO{G}}\sigma_\deltaRandKey
	\enspace ,
\end{equation}
	and the additive error term $\eps_{\gamma}$ is given by
	\begin{equation*}
		\eps_{\gamma}=24%
		\numberDecQueries(q_{\RO G}+2\numberDecQueries)2^{-\gamma/2}+4\numberDecQueries \cdot2^{-\gamma}.
	\end{equation*}
\fi
\ifTightOnSpace\else The running time of the adversaries $\Ad{B}_\IND$, $\Ad{B}_\OW$ and \Ad{C} \fi is bounded by
\ifTightOnSpace
\else 
\begin{equation*}
	\Time(A)+\Time(\SupOr, q_{\RO G}+q_{\RO H}+\numberDecapsQueries)+O(\numberDecapsQueries).
\end{equation*}
\fi
\end{corollary}
In \ref{sec:envelope} we give an alternative corollary with an $\eps_{\deltaRandKey}$ that only grows logarithmically with the number of RO queries, assuming a \emph{Gaussian-shaped tail bound} for the decryption error probability distribution.
\begin{proof}
	The corollary follows by combining \cref{cor:QFFP3,thm:INDPKEToINDCPAKEM,thm:OWPKEToINDCPAKEM,thm:QFngFPCPA}. Exploiting the very mild condition $\sqrt Cq_{\RO{G}}\sigma_\deltaRandKey\le 1/2$\footnote{Without it the bound involving $\sigma_{\deltaRandKey}$ from \cref{thm:QFngFPCPA} is almost trivial} we have used the inequality $x^2/\log(x)\le x$ for $x\le 1/2$  for $x=\sqrt Cq_{\RO{G}}\sigma_\deltaRandKey$ to simplify the error term $\eps_{\deltaRandKey}$ from \cref{thm:QFngFPCPA}.
\qed\end{proof}
We remark that the two alternative bounds in \cref{eq:epsdelta-chebyshev,eq:epsdelta-Gausstail} are just examples. If, e.g., an exponential tail bound is available instead of a Gaussian one, the techniques from \cref{sec:QROM:FindLargeValues} can be used to prove a similar, intermediate bound.
The above result has two main advantages over previous theorems for the FO transformation:
\ifTightOnSpace
\else
\begin{itemize}
	\item The additive loss $(\numberDecQueries+1) \left(\Adv^{\FngFPCPA}_{\PKE}(\Ad{C})+\eps_\deltaRandKey\right)$, with the two alternative bounds for $\eps_\deltaRandKey$ given in \cref{eq:epsdelta-chebyshev,eq:epsdelta-Gausstail}, can be much smaller than the additive loss of roughly $q_{\RO{G}}^2\deltaWorstCase$ that is present in all previous bounds for the FO transformation. In particular, instead of the \emph{quadratic} dependence on the number of hash queries $q_{\RO{G}}$, the asymptotic dependence is at most \emph{linear}. %
If an appropriate tail bound can be proven, it is even logarithmic.
	\item It holds for the explicit rejection variant of the transformation, while the bounds are competitive with previous ones in the literature that were limited to the implicit rejection variant.
\end{itemize}
\fi
 
\ifTightOnSpace \else \ifCameraReady \else  %
\section{$\gamma$-Spreadness of selected NIST proposals} \label{sec:spreadness}

\cref{thm:QFFP1} provides a tight reduction of \INDCCAKEM to \INDCPAKEM and \FFPCCA, albeit at the cost of an additive error depending on the spreadness factor $\gamma$ of the underlying PKE.
In this section, we will analyze the spreadness of some of the alternates candidates of the NIST post-quantum competition.
Since this work is considered with schemes that exhibit decryption failure and get derandomized to a scheme \PKEDerand,
we do not consider ClassicMcEliece, NTRU, NTRU prime and SIKE (since they are perfectly correct) 
and BIKE (as BIKE encrypts deterministically without incorporating \PKEDerand).
We chose our two examples, \HQCPKE and \FrodoPKE, because computing $\gamma$  for these two examples requires little additional technical overhead. 
Computing $\gamma$ for other submissions to the NIST PQC standardisation process, like, e.g., Kyber or Saber, is out of the scope of this work.

If \numberDecQueries is upper bounded by $2^{64}$ as in NIST’s CFP, we can give a simpler upper bound for the term showing up in \cref{thm:QFFP1}  by computing
\begin{align*}
	\numberDecQueries \cdot (q_{\RO G} + 2\numberDecQueries) \cdot 2^{-\gamma/2}
	\leq 2^{64} \cdot (q_{\RO G} + 2^{65}) \cdot 2^{-\gamma/2}
	\leq q_{\RO G} \cdot 2^{65-\gamma/2}
	\enspace .
\end{align*}

The following lemma makes the bound above explicit for \FrodoPKE.
\begin{restatable}[$\gamma$-Spreadness of \FrodoPKE]{lemma}{SpreadnessFrodo} \label{lem:Spreadness:Frodo}
	\FrodoPKE-\instance is $\gamma$-spread for
	\begin{align*}
		\gamma = \begin{cases}
			10752 & \instance = 1344\\
			15616 & \instance = 976 \\
			10240 & \instance = 640
		\end{cases}
	\enspace ,
	\end{align*}
	hence
	\begin{align*}
		q_{\RO G} \cdot 2^{65-\gamma/2}
			\leq  \begin{cases}
			q_{\RO G} \cdot 2^{- 5311} & \instance = 1344\\
			q_{\RO G} \cdot 2^{- 7743} & \instance = 976 \\
			q_{\RO G} \cdot 2^{- 5055} & \instance = 640
		\end{cases}
		\enspace .
	\end{align*}
\end{restatable}

\begin{proof}
	Let $(\pk = (\seedA, B), \sk) \in \supp(\FrodoKG)$, let $m \in \FrodoMSpace$, and let $c = (B', V') \in \FrodoCSpace$.
	According to the definition of \FrodoEnc, we have that
	\begin{align*}
		\Pr_\FrodoEnc[& \FrodoEnc(\pk, m) = (B', V')] \\
			 &= \Pr_{S', E' \leftarrow \chi^{\overline{m} \times n}, E'' \leftarrow \chi^{\overline{m} \times \overline{n}}}
					[S'A + E' = B' \ \wedge \ S'B + E'' + \FrodoEncode(m) =  V']
			\\ &
			\leq \Pr_{S', E' \leftarrow \chi^{\overline{m} \times n}} [S'A + E' = B']
			\\ &
			= \sum_{s' \in \supp(\chi^{\overline{m} \times n})}
						\Pr_{S', E' \leftarrow \chi^{\overline{m} \times n}} [S'A + E' = B' \wedge S' = s'] 
			\\ &
			= \sum_{s' \in \supp(\chi^{\overline{m} \times n})}
				\Pr_{E' \leftarrow \chi^{\overline{m} \times n}} [s'A + E' = B' ] 
					\cdot \Pr_{S' \leftarrow \chi^{\overline{m} \times n}} [S' = s'] 
			\\ &
			\leq \sum_{s' \in \supp(\chi^{\overline{m} \times n})}
				\Pr_{E' \leftarrow \chi^{\overline{m} \times n}} [E' = 0 ] 
				\cdot \Pr_{S' \leftarrow \chi^{\overline{m} \times n}} [S' = s'] 
			\\ &
			= \Pr_{E' \leftarrow \chi^{\overline{m} \times n}} [E' = 0 ] 
			\leq \left(\Pr_{x \leftarrow \chi} [x = 0 ]\right)^{\overline{m} \times n} 
			\enspace ,
	\end{align*}
	where we applied the law of total probability and used the fact that $\chi$ is a symmetric distribution centered at zero.
	
	We will now plug in the parameters of \FrodoPKE-\instance: For all instantiations of \instance as specified in \cite{FrodoSpec},
	$\overline{m} = 8$ and $n = i$.
	According to table 3 of \cite{FrodoSpec}, we furthermore have that
	\begin{align*}
		\Pr_{x \leftarrow \chi} [x = 0 ]
			= 2^{-16} \cdot  
				\begin{cases}
					18286  & \instance = 1344\\
					11278  & \instance = 976 \\
					9288  & \instance = 640
				\end{cases}
			< \begin{cases}
				2^{-1} & \instance = 1344\\
				2^{-2}  & \instance \in \lbrace 976, 640 \rbrace 
			\end{cases}
		\enspace .
	\end{align*} 
	
	Hence we obtain 		
	
	\begin{align*}
		\max_{c \in \FrodoCSpace} \Pr_\FrodoEnc[& \FrodoEnc(\pk, m) = c] %
		\leq  
		\begin{cases}
			2^{-8 \cdot 1344} & \instance = 1344\\
			2^{-16 \cdot \instance}  & \instance \in \lbrace 976, 640 \rbrace 
		\end{cases}
		\enspace .
	\end{align*} 
\end{proof}

The following lemma makes the bound above explicit for \HQCPKE.

\begin{restatable}[$\gamma$-Spreadness of \HQCPKE]{lemma}{SpreadnessHQC} \label{lem:Spreadness:HQC}
	\HQCPKE-\instance is $\gamma$-spread for
	\begin{align*}
		\gamma = &2 \cdot 
		\begin{cases}
			\log_2 {57600 \choose 149} > 1490   & \instance = 256 \\
			\log_2 {35840 \choose 114} > 1105  & \instance = 192 \\
			\log_2 {17664 \choose 75} > 694   & \instance = 128
		\end{cases}
		\enspace ,
	\end{align*}
	
	hence
	\begin{align*}
		q_{\RO G} \cdot 2^{65-\gamma/2}
		\leq  
			\begin{cases}
                 q_{\RO G} \cdot  2^{-1425}   & \instance = 256 \\
				 q_{\RO G} \cdot 2^{-1040}   & \instance = 192 \\
				 q_{\RO G} \cdot 2^{-629}   & \instance = 128
			\end{cases}
		\enspace .
	\end{align*}

\end{restatable}

\begin{proof}
	Let $(\pk = (h, s), \sk) \in \supp(\HQCKG)$, let $m \in \HQCMSpace$, and let $c = (u, v) \in \HQCCSpace$.
	According to the definition of \HQCEnc, we have that
	\begin{align*}
		\Pr_\HQCEnc[& \HQCEnc(\pk, m) = (u, v)] \\
		&= \Pr_{R_1, R_2 \leftarrow \mathcal{U}(S_{w_r}^{n_1 \cdot n_2}), E \leftarrow \mathcal{U}(S_{w_e}^{n_1 \cdot n_2})}
		[R_1 + h \cdot R_2 = u \ \wedge \ mG + s \cdot R_2 + E =  v] \enspace ,
	\end{align*}
	where $S_w^{n_1 \cdot n_2}$ denotes the subset of elements of hamming weight $w$ in $\bits^{n_1 \cdot n_2}$.
	
	By the law of total probability,
	\begin{align*}
		& \Pr_{R_1, R_2 \leftarrow \mathcal{U}(S_{w_r}^{n_1 \cdot n_2}), E \leftarrow \mathcal{U}(S_{w_e}^{n_1 \cdot n_2})}
		 [R_1 + h \cdot R_2 = u \ \wedge \ mG + s \cdot R_2 + E =  v] 
		\\ &
		= \sum_{r_2 \in S_{w_r}^{n_1 \cdot n_2}}
			\Pr_{R_1 \leftarrow \mathcal{U}(S_{w_r}^{n_1 \cdot n_2}), E \leftarrow \mathcal{U}(S_{w_e}^{n_1 \cdot n_2})}
					[R_1 = u - h \cdot r_2 \ \wedge \ E =  v - (mG + s \cdot r_2 ) ]
		\\ & \quad \quad \quad \quad \quad \quad 
			\cdot \Pr_{R_2 \leftarrow \mathcal{U}(S_{w_r}^{n_1 \cdot n_2})} [R_2 = r_2] 
		\\ &
		\leq \sum_{r_2 \in S_{w_r}^{n_1 \cdot n_2}} \frac{1}{ {n_1 \cdot n_2\choose w_r}} \cdot \frac{1}{ {n_1 \cdot n_2\choose w_e}} \cdot \Pr_{R_2 \leftarrow \mathcal{U}(S_{w_r}^{n_1 \cdot n_2})} [R_2 = r_2] 
		= \frac{1}{ {n_1 \cdot n_2 \choose w_r}} \cdot \frac{1}{ {n_1 \cdot n_2\choose w_e}} 
		\enspace ,
	\end{align*}
	where we used the fact that $|S_w^{N}| = {N\choose w}$ in the last line.
	
	We will now plug in the parameters of \HQCPKE-\instance: For the instantiations of \instance as specified in \cite[Section 2.7]{HQCSpec},
	we have that
	\begin{align*}
		w_e = w_r = 
		\begin{cases}
			149  & \instance = 256\\
			114  & \instance = 192 \\
			75  & \instance = 128
		\end{cases}
		\enspace ,
	\end{align*} 
	and that
	\begin{align*}
		n_1 \cdot n_2 =  
		\begin{cases}
			90 \cdot 640  & \instance = 256\\
			56 \cdot 640 & \instance = 192 \\
			46 \cdot 384 & \instance = 128
		\end{cases}
		\ = \
		\begin{cases}
			57600  & \instance = 256\\
			35840  & \instance = 192 \\
			17664  & \instance = 128
		\end{cases}
		\enspace .
	\end{align*} 
	
	Hence we obtain 		
	\begin{align*}
		\max_{c \in \HQCCSpace} \Pr_\HQCEnc[& \HQCEnc(\pk, m) = c] 
		\leq
		\begin{cases}
			(\frac{1}{ {57600 \choose 149}})^2   & \instance = 256 \\
			(\frac{1}{ {35840 \choose 114}})^2   & \instance = 192 \\
			(\frac{1}{ {17664 \choose 75}})^2   & \instance = 128
		\end{cases}
	 \enspace .
	\end{align*}
	
\end{proof}
  \fi\fi

\bibliographystyle{alpha}

\ifTightOnSpace
\else
\newcommand{\etalchar}[1]{$^{#1}$}

\fi

\ifCameraReady \else								%
	\ifSupplementaryMaterial
		\newpage
		\bigskip\noindent{\Huge\textbf{Supplementary material}}\vspace{1cm}
	\fi
	
	\appendix

\section{Overview: Relations between FO-like transformations} \label{sec:appendix:FO}
There exists a plethora of FO-like transformations, and one might wonder if a result for transformation variant X also is applicable to transformation variant Y. In order to systematize existing knowledge and to simplify such considerations,
this section recaps known relations between the security properties of FO variants on a high level.

We will now revisit other well-known variants for the FO transformation, introduced by \cite{TCC:HofHovKil17} as \FOImplicitMess, \FOImplicit and \FOExplicitBoth.
In all variants, the ${}^\explicitReject$ and ${}^\implicitReject$ stands for the way in which the KEMs reject ciphertexts that are not well-formed, i.e., ciphertexts that either fail to decrypt or whose decrypted plaintexts fail to re-encrypt:
\FOExplicitMess and \FOExplicitBoth will return a dedicated failure symbol $\explicitReject$,
\FOImplicitMess and \FOImplicitBoth will instead use an additional hash function to compute from the ciphertext a deterministic, but pseudorandom value.
Since this pseudorandom value does not communicate explicitly that the ciphertext was rejected, \FOImplicitMess and \FOImplicitBoth are often called \emph{FO with implicit rejection} (or a 'silent' KEM), and \FOExplicitMess and \FOExplicitBoth are called \emph{FO with explicit rejection}.
In both \FOExplicitMess and \FOImplicitMess, the ${}_{m}$ represents how the KEM computes its keys:
the key is computed by simply feeding the message $m$ into the key derivation oracle.
In \FOExplicitBoth and \FOImplicitBoth, the key instead is computed by including both message $m$ and ciphertext $c$ into the key derivation oracle's input.
\FOExplicitBoth and \FOImplicitBoth are hence also called \emph{ciphertext-contributing} variants.

At the time \cite{TCC:HofHovKil17} was written, all transformations above only had proofs in the classical ROM.
In order to facilitate a proof that also holds against \underline{q}uantum attackers, \cite{TCC:HofHovKil17} further modified transformations \FOExplicitMess and \FOImplicitMess
and denoted these modifications by $\mathsf{\underline{Q}FO}^\explicitReject_\hashOnlyMessage$ and $\mathsf{\underline{Q}FO}^\implicitReject_\hashOnlyMessage$, respectively.
The only difference between \FOExplicitMess/\FOImplicitMess and their $\mathsf{Q}$ counterpart is that during encapsulation, the ciphertext is concatenated with the hash value of $m$ (using a length-preserving hash function), which is then used during decapsulation to perform an additional validity check.
This additional hash value is often called \emph{key confirmation tag}, and \FOquantumExplicit and \FOquantumImplicit are often called \emph{FO with key confirmation}.
Since appending a length-preserving hash value induces communicative overhead, and since the original proofs were highly non-tight, a lot of effort has been invested into improving on both aspects.

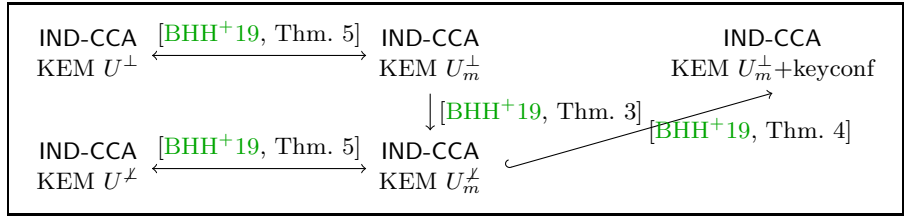
\begin{figure}[bt]
\centering
\fbox{
\begin{tikzpicture}
        \node (ind-cca-U-bot) at (0,1.5) {\begin{tabular}{c}$\INDCCA$ \\ KEM $U^{\bot}$\end{tabular}};
        \node (ind-cca-U-notbot) at (0,0) {\begin{tabular}{c}$\INDCCA$ \\ KEM $U^{\not\bot}$\end{tabular}};
        \node (ind-cca-U-m-bot) at (4.5,1.5) {\begin{tabular}{c}$\INDCCA$ \\ KEM $U_m^{\bot}$ \end{tabular}};
        \node (ind-cca-U-m-notbot) at (4.5,0) {\begin{tabular}{c}$\INDCCA$ \\ KEM $U_m^{\not\bot}$\end{tabular}};
        \node (ind-cca-U-m-bot-and-C) at (9,1.5) {\begin{tabular}{c}$\INDCCA$\\ KEM $U_m^{\bot}$+keyconf\end{tabular}};

        \draw[<->] (ind-cca-U-bot) to node[above] {\cite[Thm. 5]{TCC:BHHHP19}} (ind-cca-U-m-bot);
        \draw[<->] (ind-cca-U-notbot) to node[above] {\cite[Thm. 5]{TCC:BHHHP19}} (ind-cca-U-m-notbot);
        \draw[->] (ind-cca-U-m-bot) to node[right] {\cite[Thm. 3]{TCC:BHHHP19}} (ind-cca-U-m-notbot);
        \draw[right hook->] (5.5,0) to node[right] {\cite[Thm. 4]{TCC:BHHHP19}} (9,1);
\end{tikzpicture}
}
\caption{Relations between the security of different types of $U$-constructions as shown in \cite{TCC:BHHHP19}. The hooked arrow indicates a theorem with an $\epsilon$-injectivity constraint on the underlying deterministic scheme. Figure taken from~\cite{TCC:BHHHP19} with updated references.}
	\label{fig:summary_implications}
\end{figure}
 
Fortunately, the situation can be simplified a bit: It was proven in \cite{TCC:BHHHP19} that for either rejection variant $\FO \in \lbrace \FOImplicit, \FOExplicit \rbrace$, it does not matter which mode of key derivation is chosen,
since \FOMessGeneral is as secure as \FOBothGeneral and vice versa. (A summary of the proven relations is given in \cref{fig:summary_implications}.)
We can hence neglect this distinction and drop the subscript for the rest of this discussion.
It was furthermore shown first in \cite{EC:SaiXagYam18} that \FOImplicit is secure even against quantum attackers with bounds similar to \FOquantumImplicit, 
assuming that the underlying encryption scheme is perfectly correct.
For schemes that are not perfectly correct, established strategies can be used to generalise the result from \cite{EC:SaiXagYam18} (e.g., see \cite{PKC:HKSU20}).
To achieve security against quantum attackers, we can hence dispense with the more costly 'key confirmation variant' \FOquantumImplicit and simply use \FOImplicit.
One might wonder if a similar result could be achieved for the explicit rejection variant \FOquantumExplicit, and while an asymptotic security proof for \FOExplicit has already been established \cite{C:Zhandry19, DFMS21},
giving a proof for \FOExplicit with bounds comparable to the bounds for \FOImplicit was still an open problem until now.
While it has been proven in \cite{TCC:BHHHP19} that security of \FOExplicit implies security of \FOImplicit,
and that security of \FOImplicit implies security of \FOquantumExplicit,
it was hence not clear until now whether explicit rejection variants might not turn out to be less robust against quantum attackers than their implicit rejection counterparts.

	\ifEprint \else									%
\ifEprint %
	\subsection{Security Notions for Public-Key Encryption} \label{sec:Prels:PKE:Security}
\else %
	\section{Security Notions for Public-Key Encryption} \label{sec:Prels:PKE:Security}
\fi

We also consider all security games in the (quantum) random oracle model,
where \PKE and adversary \Ad{A} are given access to (quantum) random oracles.
(How we model quantum access is made explicit in \cref{sec:prels:compressed}.)

\ifEprint %
	\subsubsection{Definitions for \PKE}
\else %
	\subsection{Definitions for \PKE}
\fi

\begin{definition}[$\gamma$-spreadness]
	We say that \PKE is $\gamma$-spread iff for all key pairs $(\pk, \sk) \in \supp(\KG)$
	and all messages $m \in \MSpace$ it holds that
	\[ \max_{c \in \CSpace} \Pr[\Encrypt(\pk, m)= c] \leq 2^{- \gamma} \enspace , \]
	where the probability is taken over the internal randomness \Encrypt.
\end{definition}

We also recall two standard security notions for public-key encryption:
\underline{O}ne-\underline{W}ayness under
\underline{C}hosen \underline{P}laintext \underline{A}ttacks (\OWCPA)
and
\underline{Ind}istinguishability under \underline{C}hosen-\underline{P}laintext \underline{A}ttacks (\INDCPA).
\begin{definition}[\OWCPA, \INDCPA]
	Let $\PKE = (\KG,\Encrypt,\Decrypt)$ be a public-key encryption scheme with message space \MSpace.
	We define the \OWCPA game as in \cref{fig:Def:PKE:passive} and the \OWCPA \textit{advantage function of an adversary \Ad{A} against \PKE} as
	\[ \Adv^{\OWCPA}_{\PKE}(\Ad{A}) := \Pr [\OWCPA^{\Ad{A}}_\PKE \Rightarrow 1 ] \enspace .\]

	Furthermore, we define the 'left-or-right' version of \INDCPA by defining games $\INDCPA_b$, where $b\in \lbrace 0,1 \rbrace$ (also in \cref{fig:Def:PKE:passive}),
	and the \INDCPA \textit{advantage function of an adversary $\Ad{A} = (\Ad{A}_1, \Ad{A}_2)$ against \PKE}
	(where \(\Ad{A}_2\) has binary output)
	as
	\[ \Adv^{\INDCPA}_{\PKE}(\Ad{A}) := |\Pr [ \INDCPA_0^{\Ad{A}} \Rightarrow 1 ] - \Pr [ \INDCPA_1^{\Ad{A}} \Rightarrow 1 ]| \enspace . \]
	
	\begin{figure}[h!]\begin{center}\fbox{\small
		
		\nicoresetlinenr
		
		\begin{minipage}[t]{3.5cm}	
			\underline{{\bf Game} \OWCPA}
			\begin{nicodemus}
				\item $(\pk, \sk) \leftarrow \KG$
				\item $m^* \uni \MSpace$
				\item $c^* \leftarrow \Encrypt(\pk,m^*)$
				\item $m' \leftarrow \Ad{A}(\pk, c^*)$
				\item \pcreturn $\bool{m' = m^*}$
			\end{nicodemus}
		\end{minipage}
		
		\quad
		
		\begin{minipage}[t]{3.7cm}
			\underline{{\bf Game} $\INDCPA_b$}
			\begin{nicodemus}
				\item $(\pk, \sk) \leftarrow \KG$
				\item $(m^*_0, m^*_1, \state) \leftarrow \Ad{A}_1(\pk)$
				\item $c^* \leftarrow \Encrypt(\pk, m^*_b)$
				\item $b' \leftarrow \Ad{A}_2(\pk, c^*, \state)$
				\item \pcreturn $b'$
			\end{nicodemus}	
		\end{minipage}%
	}
	\end{center}
		\caption{Games \OWCPA and $\INDCPA_b$ for \PKE.}
		\label{fig:Def:PKE:passive}
	\end{figure}
	
\end{definition}

\ifEprint %
	\subsubsection{Standard notions for \KEM}
\else %
	\subsection{Standard notions for \KEM}
\fi

We now define \underline{Ind}istinguishability under \underline{C}hosen-\underline{P}laintext \underline{A}ttacks (\INDCPA)
and under \underline{C}hosen-\underline{C}iphertext \underline{A}ttacks (\INDCCA). 

\begin{definition}[\INDCPA, \INDCCA]\label{def:KEM:CPACCA}
	Let $\KEM = (\KemGen, \Encaps, \Decaps)$ be a key encapsulation mechanism with key space \KeySpace.
	For $\atk \in \lbrace \CPA, \CCA\rbrace$, we define \INDATKKEM games as in \cref{fig:Def:KEM:CPACCA}, where
	\[\mathsf{O}_\atk := \left\{
	\begin{array}{ll}
		-								&	\atk = \CPA \\
		\textnormal{\oracleDecaps}		&	\atk = \CCA
	\end{array}	\right. \enspace .	\]
	
	We define the \INDATKKEM \textit{advantage function of an adversary \Ad{A} against \KEM} as
	\[ \Adv^{\INDATKKEM}_{\KEM}(\Ad{A}) := |\Pr[\INDATKKEM^{\Ad{A}} \Rightarrow 1] - \nicefrac{1}{2}| \enspace. \]
	
	\begin{figure}[h!]\begin{center}\fbox{\small
				
		\nicoresetlinenr
				
		\begin{minipage}[t]{3.9cm}
					\underline{{\bf Game} \INDATKKEM}
					\begin{nicodemus}
						\item $(\pk,\sk) \leftarrow \KemGen$
						\item $b \uni \bits$
						\item $(K_0^*,c^*) \leftarrow \Encaps(\pk)$
						\item $K_1^* \uni \KeySpace$
						\item $b'\leftarrow \Ad{A}^{\mathsf{O}_\atk}(\pk, c^*,K_b^*)$
						\item \pcreturn $\bool{b' = b}$
					\end{nicodemus}%
				\end{minipage}%
				\;
				\begin{minipage}[t]{3.2cm}
					\underline{$\oracleDecaps(c \neq c^*)$}
					\begin{nicodemus}
						\item $K := \Decaps(\sk,c)$
						\item \pcreturn $K$
					\end{nicodemus}%
				\end{minipage}%
			}
		\end{center}
		\caption{Game \INDATKKEM for \KEM, where $\atk \in \lbrace \CPA, \CCA\rbrace$ and
			$\mathsf{O}_\atk$ is defined in \cref{def:KEM:CPACCA}.}
		\label{fig:Def:KEM:CPACCA}
	\end{figure}
	
\end{definition}
 	\fi

\ifTightOnSpace
\else %
	\section{Proof of \cref{thm:PKEpassiveToINDCPAKEM:ROM} (From $\INDCPA_{\PKE}$ or $\OWCPA_{\PKE}$ to $\INDCPA_{\FO[\PKE]}$)} \label{sec:INDCPAHHK}
	
	For easier reference, we repeat the statement of \cref{thm:PKEpassiveToINDCPAKEM:ROM}.
	
	\INDCPAHHK*
\fi

\begin{proof}
	In \cite{TCC:HofHovKil17}, the security proof for $\KemImplicitMess = \FOExplicitMess[\PKE, \RO{G},\RO H] = \TtwoExplicit[\PKEDerand, \RO{H}]$ is modularized
	into one proof for \PKEDerand and one proof for transformation \TtwoExplicitMess, as sketched in \cref{fig:FO:HHK}.
	Here, \underline{O}ne-\underline{W}ayness under \underline{V}alidity checking \underline{A}ttacks (\OWVA)
	is an intermediate helper notion that models \OWCPA security in the presence of an additional \underline{C}iphertext \underline{V}alidity \underline{O}racle $\mathsf{CVO}$ that tells the attacker whether a ciphertext is valid. To avoid confusion when looking up the theorems, Theorems 3.1 and 3.2 actually prove a security notion stronger than \OWVA security, called \OW-$\mathsf{PCVA}$ security. The \OW-$\mathsf{PCVA}$ game is like the \OWVA one except that it provides one more additional oracle to the adversary. Since \OW-$\mathsf{PCVA}$ security immediately implies \OWVA security by dismissing the additional oracle, and since Theorem 3.5 only requires \OWVA security, we omitt further details on \OW-$\mathsf{PCVA}$ security.

	\cite[Theorem 3.5]{TCC:HofHovKil17} states that \INDCCA security of $\KemImplicitMess = \TtwoExplicit[\PKEDerand, \RO{H}]$ can be based on \OWVA security of \PKEDerand, tightly.
	Clearly, the same holds when \INDCCA is replaced with \INDCPA, as one can simply set the number $q_D$ of decapsulation queries to 0. In fact, when we only need \INDCPA security, we can disregard all terms in the bound of \cite[Theorem 3.5]{TCC:HofHovKil17} that stem from how the random oracle and the decapsulation oracle were changed during the proof in order for the the decapsulation oracle to be simulatable without the secret key.
	Dismissing the respective changes, we obtain from \cite[Theorem 3.5]{TCC:HofHovKil17} that for any \INDCPA adversary \Ad{A} against \KemExplicit, issuing at most $q_\RO{G}$/$q_\RO{H}$ many queries to its respective random oracles, 
	there exist an \OWVA adversary $\tilde{\Ad{A}}$ of roughly the same running time, issuing no queries to its oracle  $\mathsf{CVO}$, such that
	\[\Adv^{\INDCPA}_{\KemImplicitMess }(\Ad{A}) \leq  \Adv^{\OWVA}_{\PKEDerand}(\tilde{\Ad{A}})\enspace .\]

	\cite[Theorem 3.1]{TCC:HofHovKil17} states that \OW-$\mathsf{PCVA}$ security of \PKEDerand can be based on \OWCPA security of \PKE, non-tightly. Since in our use case, we are only considering adversaries $\tilde{\Ad{A}}$ that do not pose any queries to oracle $\mathsf{CVO}$ or the other oracle present in the \OW-$\mathsf{PCVA}$ game, we can again disregard all terms in the bound of \cite[Theorem 3.1]{TCC:HofHovKil17} that stem from how the additional oracles got simulated during the proof.
	Dismissing the simulation of the redundant additional oracles, we obtain from \cite[Theorem 3.1]{TCC:HofHovKil17} that for any \OWVA adversary $\tilde{\Ad{A}}$ against \PKEDerand as the reduction above, 
	there exist an \OWCPA adversary $\Ad{B_{\OWCPA}}$ of roughly the same running time such that
	\[ \Adv^{\OWVA}_{\PKEDerand}(\tilde{\Ad{A}}) \leq (q_\RO{G} +  q_\RO{H} +1 )\cdot \Adv^{\OWVA}_{\PKE}(\Ad{B_{\OWCPA}}) \enspace .\]
	
	\cite[Theorem 3.2]{TCC:HofHovKil17} states that \OW-$\mathsf{PCVA}$ security of \PKEDerand can be based on \INDCPA security of \PKE, tightly. Again, we can dismiss the simulation of the redundant additional oracles and obtain from \cite[Theorem 3.2]{TCC:HofHovKil17} that for any \OWVA adversary $\tilde{\Ad{A}}$ against \PKEDerand as the reduction above,
	there exist an \INDCPA adversary $\Ad{B_{\OWCPA}}$ of roughly the same running time such that
	\[ \Adv^{\OWVA}_{\PKEDerand}(\tilde{\Ad{A}}) \leq  3 \cdot \Adv^{\INDCPA}_{\PKE}(\Ad{B_{\INDCPA}}) 
	+ \frac{2\cdot(q_\RO{G} + q_\RO{H} )+1}{|\MSpace|} \enspace .\]
	
	\begin{figure}[htb] \begin{center} \begin{tikzpicture} \small
		\node (PKEind) {\begin{minipage}{2cm}\centering \PKE \\ \INDCPA \end{minipage}};
		\node (PKEow) [below = 0.5cm of PKEind] {\begin{minipage}{2cm}\centering \PKE \\ \OWCPA \end{minipage}};
		
		\node at ($(PKEind) +(4,0)$) (dPKE) {\begin{minipage}{2cm}\centering \PKE' \\ \OWVA \end{minipage}};
		
		\node  (belowdPKE) [below = 0.5cm of dPKE] {\begin{minipage}{2cm}\centering  \phantom{X} \\ \phantom{X} \end{minipage}};
		
		\node at ($(dPKE)+(4,0)$) (KEMindCCA) {\begin{minipage}{2cm}\centering \KemExplicitMess \\ \INDCCA \end{minipage}};

		\draw[>=latex,->] (PKEind.east) -- ($(dPKE.west)$) node [draw=none, midway, above] {\PKEDerand, Th. 3.2};
		
		\draw[>=latex, dashed] (PKEow.east) -- ($(belowdPKE.west)$)	node [draw=none, midway, above] {\PKEDerand, Th. 3.1};			
		
		\draw[>=latex, dashed, ->] ($(belowdPKE.west)$) -| ($(dPKE.south)-(0,0.1)$);

		\draw[>=latex, ->] (dPKE.east) -- (KEMindCCA.west) node [draw=none, midway, above] {\TtwoExplicitMess, Th. 3.5};
				
	\end{tikzpicture} \end{center}
	\caption{
		Modular approach in \cite{TCC:HofHovKil17} for transformation \FOExplicitMess, in the ROM. Solid arrows indicate tight reductions, dashed arrows indicate non-tight reductions. The used theorem numbers are the respective theorem numbers in \cite{TCC:HofHovKil17}.
	}
	\label{fig:FO:HHK}
	\end{figure}
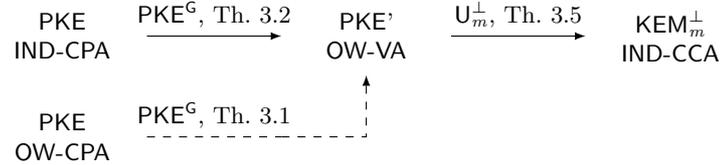

\end{proof}

\section{Proof of \cref{lem:OWTH:DiffEvents}  (\OWTH: event prob. distances if $\neg \FIND$ etc)} \label{sec:QROM:OWTH:DiffEvents}

For easier reference, we repeat the statement of \cref{lem:OWTH:DiffEvents}.

\DiffEvent*

\begin{proof}
	
	During this proof, we use $\Ad{A}^{\puncture{\SupOr}{S}}$ as a shorthand for $\Ad{A}^{\puncture{\SupOr}{S}}(\inputVar)$ and \FINDshort for \FIND.
	As argued in \cref{sec:QROM:OWTH}, \Ad{A}'s view is exactly the same in both games unless \FIND or \EVENTEXT occur,
	therefore \cref{eq:OWTH:EqualEventNoFindNoExtract} holds.
	We will first use \cref{eq:OWTH:EqualEventNoFindNoExtract} to prove \cref{eq:OWTH:DiffEventNoFind}:
	We have
	\begin{align*}
		& \left| \Pr[\EVENT  \wedge \neg \FINDshort : \Ad{A}^{\puncture{\SupOrNotProg}{S}}] - 
			\Pr[\EVENT  \wedge \neg \FINDshort : \Ad{A}^{\puncture{\SupOrProg}{S}} ] \right| \\
		& \quad  = \left|  \Pr[\EVENT  \wedge \neg \FINDshort \wedge \EVENTEXTshort: \Ad{A}^{\puncture{\SupOrNotProg}{S}}]
			- \Pr[\EVENT  \wedge \neg \FINDshort \wedge \EVENTEXTshort: \Ad{A}^{\puncture{\SupOrProg}{S}}] \right| \\
		& \quad = \left|  \Pr[\EVENT : \Ad{A}^{\puncture{\SupOrNotProg}{S}}| \neg \FINDshort \wedge \EVENTEXTshort] \cdot 
					 \Pr[\neg \FINDshort \wedge \EVENTEXTshort: \Ad{A}^{\puncture{\SupOrNotProg}{S}}] \right. \\
		& \quad \quad \left. - \Pr[\EVENT : \Ad{A}^{\puncture{\SupOrProg}{S}}|  \neg \FINDshort \wedge \EVENTEXTshort] \cdot 
						\Pr[ \neg \FINDshort \wedge \EVENTEXTshort : \Ad{A}^{\puncture{\SupOrProg}{S}}] \right| \\
		& \quad \stackrel{(*)}{=}
				\left|  \Pr[\EVENT : \Ad{A}^{\puncture{\SupOrNotProg}{S}}| \neg \FINDshort \wedge \EVENTEXTshort] 
					-  \Pr[\EVENT : \Ad{A}^{\puncture{\SupOrProg}{S}}|  \neg \FINDshort \wedge \EVENTEXTshort] \right| \\
		& \quad \quad \quad \cdot \Pr[\neg \FINDshort \wedge \EVENTEXTshort: \Ad{A}^{\puncture{\SupOrNotProg}{S}}] \\
		& \quad \leq \Pr[\neg \FINDshort \wedge \EVENTEXTshort: \Ad{A}^{\puncture{\SupOrNotProg}{S}}] 
				\leq \Pr[\EVENTEXTshort: \Ad{A}^{\puncture{\SupOrNotProg}{S}}] 
		\enspace ,
	\end{align*}
	where (*) used that if \FIND does not occur, all case-depending information is hidden from \Ad{A} until \EVENTEXT occurs,
	hence \EVENTEXT is equally likely in that case and the common factor can hence be moved to outside of the absolute value.
	
	To prove \cref{eq:OWTH:DiffFind}, it is sufficient to instead upper bound the difference between the probabilities of event $\neg \FIND$ for the two oracles: since the equation $\Pr[E] = 1 - \Pr[\neg E]$ holds for arbitrary events,
	we have that
	\begin{align*}
		|\Pr[\FINDshort : \Ad{A}^{\puncture{\SupOrNotProg}{\reproSet}}] -  \Pr[\FINDshort : \Ad{A}^{\puncture{\SupOrProg}{\reproSet}}]|
		= |\Pr[\neg \FINDshort : \Ad{A}^{\puncture{\SupOrNotProg}{\reproSet}}] -  \Pr[\neg \FINDshort : \Ad{A}^{\puncture{\SupOrProg}{\reproSet}}]| \enspace .
	\end{align*}
	The bound then follows directly from \cref{eq:OWTH:DiffEventNoFind}: We have that
	\begin{align*}
		& |\Pr[\neg \FINDshort : \Ad{A}^{\puncture{\SupOrNotProg}{S}}]
			- \Pr[\neg \FINDshort :  \Ad{A}^{\puncture{\SupOrProg}{S}}]|
		\\
		& \quad = | \Pr[\mathsf{true} \wedge \neg \FINDshort : \Ad{A}^{\puncture{\SupOrNotProg}{S}}] 
					- \Pr[\mathsf{true} \wedge \neg \FINDshort : \Ad{A}^{\puncture{\SupOrProg}{S}}] |
		\\
		& \quad \stackrel{\cref{eq:OWTH:DiffEventNoFind}}{\leq} \Pr[\EVENTEXTshort: \Ad{A}^{\puncture{\SupOrNotProg}{S}}] 
		\enspace .
	\end{align*}

\end{proof}

\section{Proof of \cref{thm:OWTH:Dist-to-Find} (\OWTH: Distinguishing to Finding)} \label{sec:QROM:OWTH:DistToFind}

For easier reference, we repeat the statement of \cref{thm:OWTH:Dist-to-Find}.

\DistToFind*

\begin{proof}

In the following helper definitions, we will again use $\Ad{A}^{\RO{O}}$ as a shorthand for $\Ad{A}^{\RO{O}}(\inputVar)$. For either  oracle $\SupOr \in \lbrace \SupOrNotProg, \SupOrProg \rbrace$,
we let 
\begin{align*} 
	p_b 										&:= \Pr[1 \leftarrow \Ad{A}^{\SupOr^b}]  \\
	p_{b, \neg \EVENTEXTshort}					&:= \Pr[b'= 1 \wedge \neg \EVENTEXT: b' \leftarrow \Ad{A}^{\SupOr^b}]\\
	p_{b, \neg \EVENTEXTshort, \neg \FINDshort}	&:= \Pr[b'= 1 \wedge \neg \FIND \wedge \neg \EVENTEXT : b' \leftarrow \Ad{A}^{\puncture{\SupOr^b}{\reproSet}}] \\
	p_{b, \neg \EVENTEXTshort, \FINDshort}		&:= \Pr[\FIND \wedge \neg \EVENTEXT : \Ad{A}^{\puncture{\SupOr^b}{\reproSet}}]
	\enspace .
\end{align*}

In order to prove \cref{thm:OWTH:Dist-to-Find}, we want to bound $\Adv^{\OWTH}_{\augQRO{f}}(\Ad{A}) = | p_0 - p_1|$.
Applying the triangle inequality yields
\begin{align}\label{eq:OWTH:DiffToDiffIfNoExtract}
	| p_0 - p_1|
	\leq & | p_0 - p_{0, \neg \EVENTEXTshort}| + | p_1 - p_{1, \neg \EVENTEXTshort}| 
			+ | p_{0, \neg \EVENTEXTshort}  - p_{1, \neg \EVENTEXTshort}|
\nonumber \\
	\stackrel{(*)}{\leq } &
		\Pr[\EVENTEXT: \Ad{A}^{\SupOr^0}] + \Pr[\EVENTEXT: \Ad{A}^{\SupOr^1}]
\nonumber \\
	& \quad  + | p_{0, \neg \EVENTEXTshort}  - p_{1, \neg \EVENTEXTshort}|
	\enspace ,	
\end{align}
where (*) used that $| p_b - p_{b, \neg \EVENTEXTshort}| = \Pr[b' = 1 \wedge \EVENTEXT: b' \leftarrow \Ad{A}^{\SupOr^b}] \leq \Pr[\EVENTEXT: \Ad{A}^{\SupOr^b}]$,
it hence remains to bound  $| p_{0, \neg \EVENTEXTshort}  - p_{1, \neg \EVENTEXTshort}|$.

We claim that for either oracle $\SupOr \in \lbrace \SupOrNotProg, \SupOrProg \rbrace$, we have that
\begin{align}\label{eq:OWTH:ClaimDistToFindWithPuncturing}
	|p_{b, \neg \EVENTEXTshort} - p_{b, \neg \EVENTEXTshort, \neg \FINDshort}|
	\leq 2 \cdot \sqrt{d \cdot \Pr[\FIND : \Ad{A}^{\puncture{\SupOr^b}{\reproSet}}]}
	\enspace .
\end{align}

Assuming that claim (\ref{eq:OWTH:ClaimDistToFindWithPuncturing}) is true,
we can then once more apply the triangle inequality to obtain
\begin{align} \label{eq:OWTH:DiffIfNoExtract}
| p_{0, \neg \EVENTEXTshort}  - & p_{1, \neg \EVENTEXTshort}|
	\leq  |p_{0, \neg \EVENTEXTshort} - p_{0, \neg \EVENTEXTshort, \neg \FINDshort}|
		+ | p_{1, \neg \EVENTEXTshort} - p_{0, \neg \EVENTEXTshort, \neg \FINDshort}|
\nonumber \\
\stackrel{(*)}{=} & 
	|p_{0, \neg \EVENTEXTshort} - p_{0, \neg \EVENTEXTshort, \neg \FINDshort}|
	+ | p_{1, \neg \EVENTEXTshort} - p_{1, \neg \EVENTEXTshort, \neg \FINDshort}|
\nonumber \\
\stackrel{(\ref{eq:OWTH:ClaimDistToFindWithPuncturing})}{\leq} & 
	2 \cdot \sqrt{d \cdot  \Pr[\FIND : \Ad{A}^{\puncture{\SupOr^0}{\reproSet}}]}
		+ 2 \cdot \sqrt{ d \cdot  \Pr[\FIND  : \Ad{A}^{\puncture{\SupOr^1}{\reproSet}}]}
\nonumber \\
\stackrel{(**)}{\leq } & 
	2 \cdot \sqrt{d \cdot ( \Pr[\FIND : \Ad{A}^{\puncture{\SupOr^1}{\reproSet}}] + \Pr[\EVENTEXT: \Ad{A}^{\SupOr^0}]) }
	+ 2 \cdot \sqrt{ d \cdot  \Pr[\FIND  : \Ad{A}^{\puncture{\SupOr^1}{\reproSet}}]}
\nonumber \\
\leq &
	4 \cdot \sqrt{ d \cdot \Pr[\FIND : \Ad{A}^{\puncture{\SupOr^1}{\reproSet}}]}
	+ 2 \cdot \sqrt{ d \cdot \Pr[\EVENTEXT: \Ad{A}^{\SupOr^0}]}
\enspace .
\end{align}
Here, (*) replaced $p_{0, \neg \EVENTEXTshort, \neg \FINDshort}$ with $p_{1, \neg \EVENTEXTshort, \neg \FINDshort}$ in the last term, using \cref{eq:OWTH:EqualEventNoFindNoExtract} from \cref{lem:OWTH:DiffEvents} which states that all events are equally likely regardless which oracle is used if neither \EVENTEXT nor \FIND occur.
(**) used \cref{eq:OWTH:DiffFind} from \cref{lem:OWTH:DiffEvents}
which states that $\Pr[\FIND : \Ad{A}^{\puncture{\SupOr^0}{\reproSet}}]$ can be upper bounded by $\Pr[\FIND : \Ad{A}^{\puncture{\SupOr^1}{\reproSet}}] + \Pr[\EVENTEXT: \Ad{A}^{\SupOr^0}]$,
and then used that the square root function is monotone increasing.

Plugging \cref{eq:OWTH:DiffIfNoExtract} into \cref{eq:OWTH:DiffToDiffIfNoExtract} yields the bound claimed in \cref{thm:OWTH:Dist-to-Find}, it hence remains to prove \cref{eq:OWTH:ClaimDistToFindWithPuncturing}, which we break down into the following steps:
Due to the deferred measurement principle, both the puncturing operation \orSemiClassical{S} and extraction oracle \SE can be rewritten such that they consist of a unitary, acting on the adversary-oracle registers and an additional measurement outcome register, and a final measurement of the outcome register at the end of the execution of \Ad{A}. We will denote the respective outcome registers by \logRegForFind (for 'finding') and \logRegForExtract (for extractions).
Second, show that it suffices to bound the distance of the states before this final measurement of \logRegForFind and \logRegForExtract.
Third, show that it suffices to bound the distance of the states for any fixed instantiation of set $S$, oracle values $(y_x)_{x \in S}$, and input string $\inputVar$.
Lastly, prove the distance bound for any fixed instantiation by considering that the two states that emerge from the same initial state;
and that the two chains of state transitions only increase the distance in terms of the probability that \FIND occurs.

To flesh out this summary, we will first write \orSemiClassical{S} as a concrete combination of unitaries and measurements:
A ´logging' register \logRegForFind holding bitstrings of length $d$ is initialised in state $\ket{0 \cdots 0}$.
Intuitively, \logRegForFind will log at its $i$-th position if the $i$-th query triggered \FIND:
Whenever \Ad{A} performs an oracle query, say it is the $i$-th, we slot in a unitary $\unitaryForFind{i}$ that marks in the $i$-th position of \logRegForFind whether the query register holds an element of $S$.
More formally,  $\unitaryForFind{i}$ acts on query register $X=X_1 \cdots X_w$ (recall that \Ad{A} can issue parallel oracle queries) and logging register \logRegForFind by 
\begin{equation*}
	\unitaryForFind{i} \ket{x_1, \cdots, x_w}_X \ket{b_1, \cdots, b_d}_\logRegForFind:=
		\begin{cases}
			\ket{x_1, \cdots, x_w}_X \ket{b_1, \cdots, b_d}_\logRegForFind 						& x_j \notin S \,\forall j \\
			\ket{x_1, \cdots, x_w}_X \ket{\textnormal{flip}_i(b_1, \cdots, b_d)}_\logRegForFind	& \exists j: x_j \in S
		\end{cases} \enspace .		
\end{equation*}
Processing oracle queries according to \puncture{\SupOr}{\reproSet} consists of first applying $\unitaryForFind{i}$ to $X$ and \logRegForFind,
then measuring \logRegForFind in the computational basis,
and then applying the oracle unitary \OrUnit (see \cref{sec:compressed-prels} for a brief description how parallel queries are answered). %

Next, we also write \SE for function $f: X \times Y \to \{0,1\}^\ell$ as a concrete combination of unitaries and measurements:
Let $A$ be the register that holds the state of \Ad{A}, and let $D$ be the oracle database register. Note that $A$ contains a register \logRegForExtract that accommodates \qExtract many elements of $X$,
which is used to log the outcome of the $i$-th query to \SE at its $i$-th position.
Whenever \Ad{A} performs a query to \SE, say it is the $i$-th, we apply a unitary \unitaryForExtract{i} that adds to the $i$-th position of \logRegForExtract the extraction outcome.
More formally, \unitaryForExtract{i} acts on query register $T$, database register $D$ and register \logRegForExtract by 
\begin{equation*}
	\unitaryForExtract{i} \ket{t}_T 
		:= \ket{t}_T \otimes \sum_{x \in X \cup \lbrace \bot \rbrace} \varSigma_{x, t} \otimes \unitaryForExtractX{i}{x}
	\enspace ,		
\end{equation*}
where $\varSigma_{x, t}$ acts on $D$ and is defined by
\begin{equation*}
	\varSigma_{x, t} := \begin{cases}
		\bigotimes_{x' < x } \left( \sum_{y \in Y: f(x,y) \neq t} \proj{y}_{D_{x'}} \right)
			\bigotimes \left( \sum_{y \in Y: f(x,y) = t } \proj{y}_{D_{x}} \right) 
		& x \in X
		\\
		\id - \sum_{x \in X} \varSigma_{x, t} & x = \bot
	\end{cases}
	\enspace ,		
\end{equation*}
and $\unitaryForExtractX{i}{x}$ acts on \logRegForExtract by 
\begin{equation*}
	\unitaryForExtractX{i}{x} \ket{x_1, \cdots, x_\qExtract}_\logRegForExtract
	:= \ket{x_1, \cdots, x_{i-1}, x_i + x, x_{i+1}, \cdots, x_\qExtract}_\logRegForExtract
	\enspace .		
\end{equation*}
Processing the $i$-th extraction query according to \SE consists of first applying \unitaryForExtract{i} to $T$, $D$ and \logRegForExtract, and then measuring \logRegForExtract in the computational basis.

We can now lift the final joint adversary-oracle state \FinalStateAD of \Ad{A}, when run with access to original oracle \SupOrNotProg,
to the joint adversary-oracle-log state $\FinalStateNotPuncWithExtract := \FinalStateAD \otimes \proj{0\cdots0}_\logRegForFind$.
(Note that \logRegForFind is initialised to and will maintain to be in state $\ket{0 \cdots 0}$.)
We will furthermore denote by \FinalStatePuncWithExtract the joint adversary-oracle-log state when \Ad{A} is run with access to \puncture{\SupOrProg}{\reproSet}.
This means that \FinalStateAD is the final state of \Ad{A} \emph{without} puncturing, and \FinalStatePuncWithExtract is the final state of \Ad{A} \emph{with} puncturing.
Let $M$ be the measurement that measures, given the registers $A$, $D$, \logRegForFind,
whether \Ad{A} outputs 1, \EVENTEXT did not occur, and $\logRegForFind$ is equal to $\ket{0\cdots0}$, the latter meaning that $\FIND$ did not happen.
Let $P_M(\Phi)$ denote the probability that $M$ returns 1 when measuring a state $\Phi$.
As our arguments will work for both oracle cases, we will simply write $p_{\neg \EVENTEXTshort}$ instead of $p_{b, \neg \EVENTEXTshort}$ and $p_{\EVENTEXTshort, \neg \FINDshort}$ instead of $p_{b, \neg \EVENTEXTshort, \neg \FINDshort}$.
We have that $p_{\neg \EVENTEXTshort} = P_M(\FinalStateNotPuncWithExtract)$ and that $p_{\neg \EVENTEXTshort, \neg \FINDshort} = P_M(\FinalStatePuncWithExtract)$,
hence we want to upper bound 
\begin{align}\label{eq:ProbToProbMeasure}
	| p_{\neg \EVENTEXTshort} - p_{\neg \EVENTEXTshort, \neg \FINDshort}| = | P_M(\FinalStateNotPuncWithExtract) - P_M(\FinalStatePuncWithExtract)|\enspace ,
\end{align}
and due to \cite[Lemma 4]{C:AmbHamUnr19}, we know that
\begin{align} \label{ProbMeasure:ProbToBures}
	| P_M(\FinalStateNotPuncWithExtract) - P_M(\FinalStatePuncWithExtract)| \leq  B(\FinalStateNotPuncWithExtract, \FinalStatePuncWithExtract) \enspace ,
\end{align}
where $B$ is the Bures distance. I.e., for two density operators $\tau_1$ and $\tau_2$,
$B(\tau_1,\tau_2) \coloneqq \sqrt{2-2F(\tau_1,\tau_2)}$, and the fidelity $F$ is defined by
$F(\tau_1,\tau_2) \coloneqq \Tr \sqrt{\sqrt{\tau_1} \tau_2 \sqrt{\tau_1}} $.
According to the definition of the Bures distance,
\begin{align*}
	B(\FinalStateNotPuncWithExtract, \FinalStatePuncWithExtract)^2 = 2(1-F(\FinalStateNotPuncWithExtract, \FinalStatePuncWithExtract))\enspace .
\end{align*}

Combining \cref{ProbMeasure:ProbToBures}) with \cref{ProbMeasure:ProbToBures} and plugging in the definition of the Bures distance hence yields
\begin{align*}
	| p_{\neg \EVENTEXTshort} - p_{\neg \EVENTEXTshort, \neg \FINDshort}|
		\leq   \sqrt{2(1-F(\FinalStateNotPuncWithExtract, \FinalStatePuncWithExtract))}  \enspace .
\end{align*}

To show that $|p_{b, \neg \EVENTEXTshort} - p_{b, \neg \EVENTEXTshort, \neg \FINDshort}| \leq 2 \cdot \sqrt{d \cdot \Pr[\FIND : \Ad{A}^{\puncture{\SupOr^b}{\reproSet}}]}$, it hence suffices to prove that 
\begin{equation}\label{eq:OWTH:Claim-Fidelity}
	F(\FinalStateNotPuncWithExtract, \FinalStatePuncWithExtract) \geq 1 - 2d \cdot \Pr[\FIND: \Ad{A}^{\puncture{\SupOr}{\reproSet}}(\inputVar)]
	\enspace .
\end{equation}
To lower bound $F(\FinalStateNotPuncWithExtract, \FinalStatePuncWithExtract)$, we make the following observation:
The measurements performed by \orSemiClassical{S} and \SE can be delayed, i.e.,
when processing an oracle query, we apply the respective unitary \unitaryForFind{i}, but do not perform the measurement of \logRegForFind.
Similarly, when performing extraction queries, we apply the respective unitary \unitaryForExtract{i}, but do not perform the measurement of \logRegForFind.
Instead, we perform a measurement of \logRegForFind and \logRegForExtract in the end, which we will denote by \MeasureLogs.
Let \FinalStateNotPuncWithoutExtract denote the final state of \Ad{A}, when run with access to original oracle \SupOrNotProg,
but without the extraction measurements.
Since the 'FIND' register \logRegForFind of \FinalStateNotPuncWithExtract will never be touched as \FinalStateNotPuncWithExtract represents the case where no puncturing is performed, \FinalStateNotPuncWithExtract is stable under 'FIND' measurements,
we hence have that  $\FinalStateNotPuncWithExtract = \MeasureLogs(\FinalStateNotPuncWithoutExtract)$.
Let \FinalStatePuncWithoutExtract denote the final state of \Ad{A} when run with access to $\puncture{\SupOr}{\reproSet}$,
but without the final measurement \MeasureLogs,
meaning $\FinalStatePuncWithExtract = \MeasureLogs(\FinalStatePuncWithoutExtract)$.
Using monotonicity of fidelity, we obtain
\begin{equation*}
	F(\FinalStateNotPuncWithExtract, \FinalStatePuncWithExtract)
	= F(\MeasureLogs(\FinalStateNotPuncWithoutExtract), \MeasureLogs(\FinalStatePuncWithoutExtract))
	\geq F(\FinalStateNotPuncWithoutExtract, \FinalStatePuncWithoutExtract)
	\enspace .
\end{equation*}

We will now break down the fidelity term $F(\FinalStateNotPuncWithoutExtract, \FinalStatePuncWithoutExtract)$ into an expected value for instances of set $S$,
oracle values $\overrightarrow{y} := (y_x)_{x \in S}$ and input $\inputVar$:
Let \GenInstance denote the sampling of an instance $\instanceOWTH := (S, (y_x)_{x \in S}, \inputVar)$ according to their distribution.
Let $\ket{\FinalStatePureNotPuncIndexed}$ be the pure state corresponding to $\FinalStateNotPuncWithoutExtract$ that would be obtained by running \Ad{A} with a fixed instance $\instanceOWTH$.
Similarly, let $\ket{\FinalStatePureKeepTrackIndexed}$ be the state corresponding to $\FinalStatePuncWithoutExtract$.
Then $\FinalStateNotPuncWithoutExtract = \Exp_{\instanceOWTH} [\proj{\ket{\FinalStatePureNotPuncIndexed}}]$,
and $\FinalStatePuncWithoutExtract = \Exp_{\instanceOWTH} [\proj{\ket{\FinalStatePureKeepTrackIndexed}}]$.
Hence we can identify
\begin{align*}
	F(\FinalStateNotPuncWithoutExtract, \FinalStatePuncWithoutExtract) &= F(	\Exp_{\instanceOWTH} \proj{\FinalStatePureNotPuncIndexed}, 
									\Exp_{\instanceOWTH} \proj{\FinalStatePureKeepTrackIndexed}) \nonumber \\
	& \stackrel{(*)}{\geq} \Exp_{\instanceOWTH} F( \proj{\FinalStatePureNotPuncIndexed},
																	  \proj{\FinalStatePureKeepTrackIndexed})  \nonumber 
	\stackrel{(**)}{\geq}  1 - \frac{1}{2} \Exp_{\instanceOWTH}
												\left\lVert \ket{\FinalStatePureNotPuncIndexed}
															- \ket{\FinalStatePureKeepTrackIndexed} \right \rVert^2
	\enspace ,
\end{align*}
Here, (*) follows from the joint concavity of the fidelity, and (**) uses the fact that for any two normalised states $\ket{\Psi}$ and $\ket{\Phi}$,
we have that $F( \proj{\Psi}, \proj{\Phi} ) \geq 1 - \frac{1}{2} \left\lVert \ket{\Psi} - \ket{\Phi} \right \rVert^2$.
(This was proven in \cite[Lemma 3]{C:AmbHamUnr19}).

In order to prove \cref{eq:OWTH:Claim-Fidelity}, it hence remains to show that for any instantiation $\instanceOWTH = (S, \overrightarrow{y}, \inputVar)$,
it holds that
\begin{equation}\label{eq:OWTH:Claim-Norm}
	\left\lVert \ket{\FinalStatePureNotPuncIndexed} - \ket{\FinalStatePureKeepTrackIndexed} \right \rVert^2
	 \leq 4 d \cdot P_\FIND^{\instanceOWTH}
 \enspace ,
\end{equation}
where $P_\FIND^{\instanceOWTH}$ denotes the probability of measuring the $\logRegForFind$ register of the final state $\ket{\FinalStatePureKeepTrackIndexed}$ resulting in anything else than $\ket{0, \cdots, 0}$.
For the rest of the proof, we hence consider $\instanceOWTH = (S, \overrightarrow{y}, \inputVar)$ to be fixed and omit the indices from our notation.

Both final states $\ket{\FinalStatePureNotPunc}$ and $\ket{\FinalStatePureKeepTrack}$ result from a chain of state transitions,
applied to the same initial state \InitialState.
Out of these transitions, $d$ many represent oracle queries and hence are either of the form $T_0^{(j)} = U_A \circ \OrUnit$
(to end up with $\ket{\FinalStatePureNotPunc}$),
where $U_A$ models the adversary's behaviour,
or of the form $T_1^{(j)} = U_A \circ \OrUnit \circ \unitaryForFind{i_j}$ for some $i_j$ (to end up with $\ket{\FinalStatePureKeepTrack}$).
The remaining \qExtract many transitions represent extraction queries and are of the form $T_0^{(j)} = T_1^{(j)} = U_A \circ \unitaryForExtract{i_j}$ for some $i_j$, where $U_A$ again models the adversary's behaviour.
Let \InbetweenStateNotPunc{j} denote the $j$-th intermediate state on the way to final state $\ket{\FinalStatePureNotPunc}$,
i.e., let $\InbetweenStateNotPunc{j} = T_0^{(j)} \circ T_0^{(j-1)} \circ \cdots \circ T_0^{(0)} \InitialState$,
and let \InbetweenStateKeepTrack{j} denote the $j$-th intermediate state on the way to final state $\ket{\FinalStatePureKeepTrack}$,
i.e., let $\InbetweenStateKeepTrack{j} = T_1^{(j)} \circ T_1^{(j-1)} \circ \cdots \circ T_1^{(1)}  \InitialState$.
Furthermore, let $\epsilon_j$ denote the distance between these intermediate states,
i.e., $\epsilon_j := \left\lVert \InbetweenStateNotPunc{j} - \InbetweenStateKeepTrack{j} \right \rVert$.
With this notation, we have that 
\begin{align*}
	\left\lVert \ket{\FinalStatePureNotPunc} - \ket{\FinalStatePureKeepTrack} \right \rVert
	= \epsilon_{d  + \qExtract} = \sum_{j=1}^{d  + \qExtract} \epsilon_{j} - \epsilon_{j-1} \enspace ,
\end{align*}
hence
\begin{align}\label{eq:OWTH:Distance-To-Sum}
	\left\lVert \ket{\FinalStatePureNotPunc} - \ket{\FinalStatePureKeepTrack} \right \rVert^2
	& \leq \left( \sum_{j=1}^{d  + \qExtract} | \epsilon_{j} - \epsilon_{j-1} | \right)^2 
	\enspace .
\end{align}

We will now bound this sum by bounding the summands $| \epsilon_{j} - \epsilon_{j-1} |$, depending on which kind of query they represent.
To this end, let $Q_{\textnormal{Or}}\subset\lbrace 1, \cdots, d + \qExtract \rbrace$ be the index set of oracle queries,
i.e., the set of indices $j$ such that $T_b^{(j)} = U_A \circ \OrUnit \circ (\unitaryForFind{i_j})^b$ for some $i_j$,
and let $Q_{\textnormal{Ext}} \subset \lbrace 1, \cdots, d + \qExtract \rbrace$ be the index set of extraction queries,
i.e., the set of indices $j$ such that $T_b^{(j)} = U_A \circ \unitaryForExtract{i_j}$ for some $i_j$.
We claim that for any oracle query, i.e., for any $j \in Q_{\textnormal{Or}}$, we have that
\begin{equation}\label{eq:OWTH:Claim-Intermediate-Oracle}
	| \epsilon_{j} - \epsilon_{j-1} | \leq 2 \cdot \left\lVert P_S \InbetweenStateKeepTrack{j-1} \right \rVert
	\enspace ,
\end{equation}
where $P_S$ is the projector unto the subspace spanned by $S$.
For any extraction query, i.e., for any $j \in Q_{\textnormal{Ext}}$, we furthermore claim that
\begin{equation}\label{eq:OWTH:Claim-Intermediate-Extraction}
	\epsilon_{j} = \epsilon_{j-1} \enspace .
\end{equation}

Plugging claims (\ref{eq:OWTH:Claim-Intermediate-Oracle}) and (\ref{eq:OWTH:Claim-Intermediate-Extraction}) into \cref{eq:OWTH:Distance-To-Sum}, we obtain
\begin{align*}
	\left\lVert \ket{\FinalStatePureNotPunc} - \ket{\FinalStatePureKeepTrack} \right \rVert^2
	& 	\leq \left( \sum_{j=1}^{d  + \qExtract} | \epsilon_{j} - \epsilon_{j-1} | \right)^2 
			\leq \left( \sum_{j \in Q_{\textnormal{Or}}} 2 \cdot  \left\lVert P_S \InbetweenStateKeepTrack{j-1} \right \rVert \right)^2 \\
	&  \stackrel{(*)}{\leq} 4 d  \cdot \left( \sum_{j \in Q_{\textnormal{Or}}}\left\lVert P_S \InbetweenStateKeepTrack{j-1} \right \rVert^2 \right)
		\stackrel{(**)}{\leq} 4 d  \cdot (P_\FIND)
	\enspace ,
\end{align*}
where (*) used Jensen's inequality;
and (**) used that $P_S$ is precisely the measurement operator corresponding to the event \FIND.

It hence remains to prove claims (\ref{eq:OWTH:Claim-Intermediate-Oracle}) and (\ref{eq:OWTH:Claim-Intermediate-Extraction}).
In order to prove claim (\ref{eq:OWTH:Claim-Intermediate-Oracle}), note that for $j \in Q_{\textnormal{Or}}$, we have that
\begin{align*}
	\epsilon_{j} &
		= \left\lVert \InbetweenStateNotPunc{j} - \InbetweenStateKeepTrack{j} \right \rVert 
		= \left\lVert U_A \circ \OrUnit \InbetweenStateNotPunc{j-1} - U_A \circ \OrUnit \circ \unitaryForFind{i_j} \InbetweenStateKeepTrack{j-1} \right\rVert
	\\ &
		= \left\lVert U_A \circ \OrUnit \left( \InbetweenStateNotPunc{j-1} - \unitaryForFind{i_j} \InbetweenStateKeepTrack{j-1} \right) \right \rVert
		\stackrel{(*)}{=} \left\lVert \InbetweenStateNotPunc{j-1} - \unitaryForFind{i_j} \InbetweenStateKeepTrack{j-1} \right \rVert 
	\\ &
		\leq \left\lVert \InbetweenStateNotPunc{j-1} - \InbetweenStateKeepTrack{j-1} \right \rVert
						+ \left\lVert \InbetweenStateKeepTrack{j-1} - \unitaryForFind{i_j} \InbetweenStateKeepTrack{j-1} \right \rVert 
		= \epsilon_{j-1} + \left\lVert (\id - \unitaryForFind{i_j}) \InbetweenStateKeepTrack{j-1} \right \rVert
	\enspace ,
\end{align*}
where (*) used that $U_A$ and \OrUnit are unitaries. Using that $\id$ and $\unitaryForFind{i_j}$ coincide on the image of $(\id - P_S)$, we can identify
\begin{align*}
	\left\lVert (\id - \unitaryForFind{i_j}) \InbetweenStateKeepTrack{j-1} \right \rVert 
	& = \left\lVert (\id - \unitaryForFind{i_j}) P_S \InbetweenStateKeepTrack{j-1} \right \rVert
	\leq \left\lVert (\id - \unitaryForFind{i_j})\right \rVert _\infty\left\lVert P_S \InbetweenStateKeepTrack{j-1} \right \rVert
	\leq 2 \cdot \left\lVert P_S \InbetweenStateKeepTrack{j-1} \right \rVert
	\enspace ,
\end{align*}
where the second-to-last inequality holds by definition of the operator norm, and the last follows from the triangle inequality.

In order to prove claim  (\ref{eq:OWTH:Claim-Intermediate-Extraction}), note that
\begin{align*}
	\epsilon_{j} &
	= \left\lVert \InbetweenStateNotPunc{j} - \InbetweenStateKeepTrack{j} \right \rVert 
	= \left\lVert U_A \circ \unitaryForExtract{i_j} \InbetweenStateNotPunc{j-1} - U_A \circ \unitaryForExtract{i_j} \InbetweenStateKeepTrack{j-1} \right \rVert 
	\stackrel{(*)}{=} \left\lVert \InbetweenStateNotPunc{j-1} - \InbetweenStateKeepTrack{j-1} \right \rVert 
	\enspace ,
\end{align*}
where (*) used that $U_A$ and \unitaryForExtract{i} are unitaries.

\qed
\end{proof}

\begin{proof}
	By definition of the \FFPCPA advantage, we have
	\begin{align*}
		\Adv^{\FFPCPA}_{\PKEDerand}(\Ad{A}) =\Pr_{m\leftarrow \Ad A^{\SupOr}(\pk)}[(m, \SRO(m))\text{ fails wrt. }(\sk, \pk)]
		\enspace .
	\end{align*}	
	To upper bound this probability, we begin by defining \FngFPCPA adversary \Ad{B}:
	On input $\pk$, \Ad{B} runs $\Ad{A}(\pk)$, simulating \SupOr to \Ad{A}. %
	When \Ad{A} finishes by outputting its message $m$, \Ad{B} computes $r \coloneqq \SRO(m)$,
	uses its failure-checking oracle to compute $b' \coloneqq \FCO_b(m,r)$ and outputs $b'$.
	In the case where the challenge bit $b$ of \Ad{B}'s \FngFPCPA game is 0, \Ad{B} perfectly simulates the \FFPCPA game to \Ad{A} and wins iff \Ad{A} wins in game \FFPCPA.
	Therefore,
	\begin{align*}
		\Pr_{m\leftarrow \Ad A^{\SupOr}(\pk)} & [(m, \SRO(m))\text{ fails wrt. }(\sk, \pk)]
			=\Pr[1\leftarrow\Ad B(\pk)|b=0]\\
		&\le \Pr[1\leftarrow\Ad B(\pk)|b=1] + 2\Adv^{\FngFPCPA}_{\PKE}(\Ad{B})
		\enspace ,
	\end{align*}
	where the last line used the definition of the \FngFPCPA advantage.
	
	To upper bound $\Pr[1\leftarrow\Ad B(\pk)|b=1]$, note that this probability formalizes \Ad{A} outputting a message that fails to decrypt, but under an independently drawn key pair $(\sk', \pk')$:
	\begin{align} \label{eq:FFP-no-CPA}
		\Pr[1\leftarrow\Ad B(\pk)|b=1]
		&=\Pr_{m\leftarrow \Ad A^{\SupOr}(\pk)}[(m, \SRO(m))\text{ fails wrt. }(\sk', \pk')]\ifTightOnSpace \else	\enspace \fi,
	\end{align}
	where the probability is taken additionally over $(\sk', \pk')\leftarrow \KG$.
	
	To upper bound this probability, we define \FFPNoKey adversary $\Ad{C}^\SupOr$ against \PKEDerand:
	Upon initialisation, \Ad{C} computes a key pair $(\pk, \sk)$ on its own and runs $\Ad{A}^\SupOr(\pk)$.
	When \Ad{A} finishes by outputting its message $m$, \Ad{C} forwards the message to its own game.
	Since \Ad{C} perfectly simulates the game in \cref{eq:FFP-no-CPA} to \Ad{A} and wins iff \Ad{A} wins,
	\[
		\Pr_{m\leftarrow \Ad A^{\SupOr}(\pk)}[(m, \SRO(m))\text{ fails wrt. }(\sk', \pk')]
		= \Adv^{\FFPNoKey }_{\PKEDerand}(\Ad{C}) \enspace .
	\]
\vspace{-1cm}\\
\qed
\end{proof} 	

\fi
\end{document}